\documentclass{article}
\usepackage[preprint]{src/neurips_2026}

\usepackage[utf8]{inputenc} 
\usepackage[T1]{fontenc}    
\usepackage{hyperref}       
\usepackage{url}            
\usepackage{booktabs}       
\usepackage{amsfonts}       
\usepackage{nicefrac}       
\usepackage{microtype}      
\usepackage{xcolor}         
\usepackage{todonotes}

\usepackage[table]{xcolor}
\usepackage{enumitem}
\usepackage{graphicx}
\usepackage{amsmath}
\usepackage{multirow}
\usepackage{pifont}
\usepackage{subcaption}
\usepackage{tikz}
\usepackage{amsthm}
\usepackage{tabularx}  
\usepackage{threeparttable}
\usepackage{titletoc}
\usepackage{caption}

\usepackage{pgfplots}
\usepgfplotslibrary{groupplots}
\usepgfplotslibrary{fillbetween}
\usetikzlibrary{intersections}
\pgfplotsset{compat=1.18}

\theoremstyle{plain}
\newtheorem{theorem}{Theorem}[section]
\newtheorem{proposition}[theorem]{Proposition}

\newtheorem{corollary}[theorem]{Corollary}
\theoremstyle{definition}
\newtheorem{definition}[theorem]{Definition}

\theoremstyle{remark}
\newtheorem{remark}[theorem]{Remark}

\newcommand{\tc}[2]{\textcolor{#1}{#2}}
\newcommand{\cmark}{\textcolor{green!70!black}{\ding{51}}} 
\newcommand{\xmark}{\textcolor{red!90!black}{\ding{55}}}   

\title{Higher-Order Fourier Neural Operator: Explicit Mode Mixer for Nonlinear PDEs}
\author{%
Alex Colagrande$^1$\qquad
  Paul Caillon$^1$\qquad
Eva Feillet$^2$\qquad
Alexandre Allauzen$^{1,3}$\\
$^1$Miles Team, LAMSADE, Université Paris Dauphine-PSL, Paris, France \\
$^2$Université Paris-Saclay, CNRS, LISN, 91400, Orsay, France \\
$^3$ESPCI PSL, Paris, France \\
\texttt{alex.colagrande@dauphine.psl.eu}
}

\begin{document}

\maketitle

\begin{abstract}
    Neural operators provide deep neural networks for learning mappings between function spaces. Among them, the Fourier Neural Operator (FNO) is particularly effective: its spectral convolution relies on low-dimensional Fourier-domain representations and can handle inputs at different resolutions. This design aligns well with settings where the Fourier basis diagonalizes the underlying operator, such as linear, constant-coefficient PDEs on periodic domains, in which Fourier modes evolve independently. However, nonlinear PDEs may benefit from an additional inductive bias, as they exhibit structured interactions between modes, governed by polynomial nonlinearities. To capture this inductive bias, we introduce the \textbf{Higher-Order Spectral Convolution}, a spectral mixer that extends FNO from diagonal modulation to explicit $n$-linear mode mixing, aligned with the dynamics of nonlinear PDEs. Our experiments on standard benchmarks show that the proposed Higher-Order FNO (HO-FNO) retains the efficiency of FNO-based architectures and consistently improves over other spectral neural operators. 
    HO-FNO also performs on par with or better than state-of-the-art transformers and state-space models on several datasets, with stronger gains in highly nonlinear regimes, such as the Poisson equation with polynomial forcing, where a single HO-FNO layer outperforms FNO models with up to $16$ layers. We open-source our code for reproducibility at: \url{https://github.com/AlexColagrande/HO-FNO}.
\end{abstract}

\section{Introduction}

Partial differential equations (PDEs) are central to modeling physical and engineering systems, including fluid dynamics, porous-media transport, solid mechanics, and atmospheric flows~\citep{staniforth2022global}. Since closed-form solutions are rarely available, these problems typically require numerical approximation.
\newline
Over the past decades, traditional numerical, physically grounded methods such as the finite element method~\citep{FEM_Johnson}, the finite volume method~\citep{FVM_LeVeque}, and the finite difference method~\citep{FDM_LeVeque} have achieved both accuracy and interpretability. 
Despite their strengths, these methods face two main limitations: very high computational cost due to fine spatiotemporal discretization and reliance on complete knowledge of the governing PDEs. 
Recently, motivated by the remarkable achievements of deep learning for modeling complex functions, numerous data-driven PDE solvers have been introduced to overcome the limitations of traditional numerical methods. Among these approaches, the operator learning framework~\citep {kovachki_survey,Survey_NO_new} stands out as the most physically grounded. 
Neural operators, in particular, aim to approximate the underlying solution operator that maps input functions, such as coefficients, forcing terms, or initial conditions, to output solutions, thereby providing an efficient alternative to classical discretization-based schemes. 
The Fourier Neural Operator (FNO)~\citep{FNO} is inspired by spectral methods, which provide the highest spatial accuracy and exponential convergence on regular grids. 
\newline
The FNO and its variants explicitly model the input function in a spectral basis via the standard or a generalized Fourier transform, with or without encoder–decoder components. Such models typically combine linear layers with nonlinear activations. These linear layers are usually implemented as global convolutions over a truncated set of modes, and they act independently on each Fourier coefficient, i.e., without explicit mode mixing. \textit{In this work, to increase the modeling capacity of spectral layers in neural operators, we propose an $n$-order spectral convolution that implements an $n$-linear global mixing of Fourier coefficients while retaining the computational efficiency of a Fourier truncation.}
Our contributions are the following:
\begin{enumerate}[nosep, topsep=-4pt, ]
    \item \textbf{Higher-Order Fourier Neural Operator.} We design the first spectral neural operators modeling explicitly the mode interaction of non-linear PDEs.
    \item \textbf{Interaction on different geometries.} We showcase the effect of modeling order $2$ interactions on spherical data by combining our method with Spherical Harmonic convolutions.
    \item \textbf{Experiments and ablation studies.} Through extensive experiments, we show the advantages of the proposed method in standard benchmarks, including regular grids, structured meshes, and point clouds. We also show that the model provides additional benefits in highly non-linear settings, for which it is designed, and that higher-order interactions maintain FNO's ability to generalize across resolution. 
\end{enumerate}

\begin{figure*} 
\centering
\captionsetup{font=small}
\includegraphics[width=0.9\linewidth]{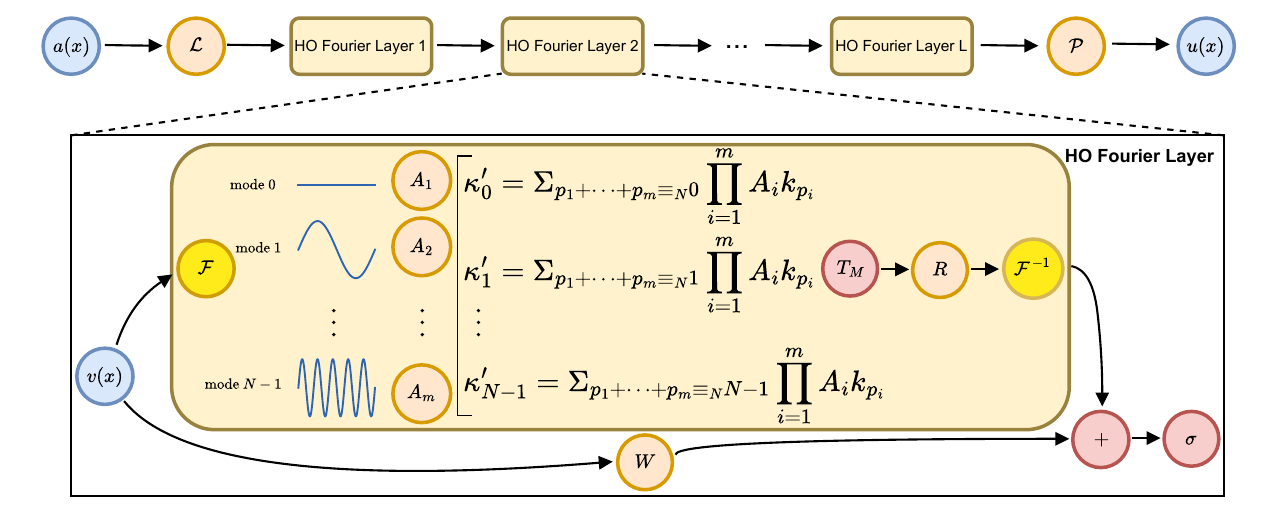}
    \caption{\textbf{Overview of the proposed HO-FNO architecture}, adapted from~\citep{FNO}. 
    \textbf{Top:} the input $a$ is lifted by $\mathcal{L}$, processed by $L$ HO-FNO layers, and projected by $\mathcal{P}$ to obtain the output $u$. 
    \textbf{Bottom:} each HO-Fourier layer transforms an intermediate representation $v$ to Fourier space, mixes the $N$ Fourier modes into higher-order pseudo-modes, keeps the lowest $M$ pseudo-modes, applies a learned linear transform $R$, and maps the result back with an inverse Fourier transform, $\mathcal{F}^{-1}$. 
    The output is combined with a local linear skip connection $W$, followed by a nonlinearity $\sigma$.
    }
    \label{fig:summary_figure}
\end{figure*}

\section{Related work}

\paragraph{Spectral neural operators} Among neural operators, FNO stands out for modeling dynamical systems on equally spaced meshes and for its ability to transfer across resolutions without retraining, a consequence of its explicit representation in the Fourier basis, which remains consistent under mesh refinement. 
For complex geometries, several variants of FNO have been introduced by changing the spectral basis, for instance, the  Spherical Fourier Neural Operator (SFNO) ~\citep{SFNO} on the sphere, and NORM ~\citep{NORM} on general Riemannian manifolds. 
Furthermore, extensions to irregular meshes have been proposed by mapping them onto regular grids, either via a non-uniform Fourier transform (GNO~\citep{GNO}), a learnable diffeomorphism implemented via message-passing (GEO-FNO~\citep{GEO-FNO}), or an optimal transport map (OTNO~\citep{OTNO}). 
In a complementary direction, DSFNO~\cite{DSFNO} replaces FNO’s static spectral kernel with hypernetwork-generated dynamic convolutions from truncated activations, but yields largely unstructured kernels. \emph{We build on this idea of using an input-dependent convolution by explicitly structuring the kernel to reflect the interaction patterns of solution operators for polynomial nonlinear PDEs.}

In the following, we refer to FNO and its variants as \emph{Spectral Neural Operators} (SNOs).

The spectral convolution of a Fourier-based neural operator closely mimics the action of the Green function, a kernel whose convolution yields the solution of linear PDEs with constant coefficients on periodic domains~\citep{stakgold2011green}. 
However, for nonlinear PDEs, the Green function no longer provides a useful representation, whereas the composition of linear spectral convolutions with nonlinear activations endows SNOs with universal approximation capabilities~\citep{FNO_theory}. 

\paragraph{Attention-based models}
Transformer models have become ubiquitous not only in language and vision but also in physical modeling~\citep{UPT, MANO}. 
 
A key factor behind the success of transformers is the attention's ability to capture pairwise interactions within the input sequence. 
This mechanism has recently been generalized to model interactions among an arbitrary number $n$ of entities, opening a line of work on higher-order attention~\citep{2-simplex_attention_logic}.
Despite their $O(N^n)$ complexity in the sequence length $N$, these higher-order variants of the transformer architecture show better scaling laws~\citep{2-simplex_attention_meta} and exponentially improved depth efficiency ~\citep{higher_order_attention_theory} compared to standard attention-based models.
\emph{Following this line of work, we introduce a new framework that realizes $n$-order interactions between coefficients directly in the Fourier domain, providing the spectral analog of higher-order attention, which operates in the Dirac domain.} 
Crucially, our method avoids the $O(N^n)$ complexity blow-up of higher-order attention. Instead, it matches the complexity of FFT-based operators with a cost of $O(N \log(N))$ per layer.

Finally, the triangular attention mechanism of the Edge Transformer ~\citep{triangular_attention-edge_transformer} and the triangle attention of AlphaFold2 ~\citep{alphafold2} are also close to our work. 
In both cases, the triangle refers to three-way interactions in the \textit{spatial domain}: given a triplet of nodes, triangular attention models the dependencies along the edges of the corresponding triangle, enabling richer modeling power. 
\emph{In our work, the triadic spectral convolution (order $n=2$) realizes the analog of triangular attention in the Fourier domain}. So the triangle corresponds to a triplet of frequency modes whose wavevectors satisfy a closure relation (e.g. $k_1+k_2=k_3$). The motivation is to capture the nonlinear triadic interactions that govern, for instance, energy transfer in PDE dynamics. 
For $n > 2$, our method is the Fourier analog of an $n$-symplicial extension of these attention mechanisms.
\section{Background} \label{sec:setting}

We consider a time-dependent PDE on a spatial domain
$\Omega\subset\mathbb{R}^d$, with boundary $\partial\Omega$, and on the
time interval $[0,T]$. A solution
$u:\Omega\times[0,T]\to\mathbb{R}$ satisfies
\begin{equation}\label{eq:PDE}
\begin{aligned}
&\forall x\in\Omega,\ \forall t\in(0,T],\qquad
\frac{\partial u}{\partial t}(x,t)
=
F\!\left(
\nu,t,x,u(x,t),\nabla_x u(x,t),\nabla_x^2u(x,t),\dots
\right),
\\
&\forall x\in\partial\Omega,\ \forall t\in(0,T],\qquad
\mathcal{B}(u)(x,t)=0,
\qquad
\forall x\in\Omega,\qquad
u(x,0)=u^0(x),
\end{aligned}
\end{equation}
where $\nu$ denotes PDE coefficients, $\mathcal{B}$ encodes the boundary
conditions, and $u^0\sim p^0$ is sampled from a probability distribution
on $L^2(\Omega;\mathbb{R})$, the space of real valued square-integrable functions on the spatial domain $\Omega$. 

\paragraph{Neural Operator.}
For dynamics governed by fixed coefficients $\nu$, the operator learning task consists of approximating the one-step solution operator $\mathcal{G}^\dagger$, for a fixed time step $\Delta t>0$ and $t\in[0,T-\Delta t]$:
\begin{equation}\label{eq:operator}
\begin{aligned}
\mathcal{G}^\dagger:
L^2(\Omega;\mathbb{R})
&\to
L^2(\Omega;\mathbb{R}),
\\
u(\cdot,t)
&\mapsto
u(\cdot,t+\Delta t).
\end{aligned}
\end{equation}

Following~\citep{kovachki_survey}, we approximate $\mathcal{G}^\dagger$ by a neural operator $\mathcal{G}_\theta$ composed of a pointwise lifting network $\mathcal{L}$, $L$ learnable operator layers $\mathcal{Q}_1,\dots,\mathcal{Q}_L$, and a pointwise projection network $\mathcal{P}$:
\begin{equation}\label{eq:NO_layers}
\mathcal{G}_\theta
=
\mathcal{P}
\circ
\mathcal{Q}_L
\circ\cdots\circ
\mathcal{Q}_1
\circ
\mathcal{L}.
\end{equation}
Let $v_\ell\in L^2(\Omega;\mathbb{R}^{c_\ell})$ denote the hidden representation at layer $\ell$, with $v_0=\mathcal{L}(u(\cdot,t))$.
For $\ell=1,\dots,L$, the layer $\mathcal{Q}_\ell:v_{\ell-1}\mapsto v_\ell$ is defined by
\begin{equation}\label{eq:Q_l}
v_\ell(x)
=
\mathcal{Q}_\ell(v_{\ell-1})(x)
=
\sigma\!\left(
W_\ell v_{\ell-1}(x)
+
\mathcal{K}_\ell(v_{\ell-1})(x)
+
b_\ell(x)
\right),
\qquad x\in\Omega,
\end{equation}
where $W_\ell$ is a pointwise linear map, $b_\ell$ is a bias term, $\mathcal{K}_\ell$ is a nonlocal integral operator, and $\sigma$ is a pointwise activation function. 
In the following, when there is no ambiguity, we drop the layer index.

\paragraph{Fourier Neural Operator (FNO).}
FNO~\citep{FNO} implements the nonlocal operator $\mathcal{K}_\ell$ as a global convolution parameterized in the Fourier domain, after a truncation of the input in the spectral domain. 
Let $\mathcal{M}\subset\mathbb{Z}^d$ denote the retained set of Fourier modes. Assuming a periodic rectangular domain,
with the convention  $\Omega=\mathbb{T}^d=\mathbb{R}^d/(2\pi\mathbb{Z})^d$, the Fourier truncation is
\begin{equation}
\mathsf{T}_{\mathcal{M}}v(x)
=
\sum_{k\in\mathcal{M}}
\widehat{v}(k)e^{ik\cdot x}.
\end{equation}
Then, for a given layer $\ell$, the spectral convolution can be written in physical space as in Eq. \ref{eq:spectral_conv}: 
\begin{equation}\label{eq:spectral_conv}
\mathcal{C}_{\theta_\ell}(v)(x)
=
\int_\Omega
\kappa_{\theta_\ell}(x-y)\mathsf{T}_{\mathcal{M}}v(y)\,dy.
\end{equation}
Equivalently, it is implemented in Fourier space by
\begin{equation}\label{eq:FNO}
\widehat{\mathcal{C}_{\theta_\ell}(v)}(k)
=
\begin{cases}
R_\ell^{(k)}\widehat{v}(k), & k\in\mathcal{M},\\
0, & k\notin\mathcal{M},
\end{cases}
\end{equation}
where $R_\ell^{(k)}$ is a learned complex-valued matrix acting on the channel dimension of the $k$-th Fourier mode. The output is then transformed back to the physical space with inverse FFT.
We remark that, \textit{within each linear spectral convolution, retained Fourier modes are updated independently of each other and cross-mode interactions are induced only through the pointwise nonlinearities between FNO layers}.

\paragraph{Polynomial nonlinearities in PDEs.}
Let $\alpha=(\alpha_1,\dots,\alpha_d)\in\mathbb{N}_0^d$ be a
multi-index, with $|\alpha|=\alpha_1+\cdots+\alpha_d$, and define
\begin{equation}
D_x^\alpha u
=
\frac{\partial^{|\alpha|}u}
{\partial x_1^{\alpha_1}\cdots\partial x_d^{\alpha_d}}.
\end{equation}
We use the convention $D_x^0u=u$. For a PDE of spatial order $m$, the collection of spatial derivatives up to order $m$ is
$
\{D_x^\alpha u:0\leq|\alpha|\leq m\}.
$
Thus Eq.~\ref{eq:PDE} can be written in multi-index form as
\begin{equation}\label{eq:PDE_multiindex}
\begin{aligned}
&\forall x\in\Omega,\ \forall t\in(0,T],\qquad
\frac{\partial u}{\partial t}(x,t)
=
F\!\left(
\nu,t,x,
\{D_x^\alpha u(x,t):0\leq|\alpha|\leq m\}
\right),
\\
&\forall x\in\partial\Omega,\ \forall t\in(0,T],\qquad
\mathcal{B}(u)(x,t)=0,
\qquad
\forall x\in\Omega,\qquad
u(x,0)=u^0(x).
\end{aligned}
\end{equation}

When the right-hand side is polynomial in $u$ and its spatial derivatives, it can be decomposed into homogeneous components:
\begin{equation}\label{eq:decoupled_PDE}
\frac{\partial u}{\partial t}(x,t)
=
\sum_{j=0}^{n}
\mathsf{P}_{\mathcal{I}_j,j}[u](x,t).
\end{equation}
Here $\mathsf{P}_{\mathcal{I}_j,j}$ denotes the $j$-homogeneous component of the PDE nonlinearity. More precisely, denoting as $\boldsymbol{\alpha}$ a set of multi-indices 
$(\alpha^{(1)},\dots,\alpha^{(j)}) \in (\mathbb{N}_0^d)^j,$
where each $\alpha^{(r)}$ is a spatial multi-index, 
let $\mathcal{I}_j\subset(\mathbb{N}_0^d)^j$ be the set of multi-index tuples appearing in degree-$j$ terms. Then
\begin{equation}
\mathsf{P}_{\mathcal{I}_j,j}[u](x,t)
=
\sum_{\boldsymbol{\alpha}\in\mathcal{I}_j}
c_{\boldsymbol{\alpha}}
\prod_{r=1}^{j}
D_x^{\alpha^{(r)}}u(x,t).
\end{equation}
We refer to the maximal value of $j$ in Eq.~\ref{eq:decoupled_PDE} as \textit{the polynomial degree of the PDE nonlinearity} and denote it by $n$. 
For instance, quadratic nonlinearities appear in Burgers' equation and in the advective term of the incompressible Navier--Stokes equations.

Now, we focus on the degree-$n$ component,
\begin{equation}\label{eq:n_linear_part}
\mathsf{P}_{\mathcal{I}_n,n}[u](x,t)
=
\sum_{\boldsymbol{\alpha}\in\mathcal{I}_n}
c_{\boldsymbol{\alpha}}
\prod_{r=1}^{n}
D_x^{\alpha^{(r)}}u(x,t).
\end{equation}
On the torus $\Omega=\mathbb{T}^d$, we expand scalar functions as
\begin{equation}\label{eq:fourier_basis}
u(x,t)
=
\sum_{k\in\mathbb{Z}^d}
\widehat{u}(k,t)e^{ik\cdot x},
\qquad
\widehat{u}(k,t)\in\mathbb{C}.
\end{equation}
Taking the Fourier transform of Eq.~\ref{eq:n_linear_part} gives
\begin{equation}\label{eq:FT_polynomial_main}
\widehat{\mathsf{P}_{\mathcal{I}_n,n}[u]}(k,t)
=
\sum_{\boldsymbol{\alpha}\in\mathcal{I}_n}
\sum_{\substack{k_1+\cdots+k_n=k\\ k_1,\dots,k_n\in\mathbb{Z}^d}}
C_{\boldsymbol{\alpha}}(k_1,\dots,k_n)
\prod_{r=1}^{n}
\widehat{u}(k_r,t),
\end{equation}
with
$\qquad
C_{\boldsymbol{\alpha}}(k_1,\dots,k_n)
=
c_{\boldsymbol{\alpha}}
\prod_{r=1}^{n}
(ik_r)^{\alpha^{(r)}},
\qquad
(ik_r)^{\alpha^{(r)}}
:=
\prod_{q=1}^{d}
(ik_{r,q})^{\alpha^{(r)}_q}.
$

The inner sum in Eq.~\ref{eq:FT_polynomial_main} is an $n$-fold convolution of Fourier modes. It makes explicit how polynomial nonlinearities mix frequencies, whereas standard spectral neural operators use linear
spectral convolutions that act mode-wise and rely on pointwise nonlinearities to induce such mixing. 
This is precisely this mixing that our proposed higher-order spectral convolution is designed to model.
\section{Beyond pointwise non-linearities}
\label{sec:beyond}

\paragraph{Motivation.}
Pointwise non-linearities already induce interactions between Fourier coefficients. However, as discussed previously, current approaches do not distinguish between different orders of non-linearity nor between the individual factors appearing in the nonlinear terms. 
Indeed, as noted by~\cite {azeglio2024convolution}, the current formulation of spectral neural operators modulates all interactions with the same set of weights inherited from the preceding linear layer. This limitation becomes clear by expanding a pointwise non-linearity as a polynomial series. 
Consider, for simplicity, the expansion of a scalar non-linearity applied after a linear combination:
\begin{equation}
\begin{aligned}
    \sigma\!\left(\tc{blue}{w_1} x_1 + \tc{green!60!black}{w_2} x_2\right) \approx & \lambda_0 
    + \tc{purple!90!black}{\lambda_1} (\tc{blue}{w_1} x_1 + \tc{green!60!black}{w_2} x_2)  + \tc{orange!80!red}{\lambda_2} (\tc{blue}{w_1} x_1 + \tc{green!60!black}{w_2} x_2)^2
    + \ldots \\
    =& \lambda_0 
    + \tc{purple!90!black}{\lambda_1} \tc{blue}{w_1} x_1 
    + \tc{purple!90!black}{\lambda_1} \tc{green!60!black}{w_2} x_2 + \tc{orange!80!red}{\lambda_2} \tc{blue}{w_1}^2 x_1^2
    + \tc{orange!80!red}{\lambda_2} \tc{green!60!black}{w_2}^2 x_2^2  + 2 \tc{orange!80!red}{\lambda_2} \tc{blue}{w_1} \tc{green!60!black}{w_2} x_1 x_2
    + \ldots
\end{aligned} \label{eq:eq_example}
\end{equation}
In Eq. \ref{eq:eq_example}, all quadratic interactions are tied through the same weights $\tc{blue}{w_1}, \tc{green!60!black}{w_2}$ and coefficient $\tc{orange!80!red}{\lambda_2}$. As a result, the model cannot independently control the different higher order terms (e.g., $x_1^2$, $x_2^2$, or $x_1 x_2$), leading to a weight-sharing constraint across nonlinear orders and factors.
\emph{In contrast, our Higher-Order Convolution explicitly parameterizes each nonlinear interaction term}. We argue that this distinction is crucial when modeling nonlinear operators arising in partial differential equations. For instance, modeling the quadratic advection term in Navier-Stokes requires controlling distinct bilinear interactions between components and derivatives, which a linear layer followed by a pointwise non-linearity can only capture implicitly with enough depth.
Next, we present, as a concrete example, the non-linear interactions on the incompressible Navier-Stokes equation. 

\paragraph{Example of Navier-Stokes equation.}\label{sec:navier}
We express the PDE in the vorticity form  as follows in Equation \ref{eq:vorticity}:
\begin{equation}\label{eq:vorticity}
\begin{aligned}
    \partial_t w(x,t) &= \nu \Delta w(x,t)
    - \colorbox{orange!20}{$\displaystyle(\nabla^\perp \Delta^{-1} w)(x,t)\cdot \nabla w(x,t)$}
    + f(x), \\
    u(x,t) &= \nabla^\perp \Delta^{-1} w(x,t), \qquad
    \nabla \cdot u(x,t)=0, \qquad
    w(x,0)=w_0(x).
\end{aligned}
\end{equation}
where $x \in (0,1)^2$, $t \in (0,T]$. Throughout this example, we assume that $w$ has zero spatial mean, so that $\Delta^{-1}$ is well-defined on the nonzero Fourier modes.
\color{black}
To observe the interaction of the Fourier modes of the vorticity, we take the Fourier transform, for $k \in \mathbb{Z}^2, \ t \in (0, T]$:
\begin{equation}\label{eq:Fourier_Navier-Stokes}
\begin{aligned}   \partial_t(\widehat{w})(k, t) = - \nu (2\pi)^2 \lvert k \rvert^2 \widehat{w}(k, t)
    - \colorbox{orange!20}{$\displaystyle
  \sum_{p + q = k}
  \frac{(p + q) \cdot p^\perp}{\lvert p \rvert^2}
  \widehat{w}(p, t)\widehat{w}(q, t)
  $} + \widehat{f}(k, t) .
\end{aligned}
\end{equation}
with  $p=(p_1, p_2)$ and $p^\perp=(-p_2, p_1)$. In Fourier space, the \colorbox{orange!20}{nonlinear advection term} of Eq. \ref{eq:Fourier_Navier-Stokes} takes the form of a convolution, inducing triadic interactions that redistribute kinetic energy across Fourier modes. We refer to the appendix \ref{app:extended_navier} for a more detailed discussion on Navier-Stokes equation.
\newline
As highlighted by \citet{triad_interaction_NS}, the primary difficulty in working with the spectral Navier–Stokes equations described in Eq.\ref{eq:Fourier_Navier-Stokes} is to appropriately account for all nonlinear interactions. 
An analytical treatment requires some means of tracking energy transfer from two arbitrary modes $p$ and $q$ into a third mode $k$. 
Therefore, it motivates the use of architectures that go beyond diagonal modulation of Fourier coefficients by explicitly parameterizing higher-order interactions in the spectral domain.

\textbf{Higher-Order Fourier Neural Operators (HO-FNO).}
We extend the kernel map to incorporate explicit $m$-linear interactions via the following Higher-Order Spectral Convolution, given in Eq. \ref{eq:HO_FNO}, where each $A_i$ is a learnable linear operator acting channel-wise in physical space. 
In this work, we instantiate the $A_i$ as point-wise linear maps shared across spatial locations but not across layers. We report in Appendix~\ref{app:structure_matrices_ablation} ablation results about the structure of the $A_i$ matrices.
The $m$-linear point-wise products of Eq. \ref{eq:HO_FNO} in physical space induce a structured $m$-linear global mixing among Fourier coefficients as given by Eq.~\ref{eq:HO_FNO_mixing}.
\begin{equation}
\label{eq:HO_FNO}
    \big(\mathcal{H}_\theta u\big)(x) = \int_\Omega k_\theta(x-y) \mathsf{T}_M \big( (A_1 u) (A_2 u) \cdots (A_m u)\big) (y) dy
\end{equation}
\begin{equation} \label{eq:HO_FNO_mixing}
    (\widehat{\mathcal{H}_\theta v})(k) = R^{(k)} \sum_{k_1 + \ldots + k_m = k} A_1 \widehat{v}(k_1) A_2 \widehat{v}(k_2) \cdots A_m \widehat{v}(k_m)
\end{equation}
Thus, each mode $k$ aggregates all $m$-tuples of modes with indices summing to $k$, mirroring the nonlinear interaction structure of PDEs with polynomial nonlinearities.
We provide further illustrations of the method in Appendix~\ref{app:additional_method_illustrations}, and  formalize the interest of the proposed inductive-bias in Appendix~\ref{sec:inductive}. 
There, we introduce a local expansion of the Fourier-truncated one-step solution operator, define the tied class of higher-order interactions induced by pointwise nonlinearities and the order-$n$ HO-FNO class realized by equation~\ref{eq:HO_FNO_mixing}. 
We show that in the matched case $m=n$, HO-FNO exactly matches the degree-$n$ Fourier interaction while shallow FNO is restricted to modeling a tied subclass of interactions. Hence, this proves that a local approximation gap remains between FNO and HO-FNO whenever the target interaction lies outside that class of interactions. 
\newline
Regarding point clouds, we build on the Geometry-Aware Fourier Neural Operator framework \cite{GEO-FNO} to extend our spectral neural operator to this setting.

\section{Experiments}\label{sec:experiments}
To showcase the interest of the higher-order spatial convolution, we introduce in subsection \ref{subsec: poly_Poisson} a simulated dataset consisting of a Poisson equation with polynomial forcing.
On this dataset, HO-FNO with a single layer outperforms FNO with multiple layers (up to $16$). 
Then, we evaluate the proposed HO-FNO across a wide range of PDE problems and compare it with over 20 recent operator-learning methods, including attention-based and state-space models. 
Table~\ref{tab:datasets_summary} provides a summary of the datasets, which span different geometry types, number of spatial dimensions, and problem sizes. We provide further information on the benchmarks and hyperparameter settings in Appendix~\ref{app:implementation_details}. 
We provide several ablations to assess the zero-shot super-resolution capabilities of HO-FNO in Appendix\ref{app:resolution_equivariance}, measure its improved efficiency compared to Transformers and state-space models, and isolate the specific contribution of higher-order spectral convolutions relative to other experimental choices. 
Finally, we ablate the matrix structure used to induce higher-order mixing in Appendix~\ref{app:structure_matrices_ablation}.

\subsection{A single HO-FNO layer is worth 16 FNO layers: Poisson Equation with Polynomial source}\label{subsec: poly_Poisson}
We introduce the Polynomial-Source Poisson dataset, a controlled nonlinear elliptic PDE benchmark designed to evaluate whether neural operators can efficiently represent multiplicative field interactions of increasing order. 
The dataset is parameterized by an integer degree $p$ that indicates the order of non-linearity of the equation. For each sample, the model receives $p$ independent input fields $u_1, \ldots, u_p$ and must predict the solution $v$ of a Poisson equation whose source term is the pointwise product of these input fields. 
The PDE is the following:
\begin{align}
    -\Delta v(x,y)
    =
    \prod_{\ell=1}^{p} u_{\ell}(x,y)
    -
    \frac{1}{|\Omega|}
    \int_{\Omega}
    \prod_{\ell=1}^{p} u_{\ell}(x',y')
    \,dx'\,dy',
    \qquad (x,y) \in \Omega.
\end{align}
The benchmark is motivated by a common structure in physics-based PDEs: nonlinear interactions between fields often appear as source terms, forcing terms, or closure terms\citep{multiple_conditioning_ref}, typical in electrostatics\citep{poisson_boltzmann_electrostatic, poisson_boltzmann_biomolecular}, phase-field models\citep{poisson_phase_field}, and pressure recovery in incompressible fluid dynamics\citep{poisson_fluid}. A detailed description of the equation and data-generation procedure is provided in Appendix ~\ref{app: poly poisson}.
\paragraph{Results} Across all considered Poisson datasets in Figure~\ref{fig:Polynomial-Source Poisson_layers}, every HO-FNO variant consistently outperforms FNO with the same number of layers. Moreover, on all the nonlinear Poisson datasets ($p \in \{2,3,5\}$) a single-layer HO-FNO of order $p$ already outperforms FNOs with up to $16$ layers. The largest gap is observed on the most challenging dataset, $p=5$, where HO-FNO achieves a loss that is two orders of magnitude lower than that of FNO. Since each HO-FNO layer has roughly the same number of parameters as an FNO layer, this means that, on these benchmarks, HO-FNO can outperform FNO models with up to $16$ times more layers and parameters. 
On the Poisson equation with linear forcing ($p=1$), the trend is similar, but the gap between FNO and HO-FNO variants is much smaller than in the nonlinear cases. This suggests that HO-FNO can also help on linear PDEs, while its largest gains emerge in nonlinear regimes.
\definecolor{FNOColor}{RGB}{220,75,60}       
\definecolor{HO2Color}{RGB}{46,160,67}       
\definecolor{HO3Color}{RGB}{30,136,229}      
\definecolor{HO5Color}{RGB}{142,68,173}      
\begin{figure*}[!h]
\centering
\resizebox{\textwidth}{!}{%
\begin{tabular}{cccc}
\begin{tikzpicture}
\begin{axis}[
    width=4.4cm, height=4.1cm,
    xlabel={\textbf{Number of Layers}},
    ylabel={\small \textbf{Test MSE}},
    title={$p=1$},
    xmin=0.8, xmax=16.2,
    ymin=1e-3, ymax=1e-2,
    ymode=log,
    xtick={1,2,4,8,16},
    ymajorgrids=true,
    grid style=dashed,
    title style={font=\footnotesize},
    label style={font=\footnotesize},
    tick label style={font=\scriptsize},
    unbounded coords=jump
]

\addplot[name path=FNOupperp1, draw=none] coordinates {
    (1,0.00516) (2,0.00577) (4,0.00423) (8,0.00316) (16,0.00391)
};
\addplot[name path=FNOlowerp1, draw=none] coordinates {
    (1,0.00398) (2,0.00533) (4,0.00327) (8,0.00266) (16,0.00333)
};
\addplot[
    fill=FNOColor,
    fill opacity=0.18,
    draw=none,
    forget plot
] fill between[of=FNOupperp1 and FNOlowerp1];
\addplot[
    color=FNOColor,
    mark=square*,
    mark size=1.5pt,
    line width=0.8pt
] coordinates {
    (1,0.00457) (2,0.00555) (4,0.00375) (8,0.00291) (16,0.00362)
};

\addplot[name path=HO2upperp1, draw=none] coordinates {
    (1,0.00727) (2,0.00184) (4,0.00269) (8,0.00170) (16,0.00309)
};
\addplot[name path=HO2lowerp1, draw=none] coordinates {
    (1,0.00601) (2,0.00166) (4,0.00189) (8,0.00166) (16,0.00229)
};
\addplot[
    fill=HO2Color,
    fill opacity=0.18,
    draw=none,
    forget plot
] fill between[of=HO2upperp1 and HO2lowerp1];
\addplot[
    color=HO2Color,
    mark=square*,
    mark size=1.5pt,
    line width=0.8pt
] coordinates {
    (1,0.00664) (2,0.00175) (4,0.00229) (8,0.00168) (16,0.00269)
};
\addplot[name path=HO3upperp1, draw=none] coordinates {
    (1,0.00295) (2,0.00246) (4,0.00337) (8,0.00342) (16,0.00370)
};
\addplot[name path=HO3lowerp1, draw=none] coordinates {
    (1,0.00251) (2,0.00188) (4,0.00251) (8,0.00222) (16,0.00324)
};
\addplot[
    fill=HO3Color,
    fill opacity=0.18,
    draw=none,
    forget plot
] fill between[of=HO3upperp1 and HO3lowerp1];
\addplot[
    color=HO3Color,
    mark=square*,
    mark size=1.5pt,
    line width=0.8pt
] coordinates {
    (1,0.00273) (2,0.00217) (4,0.00294) (8,0.00282) (16,0.00347)
};

\addplot[name path=HO5upperp1, draw=none] coordinates {
    (1,0.00638) (2,0.00212) (4,0.00327) (8,0.00231) (16,0.00342)
};
\addplot[name path=HO5lowerp1, draw=none] coordinates {
    (1,0.00522) (2,0.00200) (4,0.00247) (8,0.00189) (16,0.00258)
};
\addplot[
    fill=HO5Color,
    fill opacity=0.18,
    draw=none,
    forget plot
] fill between[of=HO5upperp1 and HO5lowerp1];
\addplot[
    color=HO5Color,
    mark=square*,
    mark size=1.5pt,
    line width=0.8pt
] coordinates {
    (1,0.00580) (2,0.00206) (4,0.00287) (8,0.00210) (16,0.00300)
};

\end{axis}
\end{tikzpicture}%
&
\begin{tikzpicture}
\begin{axis}[
    width=4.4cm, height=4.1cm,
    xlabel={\textbf{Number of Layers}},
    title={$p=2$},
    xmin=0.8, xmax=16.2,
    ymin=1.5e-4, ymax=2e-1,
    ymode=log,
    xtick={1,2,4,8,16},
    ymajorgrids=true,
    grid style=dashed,
    title style={font=\footnotesize},
    label style={font=\footnotesize},
    tick label style={font=\scriptsize},
    unbounded coords=jump
]

\addplot[name path=FNOupperp2, draw=none] coordinates {
    (1,0.09847) (2,0.00229) (4,0.00242) (8,0.00318) (16,0.00227)
};
\addplot[name path=FNOlowerp2, draw=none] coordinates {
    (1,0.09181) (2,0.00185) (4,0.00144) (8,0.00168) (16,0.00183)
};
\addplot[
    fill=FNOColor,
    fill opacity=0.18,
    draw=none,
    forget plot
] fill between[of=FNOupperp2 and FNOlowerp2];
\addplot[
    color=FNOColor,
    mark=square*,
    mark size=1.5pt,
    line width=0.8pt
] coordinates {
    (1,0.09514) (2,0.00207) (4,0.00193) (8,0.00243) (16,0.00205)
};

\addplot[name path=HO2upperp2, draw=none] coordinates {
    (1,0.00093) (2,0.00051) (4,0.00042) (8,0.00023) (16,0.00035)
};
\addplot[name path=HO2lowerp2, draw=none] coordinates {
    (1,0.00053) (2,0.00029) (4,0.00024) (8,0.00019) (16,0.00023)
};
\addplot[
    fill=HO2Color,
    fill opacity=0.18,
    draw=none,
    forget plot
] fill between[of=HO2upperp2 and HO2lowerp2];
\addplot[
    color=HO2Color,
    mark=square*,
    mark size=1.5pt,
    line width=0.8pt
] coordinates {
    (1,0.00073) (2,0.00040) (4,0.00033) (8,0.00021) (16,0.00029)
};

\addplot[name path=HO3upperp2, draw=none] coordinates {
    (1,0.00041) (2,0.00070) (4,0.00064) (8,0.00029) (16,0.00188)
};
\addplot[name path=HO3lowerp2, draw=none] coordinates {
    (1,0.00031) (2,0.00046) (4,0.00034) (8,0.00021) (16,0.00032)
};
\addplot[
    fill=HO3Color,
    fill opacity=0.18,
    draw=none,
    forget plot
] fill between[of=HO3upperp2 and HO3lowerp2];
\addplot[
    color=HO3Color,
    mark=square*,
    mark size=1.5pt,
    line width=0.8pt
] coordinates {
    (1,0.00036) (2,0.00058) (4,0.00049) (8,0.00025) (16,0.00110)
};

\addplot[name path=HO5upperp2, draw=none] coordinates {
    (1,0.00032) (2,0.00033) (4,0.00163) (8,0.00131) (16,0.00207)
};
\addplot[name path=HO5lowerp2, draw=none] coordinates {
    (1,0.00016) (2,0.00017) (4,0.00017) (8,0.00111) (16,0.00071)
};
\addplot[
    fill=HO5Color,
    fill opacity=0.18,
    draw=none,
    forget plot
] fill between[of=HO5upperp2 and HO5lowerp2];
\addplot[
    color=HO5Color,
    mark=square*,
    mark size=1.5pt,
    line width=0.8pt
] coordinates {
    (1,0.00024) (2,0.00025) (4,0.00090) (8,0.00121) (16,0.00139)
};

\end{axis}
\end{tikzpicture}%
&
\begin{tikzpicture}
\begin{axis}[
    width=4.4cm, height=4.1cm,
    xlabel={\textbf{Number of Layers}},
    title={$p=3$},
    xmin=0.8, xmax=16.2,
    ymin=3e-4, ymax=1.2,
    ymode=log,
    xtick={1,2,4,8,16},
    ymajorgrids=true,
    grid style=dashed,
    title style={font=\footnotesize},
    label style={font=\footnotesize},
    tick label style={font=\scriptsize},
    unbounded coords=jump
]

\addplot[name path=FNOupperp3, draw=none] coordinates {
    (1,0.73489) (2,0.02351) (4,0.01362) (8,0.00699) (16,0.00576)
};
\addplot[name path=FNOlowerp3, draw=none] coordinates {
    (1,0.72287) (2,0.01253) (4,0.00848) (8,0.00499) (16,0.00484)
};
\addplot[
    fill=FNOColor,
    fill opacity=0.18,
    draw=none,
    forget plot
] fill between[of=FNOupperp3 and FNOlowerp3];
\addplot[
    color=FNOColor,
    mark=square*,
    mark size=1.5pt,
    line width=0.8pt
] coordinates {
    (1,0.72888) (2,0.01802) (4,0.01105) (8,0.00599) (16,0.00530)
};

\addplot[name path=HO2upperp3, draw=none] coordinates {
    (1,0.57697) (2,0.00483) (4,0.00462) (8,0.00407) (16,0.00470)
};
\addplot[name path=HO2lowerp3, draw=none] coordinates {
    (1,0.49731) (2,0.00305) (4,0.00204) (8,0.00283) (16,0.00238)
};
\addplot[
    fill=HO2Color,
    fill opacity=0.18,
    draw=none,
    forget plot
] fill between[of=HO2upperp3 and HO2lowerp3];
\addplot[
    color=HO2Color,
    mark=square*,
    mark size=1.5pt,
    line width=0.8pt
] coordinates {
    (1,0.53714) (2,0.00394) (4,0.00333) (8,0.00345) (16,0.00354)
};

\addplot[name path=HO3upperp3, draw=none] coordinates {
    (1,0.00257) (2,0.00204) (4,0.00161) (8,0.00187) (16,0.00143)
};
\addplot[name path=HO3lowerp3, draw=none] coordinates {
    (1,0.00211) (2,0.00144) (4,0.00107) (8,0.00127) (16,0.00067)
};
\addplot[
    fill=HO3Color,
    fill opacity=0.18,
    draw=none,
    forget plot
] fill between[of=HO3upperp3 and HO3lowerp3];
\addplot[
    color=HO3Color,
    mark=square*,
    mark size=1.5pt,
    line width=0.8pt
] coordinates {
    (1,0.00234) (2,0.00174) (4,0.00134) (8,0.00157) (16,0.00105)
};

\addplot[name path=HO5upperp3, draw=none] coordinates {
    (1,0.00066) (2,0.00122) (4,0.00167) (8,0.00250) (16,0.00479)
};
\addplot[name path=HO5lowerp3, draw=none] coordinates {
    (1,0.00052) (2,0.00080) (4,0.00151) (8,0.00078) (16,0.00155)
};
\addplot[
    fill=HO5Color,
    fill opacity=0.18,
    draw=none,
    forget plot
] fill between[of=HO5upperp3 and HO5lowerp3];
\addplot[
    color=HO5Color,
    mark=square*,
    mark size=1.5pt,
    line width=0.8pt
] coordinates {
    (1,0.00059) (2,0.00101) (4,0.00159) (8,0.00164) (16,0.00317)
};

\end{axis}
\end{tikzpicture}%
&
\begin{tikzpicture}
\begin{axis}[
    width=4.4cm, height=4.1cm,
    xlabel={\textbf{Number of Layers}},
    title={$p=5$},
    xmin=0.8, xmax=16.2,
    ymin=2e-3, ymax=2,
    ymode=log,
    xtick={1,2,4,8,16},
    ymajorgrids=true,
    grid style=dashed,
    title style={font=\footnotesize},
    label style={font=\footnotesize},
    tick label style={font=\scriptsize},
    unbounded coords=jump
]

\addplot[name path=FNOupperp5, draw=none] coordinates {
    (1,1.12996) (2,1.23792) (4,1.07371) (8,1.14072) (16,1.21406)
};
\addplot[name path=FNOlowerp5, draw=none] coordinates {
    (1,1.10368) (2,0.82168) (4,0.89437) (8,1.07102) (16,1.20140)
};
\addplot[
    fill=FNOColor,
    fill opacity=0.18,
    draw=none,
    forget plot
] fill between[of=FNOupperp5 and FNOlowerp5];
\addplot[
    color=FNOColor,
    mark=square*,
    mark size=1.5pt,
    line width=0.8pt
] coordinates {
    (1,1.11682) (2,1.02980) (4,0.98404) (8,1.10587) (16,1.20773)
};

\addplot[name path=HO2upperp5, draw=none] coordinates {
    (1,1.16875) (2,1.36781) (4,1.21381) (8,1.16927) (16,1.15700)
};
\addplot[name path=HO2lowerp5, draw=none] coordinates {
    (1,1.10657) (2,1.05707) (4,0.92407) (8,1.13439) (16,1.08688)
};
\addplot[
    fill=HO2Color,
    fill opacity=0.18,
    draw=none,
    forget plot
] fill between[of=HO2upperp5 and HO2lowerp5];
\addplot[
    color=HO2Color,
    mark=square*,
    mark size=1.5pt,
    line width=0.8pt
] coordinates {
    (1,1.13766) (2,1.21244) (4,1.06894) (8,1.15183) (16,1.12194)
};

\addplot[name path=HO3upperp5, draw=none] coordinates {
    (1,1.05274) (2,1.42782) (4,1.20711) (8,1.11394) (16,1.09049)
};
\addplot[name path=HO3lowerp5, draw=none] coordinates {
    (1,0.99496) (2,1.37076) (4,1.17541) (8,1.09538) (16,0.92007)
};
\addplot[
    fill=HO3Color,
    fill opacity=0.18,
    draw=none,
    forget plot
] fill between[of=HO3upperp5 and HO3lowerp5];
\addplot[
    color=HO3Color,
    mark=square*,
    mark size=1.5pt,
    line width=0.8pt
] coordinates {
    (1,1.02385) (2,1.39929) (4,1.19126) (8,1.10466) (16,1.00528)
};

\addplot[name path=HO5upperp5, draw=none] coordinates {
    (1,0.00927) (2,0.00606) (4,0.02089) (8,0.03887) (16,0.01232)
};
\addplot[name path=HO5lowerp5, draw=none] coordinates {
    (1,0.00775) (2,0.00254) (4,0.00409) (8,0.00953) (16,0.00604)
};
\addplot[
    fill=HO5Color,
    fill opacity=0.18,
    draw=none,
    forget plot
] fill between[of=HO5upperp5 and HO5lowerp5];
\addplot[
    color=HO5Color,
    mark=square*,
    mark size=1.5pt,
    line width=0.8pt
] coordinates {
    (1,0.00851) (2,0.00430) (4,0.01249) (8,0.02420) (16,0.00918)
};
\end{axis}
\end{tikzpicture}%
\end{tabular}%
}

\vspace{-0.2cm}
\begin{tikzpicture}
\begin{axis}[
    hide axis,
    legend columns=-1,
    axis lines=none,
    ticks=none,
    legend style={
        draw=none,
        column sep=2ex,
        font=\scriptsize
    },
    xmin=0, xmax=1, ymin=0, ymax=1
]
\addlegendimage{only marks, mark=square*, color=FNOColor}
\addlegendentry{\textbf{FNO (order 1)}}

\addlegendimage{only marks, mark=square*, color=HO2Color}
\addlegendentry{\textbf{HO-FNO (order 2)}}

\addlegendimage{only marks, mark=square*, color=HO3Color}
\addlegendentry{\textbf{HO-FNO (order 3)}}

\addlegendimage{only marks, mark=square*, color=HO5Color}
\addlegendentry{\textbf{HO-FNO (order 5)}}
\end{axis}
\end{tikzpicture}%

\caption{\small Test MSE as a function of the number of layers on the Polynomial-Source Poisson datasets for $p=1,2,3,$ and $5$. Solid lines denote the mean over runs, and shaded bands indicate one standard deviation. Lower values indicate better performance. Tables with quantitative results are provided in Appendix~\ref{app:extended_results_poisson}.}
\label{fig:Polynomial-Source Poisson_layers}
\end{figure*}
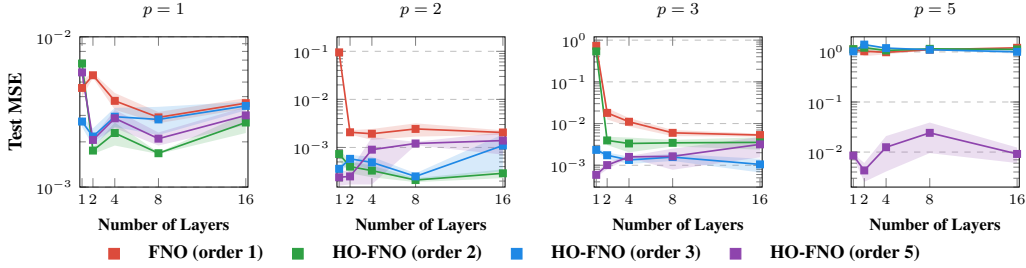

\subsection{Isolating the Effect of Higher-Order Spectral Convolutions}
 To isolate the contribution of the proposed Higher-Order spectral convolution from the choice of backbone architecture, we train both FNO and HO-FNO using the simpler backbone of the original FNO~\citeyearpar{FNO}. In all other experiments, both models use a modern residual backbone; the differences between the two backbones are detailed in Appendix~\ref{app:backbone_details}. We conduct this controlled comparison on the three Navier–Stokes datasets reported in Table~\ref{tab:results1M}. These datasets, introduced by \citet{FNO}, differ only in the viscosity, $\nu \in \{10^{-3}, 10^{-4}, 10^{-5}\}$. Further details are provided in Appendix~\ref{app: poly poisson}.
\paragraph{Results} HO-FNO outperforms FNO in all the $3$ fluid regimes, in all the $3$ reported metrics, with best performance at the highest viscosity where HO-FNO has a rollout loss that is almost an order of magnitude lower than the one of FNO. Additionally, HO-FNO performs better than DSFNO in the reported results and performs better than LaMO in the two datasets with higher viscosity while being slightly behind in the low viscosity benchmark. 

Additionally, the ablation shows that the order-$3$ model brings only marginal improvements over the order-$2$ version. This is consistent with both our theoretical motivation and the experiments on the Poisson dataset in Section~\ref{subsec: poly_Poisson}: the three benchmarks considered here involve PDEs with quadratic nonlinearities, and therefore we do not expect HO-FNO variants of order greater than $2$ to yield substantial additional gains.

\subsection{Main results} \label{subsec:benchmarks}

\paragraph{Tasks. }
We consider six PDE surrogate tasks 
that were introduced by \cite{FNO} and \cite{GEO-FNO} and have since become a standard suite for operator-learning work, as they represent a broad spectrum of physical phenomena, 
namely material stress/deformation (Elasticity, Plasticity), fluid dynamics governed by the Navier–Stokes equations (Airfoil, Pipe, and a spatio-temporal Navier–Stokes case), and porous-media flow described by Darcy’s law. 

\paragraph{Baselines} We compare HO-FNO against baselines reported in \citep{LaMO,transolver, mspt}, spanning classical models (e.g. Deep-O-Net \cite{DEEPONET}, frequency-based neural operators (FNO etc.), transformers, and state space models. 
We report the relative $L_2$ error between predicted and ground-truth fields. 
We conduct all the experiments building on the code and benchmark setup of LaMO \citep{LaMO}. 
\definecolor{msecol}{HTML}{F8F5E9}   
\definecolor{nrmsecol}{HTML}{EFF0F1}  
\definecolor{rollcol}{HTML}{E6ECF8} 
\begin{table*}[!h]

\caption{\small Test performance on Navier--Stokes datasets.
    We report the number of parameters, the test \colorbox{msecol}{MSE}, the normalized MSE (\colorbox{nrmsecol}{nRMSE}), the \colorbox{rollcol}{rollout nRMSE}.}
\label{tab:results1M}
\centering
\resizebox*{0.7\textwidth}{!}{
\begin{tabular}{c c  c  c  c  c  c}
\toprule
\multicolumn{2}{c}{model}
& FNO
    & \multicolumn{2}{c}{HO-FNO}
    & $\text{DSFNO}^*$ & LaMO\\
\cmidrule(lr){1-5}

\multicolumn{2}{c}{order}
    & 1 
    & 2
    & 3
    & 
    &  \\

\midrule

\multicolumn{2}{c}{N. parameters}
    & $1.09$M
    & $1.09$M & $1.10$M
    & $1.06$M
    & $1.54$M \\
\midrule

\multirow{3}{*}{\textbf{NS ($\nu = 10^{-3}$)}} 
    & \cellcolor{msecol} MSE 
        & \cellcolor{msecol} $3.0\times10^{-7}$
        & \cellcolor{msecol} $\underline{7.8\times10^{-8}}$ & \cellcolor{msecol}  $\mathbf{7.6\times10^{-8}}$
        &\cellcolor{msecol} --
        &\cellcolor{msecol} $1.6\times10^{-7}$ \\
    & \cellcolor{nrmsecol} nRMSE 
        & \cellcolor{nrmsecol} $4.4\times10^{-4}$
        & \cellcolor{nrmsecol} $\underline{2.8\times10^{-4}}$ & \cellcolor{nrmsecol} $\mathbf{2.7\times10^{-4}}$
        & \cellcolor{nrmsecol}--
        & \cellcolor{nrmsecol} $4.0\times10^{-4}$ \\
    & \cellcolor{rollcol} Rollout 
        & \cellcolor{rollcol} $1.2\times10^{-2}$
        & \cellcolor{rollcol}  $\underline{1.8\times10^{-3}}$ & \cellcolor{rollcol} $\mathbf{1.6\times10^{-3}}$
        & \cellcolor{rollcol} $5.6 \times 10^{-3}$
        & \cellcolor{rollcol} $3.8\times10^{-3}$ \\
\midrule

\multirow{3}{*}{\textbf{NS ($\nu = 10^{-4}$)}} 
    & \cellcolor{msecol} MSE 
        & \cellcolor{msecol} $2.6\times10^{-3}$
        & \cellcolor{msecol}$\mathbf{7.9\times10^{-4}}$ & \cellcolor{msecol} $\underline{7.9\times10^{-4}}$
        & \cellcolor{msecol}--
        & \cellcolor{msecol}$8.3 \times 10^{-4}$ \\
    & \cellcolor{nrmsecol} nRMSE 
        & \cellcolor{nrmsecol} $2.9\times10^{-2}$
        & \cellcolor{nrmsecol} $\mathbf{1.3\times10^{-2}}$ & \cellcolor{nrmsecol} $\underline{1.3\times10^{-2}}$
        &\cellcolor{nrmsecol} --
        &\cellcolor{nrmsecol} $1.4 \times 10^{-2}$ \\
    & \cellcolor{rollcol} Rollout 
        & \cellcolor{rollcol} $7.7\times10^{-2}$
        & \cellcolor{rollcol} $\mathbf{4.6\times10^{-2}}$ & \cellcolor{rollcol} $\underline{4.6\times10^{-2}}$
        & \cellcolor{rollcol} $6.0 \times 10^{-2}$
        & \cellcolor{rollcol} $5.0 \times 10^{-2}$ \\
\midrule

\multirow{3}{*}{\textbf{NS ($\nu = 10^{-5}$)}} 
    & \cellcolor{msecol} MSE 
        & \cellcolor{msecol} $1.8\times10^{-2}$
        & \cellcolor{msecol} $\underline{1.7 \times 10^{-2}}$ & \cellcolor{msecol} $\mathbf{1.7 \times 10^{-2}}$
        &\cellcolor{msecol} --
        &\cellcolor{msecol} $\mathbf{1.4 \times 10^{-2}}$ \\
    & \cellcolor{nrmsecol} nRMSE 
        & \cellcolor{nrmsecol} $6.7\times10^{-2}$
        & \cellcolor{nrmsecol}  $\mathbf{6.5 \times 10^{-2}}$ & \cellcolor{nrmsecol} $6.8 \times 10^{-2}$
        &\cellcolor{nrmsecol} --
        &\cellcolor{nrmsecol} $\mathbf{5.8 \times 10^{-2}}$ \\
    & \cellcolor{rollcol} Rollout 
        & \cellcolor{rollcol} $1.3\times10^{-1}$
        & \cellcolor{rollcol} $\mathbf{1.1 \times 10^{-1}}$ & \cellcolor{rollcol} $1.2 \times 10^{-1}$
        &\cellcolor{rollcol} --
        &\cellcolor{rollcol} $\mathbf{1.1 \times 10^{-1}}$ \\
\midrule

\end{tabular}
}
\vspace{-2mm}
\begin{minipage}{0.8\textwidth}
\tiny
$^*$ Original errors reported in \cite{DSFNO}, single-step metrics and results for Navier–Stokes with $\nu = 10^{-5}$ were not provided.
\end{minipage}
\end{table*}
\paragraph{Results} Across all benchmarks reported in Table~\ref{tab:main_result}, HO-FNO consistently outperforms other frequency-based models, and often by a wide margin over FNO. 
Moreover, \textit{HO-FNO is the only frequency-based model that is competitive with the strongest Transformer and state-space baselines}. 
It attains the best reported error on Navier--Stokes $(0.0319)$ and Plasticity $(0.0006)$, supporting the intuition that higher-order spectral interactions are a useful inductive bias for PDEs with high-order nonlinear structure. 
While Navier--Stokes contains a quadratic convection term, Plasticity involves analytic nonlinearities that are not finite-degree polynomials and can be locally approximated by an infinite power series. 
This suggests that the higher-order spectral bias may remain useful beyond strictly polynomial PDEs.
\begin{table*}[!h]
        \caption{\small Main results across standard datasets, reported as mean relative $\ell_2$ error (Eq.~\ref{eq:nRMSE_app}; lower is better). 
        \textbf{HO-FNO (Ours)} is evaluated under a parameter budget matched to LaMO on each task. 
        \textcolor{orange}{orange} denotes the best baseline, \textcolor{blue}{blue} the second-best, and \textcolor{violet}{violet} the third-best. Best result for each dataset is \textbf{in bold}.}
	\label{tab:main_result}
	\centering
	\begin{small}
		\begin{sc}
			\renewcommand{\multirowsetup}{\centering}
        \resizebox{\textwidth}{!}{
        \begin{tabular}{llcclcccc}
        \toprule
        \multicolumn{2}{c}{\multirow{3}{*}{Operator}} 
        & \multirow{2}{*}{Resolution}
        & \multicolumn{2}{c}{Regular Grid} 
        & \multicolumn{3}{c}{Structured Mesh} 
        & \multicolumn{1}{c}{Point Cloud} \\
        \cmidrule(lr){4-5} \cmidrule(lr){6-8} \cmidrule(lr){9-9}
        & & Equivariant & Navier--Stokes & Darcy & Plasticity & Airfoil & Pipe & Elasticity \\
        \midrule
        \midrule

        \multirow{4}{*}{\rotatebox[origin=c]{90}{Classic}} 
        & UNET \citeyearpar{UNET} 
        & \xmark
        & $0.1982$ & $0.0080$ & $0.0051$ & $0.0079$ & $0.0065$ & $0.0235$ \\

        & Resnet \citeyearpar{RESNET} 
        & \xmark
        & $0.2753$ & $0.0587$ & $0.0233$ & $0.0391$ & $0.0120$ & $0.0262$ \\

        & SWIN \citeyearpar{SWIN} 
        & \xmark
        & $0.2248$ & $0.0397$ & $0.0170$ & $0.0270$ & $0.0109$ & $0.0283$ \\

        & DeepONet \citeyearpar{DEEPONET} 
        & \xmark
        & $0.2972$ & $0.0588$ & $0.0135$ & $0.0385$ & $0.0097$ & f \\
                    
        \midrule

        \multirow{7}{*}{\rotatebox[origin=c]{90}{Transformer}}

        & Galerkin \citeyearpar{Galerkin_Trans}  
        & \xmark
        & $0.1401$ & $0.0084$ & $0.0120$ & $0.0118$ & $0.0098$ & $0.0240$ \\

        & HT-Net \citeyearpar{Ht-net} 
        & \xmark
        & $0.1847$ & $0.0079$ & $0.0333$ & $0.0065$ & $0.0059$ & / \\

        & OFormer \citeyearpar{OFormer} 
        & \xmark
        & $0.1705$ & $0.0124$ & $0.0017$ & $0.0183$ & $0.0168$ & $0.0183$ \\

        & GNOT \citeyearpar{GNOT}  
        & \xmark
        & $0.1380$ & $0.0105$ & $0.0336$ & $0.0076$ & $0.0047$ & $0.0086$ \\

        & FactFormer \citeyearpar{FactFormer}  
        & \xmark
        & $0.1214$ & $0.0109$ & $0.0312$ & $0.0071$ & $0.0060$ & / \\

        & ONO \citeyearpar{ONO} 
        & \cmark
        & $0.1195$ & $0.0076$ & $0.0048$ & $0.0061$ & $0.0052$ & $0.0118$ \\

        & Transolver \citeyearpar{transolver}
        & \xmark
        & $0.0957$ & $0.0059$ & $0.0013$ & $0.0053$ & $0.0046$ & $0.0064$ \\

        & Erwin \citeyearpar{erwin} & \xmark
        & - & - & $0.0010$ & $0.0257$ & $0.0061$ & \cellcolor{orange!20}\textbf{0.0034} \\

        & Transolver++ \citeyearpar{transolver++} & \xmark
        & $0.0719$ & \cellcolor{violet!10}$0.0058$ & $0.0011$ & \cellcolor{violet!10}$0.0048$ & $0.0061$ & $0.0052$ \\

        & MSPT \citeyearpar{mspt} & \xmark
        & \cellcolor{violet!10}$0.0632$ & $0.0063$ & \cellcolor{violet!10}$0.0010$ & $0.0051$ & \cellcolor{orange!20}\textbf{0.0031} & \cellcolor{blue!10}$0.0048$ \\

        \midrule
                 
        \multirow{2}{*}{\rotatebox[origin=c]{90}{SSM}}

        & LaMO \citeyearpar{LaMO} 
        & \xmark
        & \cellcolor{blue!10} $0.0460$ 
        & \cellcolor{orange!20}\textbf{0.0039} 
        & \cellcolor{blue!10}$0.0007$
        & \cellcolor{orange!20}\textbf{0.0041} 
        & \cellcolor{blue!10}0.0038
        & \cellcolor{violet!10}0.0050 \\
                   
        \\
        \midrule

        \multirow{6}{*}{\rotatebox[origin=c]{90}{Frequency}}
                      
        & WMT \citeyearpar{WMO} 
        & \xmark
        & $0.1541$ & $0.0082$ & $0.0076$ & $0.0075$ & $0.0077$ & $0.0359$ \\

        & U-FNO \citeyearpar{U-FNO}  
        & \cmark
        & $0.2231$ & $0.0183$ & $0.0039$ & $0.0269$ & $0.0056$ & $0.0239$ \\

        & FNO \citeyearpar{FNO, GEO-FNO}  
        & \cmark
        & $0.1556$ & $0.0108$ & $0.0074$ & $0.0138$ & $0.0067$ & $0.0229$ \\

        & U-NO \citeyearpar{UNO} 
        & \cmark
        & $0.1713$ & $0.0113$ & $0.0034$ & $0.0078$ & $0.0100$ & $0.0258$ \\

        & F-FNO \citeyearpar{F-FNO} 
        & \cmark
        & $0.2322$ & $0.0077$ & $0.0047$ & $0.0078$ & $0.0070$ & $0.0263$ \\

        & LSM \citeyearpar{LSM} 
        & \xmark
        & $0.1535$ & $0.0065$ & $0.0025$ & $0.0059$ & $0.0050$ & 0.0218 \\

         & \textit{FNO with modern backbone ( Ours)} 
        & \cmark
        &  $0.0976$ & $0.0068$ & $0.0066$ & $0.0065$ & $0.0069$ & $0.0205$\\

        & \textbf{HO-FNO (Ours)} 
        & \cmark
        & \cellcolor{orange!20} \textbf{0.0319} & \cellcolor{blue!10}$0.0052$ & \cellcolor{orange!20}\textbf{0.0006} & \cellcolor{blue!10}$0.0051$ & \cellcolor{violet!10}$0.0049$ & $0.0199$\\ 

        \midrule

                &\multicolumn{2}{c}{Increment over FNO with modern backbone} 
        & $\mathbf{67.3\%}$ 
        & $\mathbf{23.5\%}$ 
        & $\mathbf{90.9\%}$ 
        & $\mathbf{21.5\%}$ 
        & $\mathbf{29.0\%}$ 
        & $\mathbf{2.9\%}$ \\

        &\multicolumn{2}{c}{Increment over the best Frequency model} 
        & $\mathbf{67.3\%}$ 
        & $\mathbf{20.0\%}$ 
        & $\mathbf{76.0\%}$ 
        & $\mathbf{13.6\%}$ 
        & $\mathbf{2.0\%}$ 
        & $\mathbf{3.0\%}$ \\

        & \multicolumn{2}{c}{Increment over the best Transformer model} 
        & $\mathbf{55.2\%}$ 
        & $\mathbf{11.9\%}$ 
        & $\mathbf{40.0\%}$ 
        & $\mathbf{3.8\%}$ 
        & $-58.1\%$ 
        & $-485.3\%$ \\

        &\multicolumn{2}{c}{Increment over the best Res. Equiv. model} 
        & $\mathbf{79.5\%}$ 
        & $\mathbf{32.5\%}$ 
        & $\mathbf{82.4\%}$ 
        & $\mathbf{34.6\%}$ 
        & $\mathbf{12.5\%}$ 
        & $\mathbf{13.1\%}$ \\

		\bottomrule
		\end{tabular}}
		\end{sc}
	\end{small}
\end{table*}

\subsection{Extension to manifolds: SWE on the sphere.}
The classical Fourier transform is naturally defined on periodic euclidean domains, such as the torus $\mathbb{T}_d$. For functions defined on a manifold $\mathcal{M} \subset \mathbb{R}^D$, applying the standard Fourier transform after embedding in the ambient space ignores the intrinsic geometry of $\mathcal{M}$. A geometry-aware alternative is given by the spectral decomposition of the Laplace--Beltrami operator,
\begin{equation} \label{eq:eigen}
    -\Delta_g \phi_j = \lambda_j \phi_j \quad \text{on } \mathcal{M},
\end{equation}
whose eigenfunctions $\{\phi_j\}$ act as generalized Fourier modes and eigenvalues $\{\lambda_j\}$ as frequencies. This induces a natural notion of spectral convolution on manifolds, allowing Higher-Order Spectral Convolutions to extend beyond euclidean geometries. We refer to Appendix~\ref{app: extension to manifolds} for the extended discussion, and evaluate this setting on the spherical Shallow Water Equation in Table~\ref{tab:results_spherical}.
\paragraph{Results} HO-SFNO consistently outperforms SFNO on the spherical SWE benchmark across all metrics. The improvement is particularly strong for the full rollout, where the error decreases from $7.7 \times 10^{-1}$ to $7.0 \times 10^{-1}$, and for the early rollout interval $(0,10)$, where the error drops from $9.9 \times 10^{-2}$ to $8.0 \times 10^{-2}$. This suggests that higher-order spectral interactions improve both short-term prediction accuracy and long-horizon stability on non-Euclidean geometries.
\begin{table}[!h]
    \centering
    \small
    \vspace{-1em}
    \caption{\small Test performance on SWE dataset. Models are trained with MSE. We report variants of MSE for time intervals $(0, 10), \ (11, 25), \ (26, 50)$ and full rollout. Best per metric in \textbf{bold}.}
    \label{tab:results_spherical}
    \resizebox{0.8\textwidth}{!}{%
    \begin{tabular}{lcccccc}
        \toprule
        \textbf{Metric} & MSE & NRMSE & Rollout ($0:10$) & Rollout ($11:25$) & Rollout ($26:50$) & Rollout \\
        \midrule
        
        SFNO & $8.23$ & $1.7 \times 10^{-2}$ & $9.9 \times 10^{-2}$ & $3.0 \times 10^{-1}$ & $7.2 \times 10^{-1}$ & $7.7 \times 10^{-1}$ \\
        HO-SFNO (ours) & $\mathbf{5.56}$ & $\mathbf{1.3 \times 10^{-2}}$ & $\mathbf{8.0 \times 10^{-2}}$ & $\mathbf{2.6 \times 10^{-1}}$ & $\mathbf{6.2 \times 10^{-1}}$ & $\mathbf{7.0 \times 10^{-1}}$ \\
        
        \bottomrule
    \end{tabular}%
    }
\end{table}
\subsection{Efficiency Analysis}
A primary motivation for HO-FNO is to match transformers and SSMs performance while retaining FNO's efficiency. We provide a parameter count computation and we analyze computational and memory cost, measuring wall-clock time for inference, training and peak GPU memory usage.

Denoting by $C$ the number of input/output channels, $M$ the number of retained Fourier modes and $m$ the order of HO-FNO, the parameter count of an HO spectral layer is
$MC^{2} + mC^{2}$,
which grows linearly with the interaction order $m$. Since typically $m \ll M$ and an FNO convolution has $MC^2$ parameters, the additional parameters introduced by higher-order mixing are negligible in practice.
\paragraph{Results} The analysis~\ref{fig:efficiency_plot_normalized} shows that HO-FNO adds no visible overhead compared to FNO across the considered metrics, while remaining more efficient than competing models. Transolver is the most expensive method in all categories and benchmarks. Compared to LaMO, HO-FNO is substantially faster at inference across all settings, and more memory-efficient and faster to train on two benchmarks.
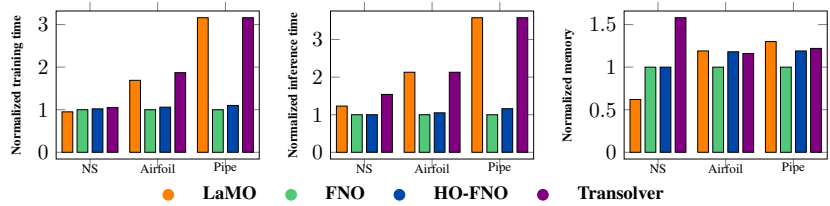
\begin{figure*}[!htb]
    \centering
    \setlength{\tabcolsep}{5pt}
    \definecolor{emerald}{RGB}{80,200,120}
    \definecolor{cobalt}{RGB}{0,71,171}

    \resizebox{0.8\textwidth}{!}{%
    \begin{tabular}{ccc}
    \begin{tikzpicture}
    \begin{axis}[
        width=4.9cm, height=4cm,
        ybar,
        bar width=5pt,
        ylabel=\tiny \textbf{Normalized training time},
        symbolic x coords={NS, Airfoil, Pipe},
        xtick=data,
        x tick label style={
            font=\tiny,
            rotate=0,
            anchor=center
        },
        ymin=0,
        enlarge x limits=0.25,
        enlarge y limits=0.05,
        yticklabel style={
            /pgf/number format/.cd,
            fixed,
            precision=1
        }
    ]
    \addplot[fill=orange] coordinates {
        (NS,0.95)
        (Airfoil,1.69)
        (Pipe,3.16)
    };

    \addplot[fill=emerald] coordinates {
        (NS,1.00)
        (Airfoil,1.00)
        (Pipe,1.00)
    };

    \addplot[fill=cobalt] coordinates {
        (NS,1.02)
        (Airfoil,1.06)
        (Pipe,1.10)
    };

    \addplot[fill=violet] coordinates {
        (NS,1.05)
        (Airfoil,1.87)
        (Pipe,3.16)
    };
    \end{axis}
    \end{tikzpicture}%
    &
    \begin{tikzpicture}
    \begin{axis}[
        width=4.9cm, height=4cm,
        ybar,
        bar width=5pt,
        ylabel=\tiny \textbf{Normalized inference time},
        symbolic x coords={NS, Airfoil, Pipe},
        xtick=data,
        x tick label style={
            font=\tiny,
            rotate=0,
            anchor=center
        },
        ymin=0,
        enlarge x limits=0.25,
        enlarge y limits=0.05,
        yticklabel style={
            /pgf/number format/.cd,
            fixed,
            precision=1
        }
    ]
    \addplot[fill=orange] coordinates {
        (NS,1.23)
        (Airfoil,2.13)
        (Pipe,3.58)
    };

    \addplot[fill=emerald] coordinates {
        (NS,1.00)
        (Airfoil,1.00)
        (Pipe,1.00)
    };

    \addplot[fill=cobalt] coordinates {
        (NS,1.00)
        (Airfoil,1.05)
        (Pipe,1.16)
    };

    \addplot[fill=violet] coordinates {
        (NS,1.54)
        (Airfoil,2.13)
        (Pipe,3.58)
    };
    \end{axis}
    \end{tikzpicture}%
    &
    \begin{tikzpicture}
    \begin{axis}[
        width=4.9cm, height=4cm,
        ybar,
        bar width=5pt,
        ylabel=\tiny \textbf{Normalized memory},
        symbolic x coords={NS, Airfoil, Pipe},
        xtick=data,
        x tick label style={
            font=\tiny,
            rotate=0,
            anchor=center
        },
        ymin=0,
        enlarge x limits=0.25,
        enlarge y limits=0.05,
        yticklabel style={
            /pgf/number format/.cd,
            fixed,
            precision=1
        }
    ]
    \addplot[fill=orange] coordinates {
        (NS,0.62)
        (Airfoil,1.19)
        (Pipe,1.30)
    };

    \addplot[fill=emerald] coordinates {
        (NS,1.00)
        (Airfoil,1.00)
        (Pipe,1.00)
    };

    \addplot[fill=cobalt] coordinates {
        (NS,1.00)
        (Airfoil,1.18)
        (Pipe,1.19)
    };

    \addplot[fill=violet] coordinates {
        (NS,1.58)
        (Airfoil,1.16)
        (Pipe,1.22)
    };
    \end{axis}
    \end{tikzpicture}%
    \end{tabular}%
    }

    \vspace{-0.2cm}

    \begin{tikzpicture}
    \begin{axis}[
        hide axis,
        legend columns=-1,
        axis lines=none,
        ticks=none,
        legend style={
            draw=none,
            column sep=2ex,
            font=\scriptsize
        },
        xmin=0, xmax=1, ymin=0, ymax=1
    ]
    \addlegendimage{only marks, mark=*, color=orange}
    \addlegendentry{\textbf{LaMO}}

    \addlegendimage{only marks, mark=*, color=emerald}
    \addlegendentry{\textbf{FNO}}

    \addlegendimage{only marks, mark=*, color=cobalt}
    \addlegendentry{\textbf{HO-FNO}}

    \addlegendimage{only marks, mark=*, color=violet}
    \addlegendentry{\textbf{Transolver}}
    \end{axis}
    \end{tikzpicture}%

    \caption{Efficiency comparison on NS, Airfoil, and Pipe after normalization with respect to FNO for each dataset and metric on a single Nvidia A100 GPU. Raw values are available in Appendix~\ref{app:wall_clock_times_and_memory_usage}.}
    \label{fig:efficiency_plot_normalized}
\end{figure*}
\section{Conclusion and Future Work}\label{sec:conclusions_and_future_work}
This paper introduces the Higher-Order Fourier Neural Operator (HO-FNO) for solving non-linear PDEs. 
Thanks to its HO-Spectral Convolution, a spectral mixer that mirrors polynomial interactions present in nonlinear PDEs, HO-FNO retains the efficiency of FFT-based methods while performing on par, or better, than more costly models such as Transformers and State-space models. 
Experimental results show state-of-the-art accuracy on standard PDE benchmarks and a remarkable depth efficiency in strongly non-linear dynamics. 
Future works include a specialized version on point clouds data to bridge the gap in that setting between frequency-based models and transformers, as well as exploring input-dependent conditioning of the higher-order kernels and multi-physics adaptations.

\section*{Acknowledgments}
This work was granted access to the HPC resources of IDRIS under the allocations AD010616824, AD011016345 and A0191016927 made by GENCI. This work has received support from the French government, managed by the
National Research Agency, under the France 2030 program with the reference “PR[AI]RIE-PSAI”
(ANR-23-IACL-0008) and "PEPR-SHARP" (ANR-23-PEIA-0008).

\clearpage
\newpage
\bibliographystyle{plainnat}
\bibliography{bibliography}

\clearpage
\newpage
\appendix
\begin{center}
    {\LARGE \textbf{Appendix}}
\end{center}

\startcontents[appendix]
\printcontents[appendix]{}{1}{}
\newpage

\section{Notations}
We summarize here the notations used throughout the paper.
\begin{table}[h!]
\centering
\renewcommand{\arraystretch}{1.1}
\begin{tabular}{p{3cm} p{11cm}}
\toprule
\textbf{Symbol} & \textbf{Meaning} \\
\midrule
$\Omega$ & Spatial domain. \\
$\partial\Omega$ & Boundary of the spatial domain. \\
$d$ & Number of spatial dimensions. \\
$[0,T]$ & Temporal domain. \\
$\mathbb{T}^d$ & $d$-dimensional torus, used as a periodic spatial domain. \\
$x\in\Omega$ & Spatial coordinate. \\
$t\in[0,T]$ & Time variable. \\
$u(x,t)\in\mathbb{R}$ & Scalar solution field evaluated at space--time point $(x,t)$. \\
$u^0$ & Initial condition, with $u(x,0)=u^0(x)$. \\
$p^0$ & Probability distribution from which initial conditions are sampled. \\
$\nu$ & PDE coefficients, such as diffusivity (Burger), viscosity (Navier-Stockes, or hyperdiffusion coefficients (SWE). \\
$\mathcal{B}(u)$ & Boundary condition operator. \\
\midrule
$\alpha=(\alpha_1,\dots,\alpha_d)$
& Spatial multi-index, where $\alpha_q$ is the derivative order in direction $x_q$. \\
$|\alpha|$ & Total order of the multi-index, $|\alpha|=\alpha_1+\cdots+\alpha_d$. \\
$D_x^\alpha u$ & Spatial derivative associated with $\alpha$:
$
D_x^\alpha u
=
\frac{\partial^{|\alpha|}u}
{\partial x_1^{\alpha_1}\cdots\partial x_d^{\alpha_d}}.
$ \\

$m$ & spatial order of the PDE. \\
\midrule
$\mathcal{G}^\dagger$ & Exact one-step solution operator mapping $u(\cdot,t)$ to $u(\cdot,t+\Delta t)$. \\
$\mathcal{G}_\theta$ & Learned neural operator that approximates $\mathcal{G}^\dagger$. \\
$\mathcal{L}$ & Pointwise lifting network. \\
$\mathcal{P}$ & Pointwise projection network. \\
$\mathcal{Q}_\ell$ & $\ell$-th learnable operator layer. \\
$L$ & Number of learnable operator layers. \\
$v_\ell\in L^2(\Omega;\mathbb{R}^{c_\ell})$ & Hidden representation at layer $\ell$, with $c_\ell$ latent channels. \\
$W_\ell$, $b_\ell$ & Pointwise linear map and bias term in layer $\ell$. \\
$\mathcal{K}_\ell$ & Nonlocal integral operator in layer $\ell$. \\
\midrule
$\mathcal{M}\subset\mathbb{Z}^d$ & Set of retained Fourier modes. \\
$\widehat{u}(k,t)\in\mathbb{C}$ & Fourier coefficient of $u(\cdot,t)$ at Fourier mode $k$ \\
$i$ & Imaginary unit, $i^2=-1$. \\
$\mathsf{T}_{\mathcal{M}}$ & Fourier truncation operator retaining only modes in $\mathcal{M}$. \\
$\mathcal{C}_{\theta_\ell}$ & Spectral convolution operator in FNO layer $\ell$. \\
$\kappa_{\theta_\ell}$ & Kernel of the convolution operator $\mathcal{C}_{\theta_\ell}$. \\
$R_\ell^{(k)}$ & Learned complex-valued matrix acting on the channel dimension of Fourier mode $k$ in layer $\ell$. \\
$N$ & Number of retained Fourier modes, $N=|\mathcal{M}|$. \\
\midrule
$j$ & Polynomial degree of a homogeneous component. \\
$n$ & Maximal polynomial degree of the PDE nonlinearity. \\
$\boldsymbol{\alpha}
=
(\alpha^{(1)},\dots,\alpha^{(j)})$ & Tuple of $j$ spatial multi-indices appearing in a degree-$j$ monomial. \\
$\mathcal{I}_j\subset(\mathbb{N}^d)^j$ & Set of multi-index tuples appearing in the degree-$j$ homogeneous component. \\
$\mathsf{P}_{\mathcal{I}_j,j}[u]$ & Degree-$j$ homogeneous component of the PDE right-hand side. \\
$c_{\boldsymbol{\alpha}}$ & Scalar coefficient associated with the monomial indexed by $\boldsymbol{\alpha}$. \\
$C_{\boldsymbol{\alpha}}(k_1,\dots,k_j)$ & Fourier-domain coefficient induced by the spatial derivatives in the monomial indexed by $\boldsymbol{\alpha}$. \\
\bottomrule
\end{tabular}
\end{table}

\newpage
\section{Comparison of Mode Mixing in FNO and HO-FNO} \label{app:additional_method_illustrations} 
In the following, we further detail the differences between the spectral convolution of FNO and the higher-order variants introduced in this work. 

\paragraph{Mode mixing in standard FNO.}
In the standard FNO layer, each retained Fourier mode is updated independently through a learned complex-valued linear transformation. Therefore, the spectral convolution acts diagonally across modes: the updated representation of mode $m_i$ only depends on the input representation at the same mode $m_i$. In Figure~\ref{fig:mode_mixing}, this corresponds to the left diagram, where each retained mode is mapped to itself without interacting with the other modes.

\paragraph{Higher-order mode mixing.}
In contrast, our higher-order spectral convolution allows the update of a retained mode to depend on multiplicative interactions between several input modes. The right diagram in Figure~\ref{fig:mode_mixing} shows an order-$2$ example: pairs of modes are multiplied together, and all pairwise contributions that point to the same output mode are aggregated to form its updated representation. This induces a richer spectral mixing mechanism, where information from non-retained or different retained modes can contribute to the evolution of a given retained mode through higher-order interactions.

\paragraph{Interpretation.}
This mechanism is particularly relevant for nonlinear PDEs, where products in physical space correspond to convolutions in Fourier space. As a result, nonlinearities naturally generate interactions between Fourier modes. While the standard FNO spectral convolution only applies a learned linear transformation to each retained mode, the proposed higher-order convolution explicitly parameterizes such mode interactions. This provides an architectural bias better aligned with nonlinear operators, while preserving the efficiency of working in a truncated Fourier representation.
\begin{figure}[!h]
    \centering
    \includegraphics[width=0.6\linewidth]{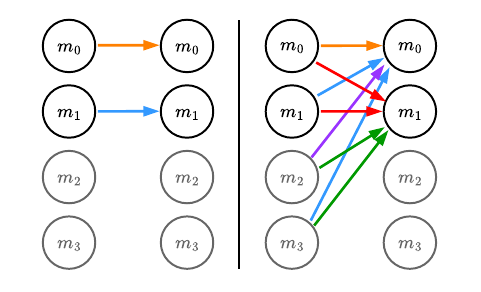}
    \caption{Example of mode mixing for a signal with $4$ Fourier modes, among which $2$ are retained. Diagonal modulation induced by the spectral convolution in FNO is compared with the proposed Higher Order mode mixing of order $2$. Arrows of the same color indicate modes that are multiplied together, and all contributions pointing to a given mode are aggregated to produce its updated representation.}
    \label{fig:mode_mixing}
    \vspace{-1em}
\end{figure}

\newpage

\section{Theoretical Results}\label{app:theoretical_results}
\subsection{A local inductive-bias induced by the polynomial Fourier structure}
\label{sec:inductive}

In this appendix, we provide further details on the notations from section~\ref{sec:setting}, and on the motivation for the design of HO-FNO. 

\paragraph{Writing a PDE as a sum of homogeneous components.} 
A homogeneous component of degree $j$ is the part of the PDE right-hand side whose terms are all polynomial monomials of total degree $j$ in $u$ and its spatial derivatives. 
More precisely, we may view $ u, \partial_{x_1}u, \partial_{x_2}u, \partial_{x_1x_1}u, \ldots$ as formal variables. 
A term is homogeneous of degree $j$ if it contains exactly $j$ factors of these variables.
For instance, the term $u\,\partial_x u$ is homogeneous of degree $2$, because it contains two factors: $u$ and $\partial_x u$. 
Similarly, $(\partial_x u)^2$ is also homogeneous of degree $2$, while $u^2\partial_x u$ is homogeneous of degree $3$. 
A linear diffusion term such as $\nu \Delta u$ is homogeneous of degree $1$, because it contains one factor, namely $\Delta u$, etc. 

Here we are interested in homogeneous components of degree $j$
\[
\mathsf{P}_{\mathcal{I}_j,j}[u](x,t)
=
\sum_{\boldsymbol{\alpha}\in\mathcal{I}_j}
c_{\boldsymbol{\alpha}}
\prod_{r=1}^{j}
D_x^{\alpha^{(r)}}u(x,t)
\]
i.e. each monomial contains exactly
$j$ factors of the form $D_x^{\alpha^{(r)}}u$, with $\alpha^{(r)}$ a multi-index.

Thus, in the decomposition
\[
\frac{\partial u}{\partial t}(x,t)
=
\sum_{j=0}^{n}
\mathsf{P}_{\mathcal{I}_j,j}[u](x,t),
\]
the right-hand side is grouped by polynomial degree:
with $\mathsf{P}_{\mathcal{I}_0,0}$ containing constant or forcing terms, $
\mathsf{P}_{\mathcal{I}_1,1}$ containing linear terms, 
 $\mathsf{P}_{\mathcal{I}_2,2}$ contains quadratic terms, and so on.

For example, in the viscous Burgers' equation
\[
\partial_t u
=
-\frac{1}{2}\partial_x(u^2)
+
\nu \partial_{xx}u
\]
the term
$
\nu \partial_{xx}u
$
is the degree-$1$ homogeneous component, while
$
-u\,\partial_x u
$
is the degree-$2$ homogeneous component. Therefore, the degree of
nonlinearity is $n=2$.

\paragraph{Motivating the design of HO-FNO.}
In Section~\ref{sec:setting}, we showed that if the PDE has degree of nonlinearity $n$, then the Fourier transform of its $n$-linear part takes the form
\begin{equation}\label{eq:FT_polynomial}
\widehat{\mathsf{P}_{\mathcal{I}_n,n}[u]}(k,t)
=
\sum_{\boldsymbol{\alpha}\in\mathcal{I}_n}
\sum_{\substack{k_1+\cdots+k_n=k\\ k_1,\dots,k_n\in\mathbb{Z}^d}}
C_{\boldsymbol{\alpha}}(k_1,\dots,k_n)
\prod_{r=1}^{n}
\widehat{u}(k_r,t),
\end{equation}
with
\[
C_{\boldsymbol{\alpha}}(k_1,\dots,k_n)
=
c_{\boldsymbol{\alpha}}
\prod_{r=1}^{n}
(ik_r)^{\alpha^{(r)}},
\qquad
(ik_r)^{\alpha^{(r)}}
:=
\prod_{q=1}^{d}
(ik_{r,q})^{\alpha^{(r)}_q}.
\]
 
In a more concise manner, we will write it as
\[
\mathsf{P}_{\mathcal I,n}(\widehat u)(k,t)
=
\sum_{k_1+\cdots+k_n=k}
C\,\widehat u(k_1,t)\cdots \widehat u(k_n,t).
\] 

In this appendix, we formalize the corresponding inductive-bias statement for HO-FNO in the matched case $m=n$. More precisely, we compare two classes of \emph{local hidden-space degree-$n$ interaction operators}: a tied class induced by repeated use of the same latent linear map inside a pointwise nonlinearity, and the order-$n$ HO-FNO class induced by Eq.~\ref{eq:HO_FNO_mixing}.

We show in the following that at the level of the degree-$n$ hidden interaction, an order-$n$ HO-FNO layer explicitly realizes the same factorwise $n$-linear Fourier mixing pattern as the PDE, while a tied pointwise model can only generate degree-$n$ interactions through repeated application of the same latent linear map. Therefore, if the target degree-$n$ interaction belongs to the HO-FNO class but is at positive distance from the tied class, then HO-FNO matches it exactly while tied pointwise models incur a non-vanishing local approximation error of order $\rho^n$ on balls of radius $\rho$.

\subsection{Motivation from the local expansion of the solution operator}

We keep the notation of Section~\ref{sec:setting}. Let
\[
\mathcal G^\dagger : L^2(\Omega;\mathbb R) \to L^2(\Omega;\mathbb R)
\]
denote the one-step solution operator introduced in Eq.~\eqref{eq:operator}. 
Since the analysis of Section~\ref{sec:setting} is carried out on the torus $\Omega=\mathbb T^d$, assuming the same number of retained modes $M$ in each spatial dimension, we introduce a Fourier truncation level $M\in\mathbb N$ and define the set of retained modes
\[
\mathcal K_M := \{k\in\mathbb Z^d : |k|\le M\}.
\]
For $u$, we write
\[
\widehat u_M := \big(\widehat u(k)\big)_{k\in \mathcal K_M}\in \mathbb C^{|\mathcal K_M|},
\]
and we denote by $\Pi_M$ the projection onto these retained modes. We define the Fourier-truncated one-step operator
\[
\mathcal G^\dagger_M := \Pi_M \circ \mathcal G^\dagger \circ \Pi_M .
\]

We assume that $\mathcal G^\dagger_M$ is analytic in a neighborhood of the origin. Hence it admits a local expansion
\begin{equation}
\mathcal G^\dagger_M(\widehat u_M)
=
\mathcal J_1(\widehat u_M)
+\mathcal J_2(\widehat u_M)
+\cdots
+\mathcal J_n(\widehat u_M)
+\mathcal R_{n+1}(\widehat u_M),
\label{eq:local_expansion_Gdagger}
\end{equation}
where:
\begin{itemize}
    \item each $\mathcal J_r$ is homogeneous of degree $r$,
    \item $\|\mathcal R_{n+1}(\widehat u_M)\| = O(\|\widehat u_M\|^{n+1})$ as $\widehat u_M\to 0$.
\end{itemize}

This expansion motivates isolating the degree-$n$ interaction. The theorem below, is stated at the hidden-space level, where the HO-FNO layer acts and where the operator classes are defined.

\subsection{From the scalar PDE notation to the hidden representation}

The PDE-side analysis in Section~\ref{sec:setting} is written for scalar functions
\[
u : \mathbb T^d \to \mathbb R.
\]
The HO-FNO architecture, acts on hidden representations. Following Section~\ref{sec:setting}, we denote by
\[
G_\theta = P \circ Q_L \circ \cdots \circ Q_1 \circ L
\]
the neural operator, where $L$ and $P$ are the lifting and projection networks, and
\[
Q_\ell(v_{\ell-1})=\sigma(W_\ell v_{\ell-1} + K_\ell(v_{\ell-1})+b_\ell)
\]
is the $\ell$-th operator layer. Next we drop the $\ell$ index. 

We therefore formulate the local degree-$n$ comparison at the hidden level. Let
\[
\widehat v(k)\in \mathbb C^C,
\qquad k\in \mathcal K_M,
\]
denote the Fourier coefficients of a hidden representation $v$ with $C$ channels.

Motivated by Eq.~\eqref{eq:FT_polynomial}, we define a target degree-$n$ hidden interaction operator $\mathcal T_n$ by
\begin{equation}
(\widehat{\mathcal T_n v})(k)
=
\sum_{k_1+\cdots+k_n=k}
T_k\big(\widehat v(k_1),\dots,\widehat v(k_n)\big),
\label{eq:target_hidden_operator}
\end{equation}
where, for each retained mode $k\in\mathcal K_M$,
\[
T_k : (\mathbb C^C)^n \to \mathbb C^C
\]
is an $n$-linear map.

Equation~\eqref{eq:target_hidden_operator} is the channel-valued analogue of Eq.~\eqref{eq:FT_polynomial}: the output mode $k$ is obtained by combining all $n$-tuples of input modes whose indices sum to $k$.

\subsection{Two degree-$n$ coefficient-family classes}

We now define the two classes that will be compared. We formulate them as classes of coefficient families
\[
\mathbf T = \{T_k\}_{k\in \mathcal K_M},
\]
where each $T_k : (\mathbb C^C)^n \to \mathbb C^C$ is $n$-linear.

\begin{definition}[Tied class induced by pointwise nonlinearities]
\label{def:tied_class_exact}
We define $\mathfrak C_n^{\mathrm{tied}}$ as the set of coefficient families $\mathbf T=\{T_k\}_{k\in\mathcal K_M}$ for which there exist
\[
A \in \mathbb C^{C\times C},
\qquad
W_k \in \mathbb C^{C\times C}
\quad (k\in\mathcal K_M),
\]
such that, for every $k\in\mathcal K_M$ and every $(z_1,\dots,z_n)\in(\mathbb C^C)^n$,
\begin{equation}
T_k(z_1,\dots,z_n)
=
W_k\big((Az_1)\odot\cdots\odot(Az_n)\big),
\label{eq:tied_class_exact}
\end{equation}
where $\odot$ denotes the channelwise (Hadamard) product.
\end{definition}

This class corresponds to the situation where the same latent linear map is reused in each factor of the degree-$n$ interaction. It is the hidden-space version of the tying phenomenon illustrated in Eq.~\ref{eq:eq_example}.

\begin{definition}[Order-$n$ HO-FNO class]
\label{def:ho_class_exact}
We define $\mathfrak C_n^{\mathrm{HO}}$ as the set of coefficient families $\mathbf T=\{T_k\}_{k\in\mathcal K_M}$ for which there exist
\[
A_1,\dots,A_n \in \mathbb C^{C\times C},
\qquad
W_k \in \mathbb C^{C\times C}
\quad (k\in\mathcal K_M),
\]
such that, for every $k\in\mathcal K_M$ and every $(z_1,\dots,z_n)\in(\mathbb C^C)^n$,
\begin{equation}
T_k(z_1,\dots,z_n)
=
W_k\big((A_1z_1)\odot\cdots\odot(A_n z_n)\big).
\label{eq:ho_class_exact}
\end{equation}
\end{definition}

This class is exactly the class induced by Eq.~\ref{eq:HO_FNO_mixing} of the paper in the matched case $m=n$.

\begin{remark}
By definition,
\[
\mathfrak C_n^{\mathrm{tied}} \subset \mathfrak C_n^{\mathrm{HO}},
\]
since the tied class is recovered from \eqref{eq:ho_class_exact} by imposing
\[
A_1=\cdots=A_n=A.
\]
\end{remark}

\subsection{Pointwise nonlinearities induce the tied class}

We now formalize the idea behind Section~\ref{sec:beyond}.

\begin{proposition}[Degree-$n$ term induced by a tied pointwise model]
\label{prop:pointwise_tied}
Let
\[
\sigma(z)=\sum_{r\ge 0}\alpha_r z^r
\]
be analytic near the origin. Consider the simplified local hidden map
\[
\Phi(v)=\sigma(A v),
\]
where $A$ acts channelwise and pointwise in physical space, and $\sigma$ is applied channelwise in the physical space. Then the coefficient family of the degree-$n$ homogeneous part of $\Phi$ belongs to the class $\mathfrak C_n^{\mathrm{tied}}$.
\end{proposition}

\begin{proof}
Expanding $\sigma$ at the origin gives
\[
\sigma(z)=\alpha_0+\alpha_1 z+\cdots+\alpha_n z^n+\cdots.
\]
Hence the degree-$n$ part of $\Phi$ is
\[
\Phi^{[n]}(v)=\alpha_n (A v)^{\odot n}.
\]
Passing to Fourier space, the pointwise product $(A v)^{\odot n}$ becomes an $n$-fold convolution over all tuples $(k_1,\dots,k_n)$ satisfying
\[
k_1+\cdots+k_n=k.
\]
Since the same linear map $A$ is reused in every factor, the resulting coefficient family is exactly of the form \eqref{eq:tied_class_exact}.
\end{proof}

\subsection{Order-$n$ HO-FNO realizes the HO class}

\begin{proposition}[Order-$n$ HO-FNO realizes $\mathfrak C_n^{\mathrm{HO}}$]
\label{prop:hofno_exact}
An HO-FNO layer of order $m=n$, as defined in Eq.~\ref{eq:HO_FNO_mixing}, induces a degree-$n$ hidden interaction operator whose coefficient family belongs to $\mathfrak C_n^{\mathrm{HO}}$.
\end{proposition}

\begin{proof}
Setting $m=n$ in Eq.~\ref{eq:HO_FNO_mixing} yields
\[
(\widehat{H_\theta v})(k)
=
W_k
\sum_{k_1+\cdots+k_n=k}
(A_1\widehat v(k_1))\odot\cdots\odot(A_n\widehat v(k_n)),
\]
which is precisely the form \eqref{eq:ho_class_exact}.
\end{proof}

\subsection{Distance to the tied class}

To compare the two classes, we measure the distance between the target degree-$n$ interaction $\mathcal T_n$ and the tied class.

Fix a basis of $\mathbb C^C$, and represent each $n$-linear map $T_k$ in Eq.~\eqref{eq:target_hidden_operator} by its coefficient array. We denote by
\[
\|\mathbf T\|_F^2 := \sum_{k\in\mathcal K_M}\|T_k\|_F^2
\]
the corresponding Frobenius norm of the coefficient family $\mathbf T := \{T_k\}_{k\in\mathcal K_M}$.

\begin{definition}[Distance to the tied class]
\label{def:distance_tied_exact}
Let $\mathbf T=\{T_k\}_{k\in\mathcal K_M}$ denote the coefficient family of the target operator $\mathcal T_n$. We define
\[
\delta_n(\mathbf T)
:=
\inf_{\widetilde{\mathbf T}\in \mathfrak C_n^{\mathrm{tied}}}
\|\mathbf T-\widetilde{\mathbf T}\|_F.
\]
\end{definition}

If $\delta_n(\mathbf T)>0$, then the target degree-$n$ interaction cannot be represented by any tied-factor coefficient family of the form induced by Proposition~\ref{prop:pointwise_tied}.

\subsection{Main theorem}

We now state the local hidden-space approximation result.

\begin{theorem}[Local degree-$n$ separation at the hidden level]
\label{thm:aligned_main}
Let $\mathcal T_n$ be the target degree-$n$ hidden interaction operator defined in Eq.~\ref{eq:target_hidden_operator}, with coefficient family $\mathbf T=\{T_k\}_{k\in\mathcal K_M}$.

Assume:
\begin{enumerate}
    \item $\mathbf T \in \mathfrak C_n^{\mathrm{HO}}$;
    \item $\delta_n(\mathbf T)>0$.
\end{enumerate}

Then:

\begin{enumerate}
    \item there exists an order-$n$ HO-FNO layer $H_\theta$ such that
    \[
    H_\theta^{[n]} = \mathcal T_n;
    \]

    \item there exists a constant $c_n>0$ such that for every coefficient family $\widetilde{\mathbf T}\in \mathfrak C_n^{\mathrm{tied}}$, with associated degree-$n$ hidden operator $\widetilde{\mathcal T}_n$, one has
    \[
    \sup_{\|v\|\le \rho}
    \|\mathcal T_n(v)-\widetilde{\mathcal T}_n(v)\|
    \ge
    c_n\,\delta_n(\mathbf T)\,\rho^n,
    \qquad \forall \rho>0.
    \]
\end{enumerate}
\end{theorem}

\begin{proof}
Since $\mathbf T\in \mathfrak C_n^{\mathrm{HO}}$, there exist matrices $A_1,\dots,A_n$ and $W_k$ such that the coefficient family of $\mathcal T_n$ has the form \eqref{eq:ho_class_exact}. By Proposition~\ref{prop:hofno_exact}, an HO-FNO layer of order $n$ realizes exactly this degree-$n$ interaction. This proves the first statement.

For the second statement, define on the finite-dimensional coefficient space the norm
\[
\|\mathbf S\|_{*}
:=
\sup_{\|v\|\le 1}\|\mathcal S_n(v)\|,
\]
where $\mathcal S_n$ denotes the degree-$n$ homogeneous hidden operator associated with the coefficient family $\mathbf S$. This is a norm because $\|\mathbf S\|_*=0$ implies that the homogeneous polynomial map $\mathcal S_n$ vanishes identically, hence all coefficients of $\mathbf S$ are zero. Since the coefficient space is finite-dimensional, there exists $c_n>0$ such that
\[
\|\mathbf S\|_{*}\ge c_n\|\mathbf S\|_F
\qquad\text{for all coefficient families }\mathbf S.
\]

Now let $\widetilde{\mathbf T}\in \mathfrak C_n^{\mathrm{tied}}$ and let $\widetilde{\mathcal T}_n$ be its associated operator. By homogeneity of degree $n$,
\[
\sup_{\|v\|\le \rho}\|\mathcal T_n(v)-\widetilde{\mathcal T}_n(v)\|
=
\rho^n
\sup_{\|v\|\le 1}\|\mathcal T_n(v)-\widetilde{\mathcal T}_n(v)\|.
\]
Therefore
\[
\sup_{\|v\|\le \rho}\|\mathcal T_n(v)-\widetilde{\mathcal T}_n(v)\|
=
\rho^n \|\mathbf T-\widetilde{\mathbf T}\|_{*}
\ge
c_n \rho^n \|\mathbf T-\widetilde{\mathbf T}\|_F.
\]
Taking the infimum over all $\widetilde{\mathbf T}\in \mathfrak C_n^{\mathrm{tied}}$ yields
\[
\sup_{\|v\|\le \rho}\|\mathcal T_n(v)-\widetilde{\mathcal T}_n(v)\|
\ge
c_n\,\delta_n(\mathbf T)\,\rho^n.
\]
\end{proof}

\subsection{Conditional consequence for a local model}

Theorem~\ref{thm:aligned_main} is stated at the hidden-space level. To translate it into an end-to-end statement, one needs an additional modeling assumption.

\begin{corollary}[Conditional end-to-end consequence]
\label{cor:conditional_end_to_end}
Assume that, after lifting, the local degree-$n$ interaction of the model is given by $\mathcal T_n$, and that all lower-order terms in the local expansion are absent or already matched. Then:
\begin{itemize}
    \item if $\mathbf T\in \mathfrak C_n^{\mathrm{HO}}$, replacing the degree-$n$ interaction by an order-$n$ HO-FNO layer eliminates the leading degree-$n$ mismatch;
    \item if the degree-$n$ interaction is constrained to lie in $\mathfrak C_n^{\mathrm{tied}}$, the leading mismatch is bounded below by $c_n\delta_n(\mathbf T)\rho^n$ on balls of radius $\rho$.
\end{itemize}
\end{corollary}

\subsection{Interpretation}
Section~\ref{sec:setting} establishes that the degree-$n$ part of the PDE gives rise, in Fourier space, to an $n$-linear convolution of modes. Section~\ref{sec:beyond} shows that pointwise nonlinearities can only recover higher-order interactions through tied weight sharing across the different factors. In contrast, Eq.~\ref{eq:HO_FNO_mixing} shows that an HO-FNO layer with order $m=n$ explicitly parameterizes factorwise $n$-linear mode mixing. Therefore, in the matched case $m=n$, HO-FNO is algebraically aligned with the degree-$n$ Fourier interaction pattern of the PDE, while tied pointwise models can only represent a restricted subclass of these interactions.

\begin{remark}[Quadratic case]
For $n=2$, the theorem recovers the triadic setting emphasized in Section~\ref{sec:navier}: the relevant interactions occur on tuples $(p,q,k)$ satisfying $p+q=k$. In that case, the result should be read as a precise structural-alignment statement. For $n\ge 3$, the same statement applies to genuinely higher-order mode interactions.
\end{remark}

\section{Extended Derivation of Fourier Mixing in Navier--Stokes}
\label{app:extended_navier}

We present here a detailed discussion of the non-linear interactions on the incompressible Navier--Stokes equation. 

The Incompressible Navier--Stokes equation is tipically presented in the following form:

\begin{align}
    \partial_t w(x,t) +u(x,t) \cdot \nabla w(x,t) &= \nu\Delta w(x,t) + f(x) \quad && x \in (0,1)^2, \ t \in (0, T]
    \\ \nabla \cdot u(x,t) &= 0 \quad &&x \in (0,1)^2, \ t \in [0,T]
    \\ w(x, 0) &= w_0(x)  \quad &&x \in (0,1)^2
\end{align}
Where $\nabla w(x,t) = \big(\partial_{x_1} w(x,t), \ \partial_{x_2} w(x,t)\big)$ is the gradient of $w$, $\Delta w(x,t) = \partial_{x_1x_1} w(x,t) + \partial_{x_2x_2} w(x,t)$ is the Laplacian of $w$ and $\nabla \cdot u = \frac{\partial u_1(x, t)}{\partial x_1} + \frac{\partial u_2(x, t)}{\partial x_2}$ is the divergence of $u$. 
$u(x, t)$ is the velocity at the point $x$ at time $t$ and $w$ is the vorticity field $w(x,t) = \partial_{x_1} u_2(x,t) - \partial_{x_2}u_1(x, t)$.

\paragraph{From the velocity to the vorticity formulation}
Firstly we will express the PDE in terms of the sole vorticity $w$. To do so we need to express $u$ in function of $w$. By the incompressibility condition $\nabla \cdot u = 0$ implies that exists a function, called streamfunction, $\psi = \psi(x, t)$ such that $u = \nabla^{\perp} \psi = \left(-\frac{\partial \psi}{\partial x_2}, \frac{\partial \psi}{\partial x_1}\right)$, therefore, by substitution we obtain $w$ in function of the stream function
\begin{equation}
    w  = \partial_{x_1}u_2 - \partial_{x_2}u_1 = \partial_{x_1}\left( \frac{\partial \psi}{\partial x_1} \right) + \partial_{x_2} \left( \frac{\partial \psi}{\partial x_2} \right) = \Delta \psi
\end{equation}
Therefore $\psi$ is obtained from $w$ by solving the Poisson problem $\Delta \psi = w$ in $(0,1)^2$ with appropriate boundary conditions. 
Once $\psi$ is founded, the velocity $u$ is recovered by $u = \nabla^\perp \psi$ and since $w = \Delta \psi$ we can write $u$ in function of $w$ as $u = \nabla^\perp \Delta^{-1} w$ and same for Navier--Stokes equation:

\begin{equation}\label{eq: vorticity eq}
\begin{aligned}
    \partial_t(w) &= \nu \Delta w(x,t) - (\nabla^\perp \Delta^{-1} w) \cdot \nabla w(x,t) + f(x) \quad && x \in (0,1)^2, \ t \in (0, T]
    \\ \nabla \cdot \nabla^\top \Delta^{-1} w &= 0 \quad &&x \in (0,1)^2, \ t \in [0,T]
    \\ w(x, 0) &= w_0(x)  \quad &&x \in (0,1)^2
\end{aligned}
\end{equation}

\paragraph{Fourier Transform of the Navier--Stokes equation}
Now we take the Fourier transform of the vorticity version of the Navier Stokes equation, by taking in consideration that $\widehat{\nabla w}(k, t) = 2 \pi i k \cdot \widehat{w}(k, t)$, $\widehat{\Delta w}(k, t) = -(2\pi)^2 \lvert k \rvert^2 \widehat{w}(k,t)$ and $\widehat{w \odot w} = \sum_{q + p = k} \widehat{w}(q, t) \widehat{w}(p, t)$. Therefore equation ~\ref{eq: vorticity eq} becomes

\begin{equation}
    \partial_t(\widehat{w})(k, t) = - \nu (2\pi)^2 \lvert k \rvert^2 \widehat{w}(k, t) - \sum_{p + q = k} \frac{(p + q) \cdot p^\perp}{\lvert p \rvert^2} \widehat{w}(p, t) \widehat{w}(q, t) + \widehat{f}(k, t) 
\end{equation}
For $k \in \mathbb{Z}^2, \ t \in (0, T]$.

\newpage

\section{Datasets} \label{app:datasets}
We report in this section a detailed description of the benchmarks adopted in this works. 

Table~\ref{tab:datasets_summary} summarizes the datasets considered in our experiments.
Each dataset is characterized by the number of trajectories and timesteps, the mesh type, the physical
problem dimension, and the resolution.

\renewcommand{\arraystretch}{1.2}
\begin{table}[!h]
\centering
\caption{Benchmark PDE datasets used in our experiments.}
\resizebox{\textwidth}{!}{%
\begin{tabular}{lllll}
\toprule
\textbf{Dataset Name} & \textbf{\# Trajectories} & \textbf{\# Timesteps} & \textbf{Mesh Type} & \textbf{Resolution}  \\
\midrule
\shortstack[l]{Navier--Stokes ($2$D) \\ \scriptsize $\nu = 10^{-3}$} & $5000$ & $50$ & Regular (2D periodic box) & $(64 \times 64, 1)$ \\
\shortstack[l]{Navier--Stokes ($2$D) \\ \scriptsize $\nu = 10^{-4}$} & $10000$ & $30$ & Regular (2D periodic box) & $(64 \times 64, 1)$ \\
\shortstack[l]{Navier--Stokes ($2$D) \\ \scriptsize $\nu = 10^{-5}$} & $1200$ & $20$ & Regular (2D periodic box) & $(64 \times 64, 1)$ \\
Darcy flow & $1200$ & $1$ & Regular & $(85 \times 85, 1)$ \\
Plasticity & $980$ & $21$ & Structured & $(101 \times 31, 2)$ \\
Airfoil & $1100$ & $1$ & Structured & $(221 \times 51, 2)$ \\
Pipe & $1200$ & $1$ & Structured & $(129 \times 129, 2)$ \\
Elasticity & $1200$ & $1$ & $2$D Point Cloud & $(972, 2)$ \\
PlanetSWE ($2$D) & $50$ & $100$ & Sphere (latitude-longitude grid) & $256 \times 128$ \\
\shortstack[l]{Polynomial-Source Poisson \\ \scriptsize $p \in \{1,2,3,5\}$} & $1200$ & $1$ & Regular (2D periodic box) & $64 \times 64$ \\
\bottomrule
\end{tabular}%
}
\label{tab:datasets_summary}
\end{table}

\subsection{$2$D Navier Stokes equations~\citep{FNO}} \label{app:navier}
The $2$D Navier--Stokes equation for a viscous, incompressible fluid in vorticity form on the
unit torus:

\begin{align}
    \partial_t w(x,t) +u(x,t) \cdot \nabla w(x,t) &= \nu\Delta w(x,t) + f(x) \quad && x \in (0,1)^2, \ t \in (0, T]
    \\ \nabla \cdot u(x,t) &= 0 \quad &&x \in (0,1)^2, \ t \in [0,T]
    \\ w(x, 0) &= w_0(x) \quad x \in (0,1)^2 \quad &&x \in (0,1)^2
\end{align}
The initial condition $w_0(x)$ is generated according to $w_0 \sim \mu$ where 
\begin{equation}
    \mu = \mathcal{N}(0, 7^{3/2}(-\Delta + 49I)^{-2.5})
\end{equation}
with periodic boundary conditions. The forcing is kept fixed:
\begin{equation}
    f(x) = 0.1(\sin(2\pi(x_1+x_2)) + \cos(2\pi(x_1+x_2)))
\end{equation}
The equation is solved using the stream-function formulation with a pseudospectral method. First a Poisson equation is solved in Fourier space to find the velocity field.
Then the vorticity is differentiated and the non-linear term is computed is physical space after which
it is dealiased. Time is advanced with a Crank–Nicolson update where the non-linear term does not
enter the implicit part. 

All data are generated on a $256 \times 256$ grid and are downsampled to $64 \times 64$. The Crank–Nicolson scheme used in the data-generation process has a time step of $10^{-4}$, with the solution recorded every $t=1$ time unit.

We use three datasets on Navier--Stokes equations, with viscosity $\nu  = 10^{-3}$, $\nu  = 10^{-4}$ and $\nu = 10^{-5}$, provided in ~\citep{FNO}.

\subsection{Darcy Flow~\citep{FNO}}
This benchmark represents the flow through porous media. 2D Darcy flow over a unit square
is given as follows:
\begin{equation}\label{eq:darcy}
    \begin{aligned}
        \nabla \cdot (a(x)\nabla u(x)) &= f(x), \quad x \in (0,1)^2 \\
        u(x) &= 0, \quad x \in \partial (0,1)^2
    \end{aligned}
\end{equation}
where $a(x)$ is the viscosity, $f(x)$ is the forcing term, and $u(x)$ is the solution. This dataset employs a constant value of forcing term $F(x)=\beta$. 
Further, Equation ~\ref{eq:darcy} is modified in the form of a temporal evolution as follows:
\begin{equation}
    \partial_t u(x,t) - \nabla \cdot (a(x) \nabla u(x,t)) = f(x), \quad x \in (0,1)^2, \label{eq:temporal}
\end{equation}
In this dataset, the input is represented by the parameter $a$, and the corresponding output is the solution $u$. The process is discretized into $421 \times 421$ regular grid and then downsampled to a resolution of $85 \times 85$ for main experiments. For training, 1000 samples are used, 200 samples are generated for testing, and different cases contain different medium structures.

\subsection{Plasticity~\cite{GEO-FNO} }
This benchmark consider the plastic forging benchmark introduced in the Geo-FNO dataset suite. The task models the deformation of a two-dimensional block of material under the action of a rigid die. The computational domain is a rectangle
\begin{align*}
    \Omega = [0,L] \times [0,H],
\end{align*}
where $L=50\,\mathrm{mm}$ and $H=15\,\mathrm{mm}$. At time $t=0$, the upper boundary of the material is impacted by a frictionless rigid die moving downward at a constant velocity $v=3\,\mathrm{m\,s^{-1}}$. The bottom edge of the block is clamped, while the displacement boundary condition induced by the die is imposed on the top edge.

The geometry of the die is parameterized by a function
\begin{align*}
    S_d \in H^1([0,L];\mathbb{R}),
\end{align*}
which describes the shape of the impacting surface. In the dataset, each die shape is sampled by drawing its values at a finite set of control points and interpolating them with cubic Hermite splines. More precisely, for each sample, the values
\begin{align*}
    S_d(x_k),
    \qquad
    x_k = \frac{kL}{7},
    \qquad
    k=0,\ldots,7,
\end{align*}
are drawn uniformly at random. The resulting spline defines the die geometry used in the finite-element simulation.

The material follows an elasto-plastic constitutive law. Let $\boldsymbol{\sigma}$ denote the stress tensor, $\boldsymbol{\epsilon}$ the strain tensor, and $\boldsymbol{\epsilon}_p$ the plastic strain tensor. The constitutive relation is given by
\begin{align*}
    \boldsymbol{\sigma}
    &=
    \mathbb{C} : (\boldsymbol{\epsilon} - \boldsymbol{\epsilon}_p),
    \\
    \dot{\boldsymbol{\epsilon}}_p
    &=
    \lambda \nabla_{\boldsymbol{\sigma}} f(\boldsymbol{\sigma}),
    \\
    f(\boldsymbol{\sigma})
    &=
    \sqrt{
        \frac{3}{2}
        \left\|
            \boldsymbol{\sigma}
            -
            \frac{1}{3}
            \operatorname{tr}(\boldsymbol{\sigma})\boldsymbol{I}
        \right\|_F^2
    }
    -
    \sigma_Y .
\end{align*}
Here, $\mathbb{C}$ is the isotropic stiffness tensor, $\lambda$ is the plastic multiplier, $\sigma_Y$ is the yield stress, $\operatorname{tr}(\boldsymbol{\sigma})$ denotes the trace of the stress tensor, $\boldsymbol{I}$ is the identity tensor, and $\|\cdot\|_F$ is the Frobenius norm. The plastic multiplier satisfies the classical complementarity conditions
\begin{align*}
    \lambda \geq 0,
    \qquad
    f(\boldsymbol{\sigma}) \leq 0,
    \qquad
    \lambda f(\boldsymbol{\sigma}) = 0.
\end{align*}
The physical parameters are fixed to Young's modulus $E=200\,\mathrm{GPa}$, yield strength $\sigma_Y=70\,\mathrm{MPa}$, and mass density $\rho^s = 7850\,\mathrm{kg\,m^{-3}}$.

This benchmark is particularly challenging because the underlying constitutive law is not polynomial in the state variables. Indeed, the yield function contains the Frobenius norm of the deviatoric stress tensor,
\begin{align*}
    \boldsymbol{\sigma}
    -
    \frac{1}{3}
    \operatorname{tr}(\boldsymbol{\sigma})\boldsymbol{I},
\end{align*}
and the plastic flow rule depends on the gradient
\begin{align*}
    \nabla_{\boldsymbol{\sigma}} f(\boldsymbol{\sigma}).
\end{align*}
Therefore, the dynamics involve both a square-root nonlinearity and a normalization-type operation induced by the derivative of the norm. This makes the solution operator fundamentally different from the polynomial nonlinearities considered in our controlled Poisson datasets. In particular, the operator cannot be represented exactly by a finite-degree polynomial expansion of the input fields.

Moreover, the elasto-plastic response introduces nonsmooth behavior through the transition between elastic and plastic regimes. The active set
\begin{align*}
    \left\{
        \boldsymbol{x} \in \Omega
        :
        f(\boldsymbol{\sigma}(\boldsymbol{x})) = 0
    \right\}
\end{align*}
separates regions where the material remains elastic from regions where plastic deformation occurs. Across this interface, the evolution law changes because the plastic multiplier $\lambda$ becomes active only when the yield constraint is saturated. As a consequence, the solution may contain sharp fronts and discontinuities, making this dataset a demanding test case for neural operators.

The data are generated with the commercial finite-element solver ABAQUS using $3000$ four-node quadrilateral bilinear elements. The dataset contains $900$ training trajectories and $80$ test trajectories. Each sample consists of $20$ time steps on a structured spatial grid of resolution $101 \times 31$. The input to the learning problem is the die geometry $S_d$, while the target is the time-dependent deformation of the material, including the evolution of the mesh grid and the displacement field.

This benchmark therefore evaluates the ability of a neural operator to learn a time-dependent, geometry-conditioned, nonlinear solution operator arising from elasto-plastic mechanics. Compared with polynomial-source elliptic problems, the plastic forging dataset is substantially more complex: the governing law contains non-polynomial norm-based nonlinearities, gradient-dependent plastic flow, inequality constraints, and nonsmooth transitions between elastic and plastic regimes.

\subsection{Airfoil~\cite{GEO-FNO}}
This benchmark pertains to transonic flow over an airfoil. Due to the negligible viscosity of air, the viscous term $\nu \nabla^2U$ is omitted from the Navier--Stokes equation. Consequently, the governing equations for this scenario are expressed as follows:
\begin{equation}
    \frac{\partial \rho f}{\partial t} + \nabla \cdot (\rho f U) = 0
\end{equation}
\begin{equation}
    \frac{\partial (\rho f U)}{\partial t} + \nabla \cdot (\rho f UU + pI) = 0
\end{equation}
\begin{equation}
    \frac{\partial E}{\partial t} + \nabla \cdot ((E + p)U) = 0, \label{eq:Navier--Stokes}
\end{equation}

where $\rho f$ represents fluid density, and $E$ denotes total energy. The input shape is discretized into a structured mesh with dimensions $221 \times 51$, and the output represents the Mach number at each mesh point. All shapes are derived from the NACA-0012 case provided by the National Advisory Committee for Aeronautics. In the training, 1000 samples from various airfoil designs are used for training, while the remaining 200 samples are reserved for testing.

\subsection{Pipe~\cite{GEO-FNO}}
This benchmark aims to estimate the horizontal fluid velocity based on the structure of the pipe.  The governing equations are as follows:
\begin{equation}
    \nabla \cdot \mathbf{U} = 0,
\end{equation}
\begin{equation}
    \frac{\partial \mathbf{U}}{\partial t} + \mathbf{U} \cdot \nabla \mathbf{U} = \mathbf{f}^{-1} \frac{1}{\rho} \nabla p + \nu \nabla^2 \mathbf{U}. 
    \label{eq:governing}
\end{equation}

Each sample represents the pipe as a structured mesh with dimensions $129 \times 129$. The input tensor, shaped as $129 \times 129 \times 2$, encodes the position of each discretized mesh point. The output tensor, with a shape of $129 \times 129 \times 1$, provides the velocity value at each point. For training, $1000$ samples with varying pipe shapes are generated, while $200$ additional samples, created by altering the pipe’s centerline, are reserved for testing.

\subsection{Elasticity~\cite{GEO-FNO}}
The governing equation of a solid body is given by
\begin{align*}
    \rho^s \frac{\partial^2 \boldsymbol{u}}{\partial t^2}
    +
    \nabla \cdot \boldsymbol{\sigma}
    =
    0,
\end{align*}
where $\rho^s$ is the mass density, $\boldsymbol{u}$ is the displacement vector, and $\boldsymbol{\sigma}$ is the stress tensor. Constitutive models, which relate the strain tensor $\boldsymbol{\epsilon}$ to the stress tensor $\boldsymbol{\sigma}$, are required to close the system.

We consider the unit-cell problem
\begin{align*}
    \Omega = [0,1] \times [0,1],
\end{align*}
with an arbitrarily shaped void at the center of the domain. The prior on the void radius is
\begin{align*}
    r
    =
    0.2
    +
    \frac{0.2}{1+\exp(\tilde{r})},
    \qquad
    \tilde{r}
    \sim
    \mathcal{N}
    \left(
        0,
        4^2(-\Delta + 3^2)^{-1}
    \right),
\end{align*}
which enforces the constraint
\begin{align*}
    0.2 \leq r \leq 0.4.
\end{align*}
The unit cell is clamped on the bottom edge, and a tension traction
\begin{align*}
    \boldsymbol{t} = [0,100]
\end{align*}
is applied on the top edge.

The material is modeled as an incompressible Rivlin--Saunders material. Its constitutive law is given by
\begin{align*}
    \boldsymbol{\sigma}
    &=
    \frac{\partial w(\boldsymbol{\epsilon})}{\partial \boldsymbol{\epsilon}},
    \\
    w(\boldsymbol{\epsilon})
    &=
    C_1(I_1 - 3)
    +
    C_2(I_2 - 3),
\end{align*}
where $I_1$ and $I_2$ are scalar invariants of the right Cauchy--Green stretch tensor $\boldsymbol{C}$, defined as
\begin{align*}
    I_1
    &=
    \operatorname{tr}(\boldsymbol{C}),
    \\
    I_2
    &=
    \frac{1}{2}
    \left(
        \operatorname{tr}(\boldsymbol{C})^2
        -
        \operatorname{tr}(\boldsymbol{C}^2)
    \right),
    \\
    \boldsymbol{C}
    &=
    2\boldsymbol{\epsilon}
    +
    \boldsymbol{I}.
\end{align*}
The energy-density parameters are
\begin{align*}
    C_1 = 1.863 \times 10^5,
    \qquad
    C_2 = 9.79 \times 10^3.
\end{align*}

We use $1000$ training samples and $200$ test samples generated with a finite-element solver using approximately $100$ quadratic quadrilateral elements. Each simulation takes approximately $5$ CPU seconds. The inputs are represented as point clouds of size around $1000$, and the target output is the stress field.

\subsection{PlanetSWE~\citep{The_Well}}\label{app: PlanetSWE}
The rotated, hyperviscous, forced Shallow Water Equation (SWE) on a sphere is a classical test problem for dynamical systems cores to be used in large-scale weather and climate models as they capture a number of similar phenomena but are better understood and operate at a more practical scale \citep{williamson_standard_1992}. We used the forced hyperviscous equations in two dimensions:
\begin{align}
    \partial_t u(x, t) &= -u(x,t) \cdot \nabla_x u(x,t) - g \nabla_x h(x,t) - \nu \nabla_x^4u(x,t) - 2 \Omega \times u(x,t)
    \\ \partial_t h(x,t) &= -H \nabla_x \cdot u(x,t) - \nabla_x \cdot (h(x,t)u(x,t)) - \nu \nabla_x^4 h(x,t) + F(x,t)
\end{align}
where $\nu$ is the hyper-diffusion coefficient, $\Omega$ is the Coriolis parameter, $u$ is the velocity field, $H$ is the mean height, and $h$ denotes deviation from the mean height. 
$F$ is a daily/seasonally varying forcing with periods
of $24$ and $1008$ simulation “hour” respectively.  

Initial conditions are randomly sampled from ERA5\citep{ERA5}. $u,\ v,\ z$ are taken from the hpa $500$ level with $z$ used as
$h$ is the shallow water set-up. Prefiltering was performed by executing ten iterations
of $50$ steps followed by solving a balance BVP. 

The dataset we used was generated in \citep{planetSWE} and is part of The Well dataset \citep{The_Well}.

The simulations were performed using the spin-weighted
spherical harmonic spectral method in Dedalus \citep{dedalus} with $500$ simulation hours of burn-in where the next three simulation years ($3024$ hours), were collected for the data set.
Integration is performed forward in time using a semi-implicit RK2 integrator. Step-sizes are computed using
the CFL-checker in Dedalus. The $3/2$ rule is used for de-aliasing. Background orography is taken from earth
orography and passed through mean-pooling three times (until the simulations became stable empirically).
Hyperdiffusion is matched at $\ell = 96$.

The original dataset from The Well ~\citep{The_Well} contains
$120$ trajectories of $3024$, each consisting of $3024$ timesteps at a spatial resolution of $256 \times 512$. For faster training, we restricted our experiments to the first $50$ trajectories, truncated to the initial $100$ timesteps, and downsampled the spatial resolution to $256 \times 128$ by averaging.

\subsection{Proposed Poisson equation with polynomial source (Polynomial-Source Poisson)}
\label{app: poly poisson}

We introduce the \textbf{Polynomial-Source Poisson} dataset, a controlled nonlinear elliptic PDE benchmark designed to evaluate whether neural operators can efficiently represent multiplicative field interactions of increasing order. 

The dataset is parameterized by an integer $p \geq 1$. For each sample, the model receives $p$ independent input fields $u_1,\ldots,u_p$ defined on the two-dimensional periodic unit square and must predict the solution $v$ of a Poisson equation whose source term is the pointwise product of these input fields.

The dataset is motivated by a common structure in physics-based PDEs: nonlinear interactions between fields often appear as source terms, forcing terms, or closure terms\citep{multiple_conditioning_ref}, and these sources are then propagated through an elliptic, parabolic, or spectral operator \citep{quarteroni_spectral_book}. In particular, Poisson-type equations with nonlinear sources arise in nonlinear electrostatics\citep{poisson_boltzmann_electrostatic, poisson_boltzmann_biomolecular}, phase-field models\citep{poisson_phase_field, poisson_phase_field_more_recent_ref}, and pressure recovery in incompressible fluid dynamics\citep{poisson_fluid}. 

Our dataset isolates this mechanism in a simple setting: the nonlinearity is a degree-$p$ pointwise product, while the nonlocal part of the operator is the inverse Laplacian.

\paragraph{Domain and boundary conditions.}
All fields are defined on the periodic domain
\begin{align*}
    \Omega = [0,1]^2,
\end{align*}
with periodic boundary conditions in both spatial directions. We discretize $\Omega$ using a uniform grid of resolution $n \times n$, with $n=64$ in our experiments. Grid points are given by
\begin{align*}
    x_i = \frac{i}{n},
    \qquad
    y_j = \frac{j}{n},
    \qquad
    i,j = 0,\ldots,n-1.
\end{align*}
Each input sample is therefore represented as a tensor of shape $n \times n \times p$, and the corresponding target is represented as a tensor of shape $n \times n \times 1$.

\paragraph{PDE definition.}
For a fixed degree-$p$ multilinear source, the input consists of $p$ scalar fields
\begin{align*}
    u_1,\ldots,u_p : \Omega \to \mathbb{R}.
\end{align*}
The target field $v : \Omega \to \mathbb{R}$ is defined as the zero-mean periodic solution of
\begin{align*}
    -\Delta v(x,y)
    =
    \prod_{\ell=1}^{p} u_{\ell}(x,y)
    -
    \frac{1}{|\Omega|}
    \int_{\Omega}
    \prod_{\ell=1}^{p} u_{\ell}(x',y')
    \,dx'\,dy',
    \qquad (x,y) \in \Omega.
\end{align*}
Equivalently, if we define the nonlinear source term
\begin{align*}
    f_p(x,y)
    =
    \prod_{\ell=1}^{p} u_{\ell}(x,y),
\end{align*}
then the PDE is
\begin{align*}
    -\Delta v
    =
    f_p - \langle f_p \rangle,
\end{align*}
where
\begin{align*}
    \langle f_p \rangle
    =
    \frac{1}{|\Omega|}
    \int_{\Omega} f_p(x,y)\,dx\,dy.
\end{align*}
The subtraction of the spatial mean is required because the periodic Poisson equation is solvable only when the right-hand side has zero mean. To fix the additive constant, we impose the normalization
\begin{align*}
    \int_{\Omega} v(x,y)\,dx\,dy = 0.
\end{align*}
Thus, the dataset defines the operator
\begin{align*}
    \mathcal{G}_p :
    (u_1,\ldots,u_p)
    \mapsto
    v
    =
    (-\Delta)^{-1}
    \left(
        \prod_{\ell=1}^{p} u_{\ell}
        -
        \left\langle
            \prod_{\ell=1}^{p} u_{\ell}
        \right\rangle
    \right),
\end{align*}
where $(-\Delta)^{-1}$ denotes the inverse periodic Laplacian on zero-mean functions.

\paragraph{Fourier-space solution.}
The target is computed spectrally. Let $\widehat{f}_p(k)$ denote the discrete Fourier transform of the zero-mean source term, where $k=(k_1,k_2) \in \mathbb{Z}^2$ is a discrete Fourier mode. Since the eigenvalues of $-\Delta$ on the periodic unit square are
\begin{align*}
    4\pi^2 |k|^2
    =
    4\pi^2(k_1^2+k_2^2),
\end{align*}
the Fourier coefficients of the solution are given by
\begin{align*}
    \widehat{v}(k)
    =
    \frac{\widehat{f}_p(k)}
    {4\pi^2 |k|^2},
    \qquad k \neq 0,
\end{align*}
and
\begin{align*}
    \widehat{v}(0) = 0.
\end{align*}
The zero mode is set to zero to enforce the zero-mean normalization of $v$. The target field is then obtained by applying the inverse Fourier transform to $\widehat{v}$.

\paragraph{Input distribution.}
The input fields are sampled as independent band-limited random Fourier fields. For each sample and each input channel $\ell \in \{1,\ldots,p\}$, we draw a random periodic field of the form
\begin{align*}
    u_{\ell}(x,y)
    =
    \sum_{m=1}^{M}
    a_{\ell,m}
    \cos
    \left(
        2\pi
        \left(
            k_{\ell,m}^{(1)} x
            +
            k_{\ell,m}^{(2)} y
        \right)
        +
        \phi_{\ell,m}
    \right),
\end{align*}
where $M$ is the number of active Fourier modes. The frequencies
\begin{align*}
    k_{\ell,m}
    =
    \left(
        k_{\ell,m}^{(1)},
        k_{\ell,m}^{(2)}
    \right)
    \in \mathbb{Z}^2
\end{align*}
are sampled from a prescribed Fourier annulus
\begin{align*}
    k_{\min}
    \leq
    |k_{\ell,m}|
    \leq
    k_{\max}.
\end{align*}
In our main experiments, we use
\begin{align*}
    k_{\min}=8,
    \qquad
    k_{\max}=18,
    \qquad
    M=64.
\end{align*}
The phases are sampled independently as
\begin{align*}
    \phi_{\ell,m} \sim \operatorname{Uniform}(0,2\pi),
\end{align*}
and the amplitudes are sampled with a mild spectral decay,
\begin{align*}
    a_{\ell,m}
    =
    \frac{\xi_{\ell,m}}
    {|k_{\ell,m}|^{\alpha}},
    \qquad
    \xi_{\ell,m} \sim \mathcal{N}(0,1),
\end{align*}
with $\alpha=0.75$ unless otherwise stated.

After sampling, each input channel is normalized independently to have zero empirical mean and unit empirical standard deviation on the grid:
\begin{align*}
    u_{\ell}
    \leftarrow
    \frac{
        u_{\ell} - \operatorname{mean}(u_{\ell})
    }{
        \operatorname{std}(u_{\ell}) + \epsilon
    },
\end{align*}
where $\epsilon$ is a small numerical constant. This normalization ensures that the magnitude of the degree-$p$ product remains comparable across samples and prevents high-order variants such as $p=5$ from becoming numerically degenerate.

\paragraph{Data generation procedure.}
For each sample, the dataset is generated as follows.

\begin{enumerate}
    \item Sample $p$ independent band-limited random Fourier fields $u_1,\ldots,u_p$ on the periodic $n \times n$ grid.

    \item Normalize each input channel $u_{\ell}$ to zero mean and unit standard deviation.

    \item Form the degree-$p$ nonlinear source
    \begin{align*}
        f_p(x,y)
        =
        \prod_{\ell=1}^{p} u_{\ell}(x,y).
    \end{align*}

    \item Enforce the periodic Poisson solvability condition by removing the spatial mean:
    \begin{align*}
        \widetilde{f}_p(x,y)
        =
        f_p(x,y)
        -
        \operatorname{mean}(f_p).
    \end{align*}

    \item Solve the periodic Poisson equation
    \begin{align*}
        -\Delta v = \widetilde{f}_p
    \end{align*}
    using a Fourier-domain inverse Laplacian.

    \item Set the zero Fourier mode of $v$ to zero, which fixes the additive constant.

    \item Store $(u_1,\ldots,u_p)$ as the input and $v$ as the target.
\end{enumerate}

The target fields are optionally normalized using global training-set statistics:
\begin{align*}
    v
    \leftarrow
    \frac{
        v - \mu_{\operatorname{train}}
    }{
        \sigma_{\operatorname{train}} + \epsilon
    },
\end{align*}
where $\mu_{\operatorname{train}}$ and $\sigma_{\operatorname{train}}$ are computed over all target values in the training set. We use global target normalization rather than samplewise normalization, since samplewise normalization would remove meaningful amplitude information from the operator.

\paragraph{Dataset variants.}
The dataset naturally defines a family of tasks indexed by the polynomial degree $p$. In this work, we consider
\begin{align*}
    p \in \{2,3,5\}.
\end{align*}
The corresponding datasets are denoted \textbf{Polynomial-Source Poisson-2}, \textbf{Polynomial-Source Poisson-3}, and \textbf{Polynomial-Source Poisson-5}. The input and output shapes are therefore
\begin{align*}
    \text{input shape} &= 64 \times 64 \times p, \\
    \text{output shape} &= 64 \times 64 \times 1.
\end{align*}
Unless otherwise specified, each variant contains $1000$ training samples and $200$ test samples.

\paragraph{Why this dataset is challenging.}
The dataset is designed so that the difficulty is concentrated in the nonlinear source term. The inverse Laplacian is a linear, translation-equivariant, nonlocal operator that is diagonal in Fourier space. By contrast, the source
\begin{align*}
    f_p = u_1 u_2 \cdots u_p
\end{align*}
requires a pointwise multiplicative interaction among $p$ independent input fields. In Fourier space, this product corresponds to a $p$-fold convolution of the input spectra. Therefore, increasing $p$ increases the combinatorial complexity of the spectral interactions that must be represented.

This makes the dataset particularly suitable for evaluating models designed to capture higher-order interactions. A first-order architecture must synthesize the product indirectly through repeated linear transformations and pointwise nonlinearities. In contrast, an architecture with explicit order-$p$ multiplicative interactions can represent the dominant nonlinearity in a single stage, followed by a spectral operator approximating the inverse Laplacian.

The multi-channel construction is important. If the target source were simply $u^p$ for a single scalar input field, then a pointwise multilayer perceptron could approximate the scalar map $u \mapsto u^p$ locally. By using independent fields $u_1,\ldots,u_p$ the benchmark requires genuine cross-channel multiplicative coupling. This makes the task a cleaner diagnostic of whether the model can represent higher-order interactions between different input fields.

\paragraph{Role of the Fourier annulus.}
The input fields are not sampled only from very low frequencies. Instead, their spectra are concentrated in a medium-frequency annulus. This choice avoids an overly smooth regime in which the nonlinear product can be learned easily from low-dimensional correlations. 

At the same time, the input modes themselves lie below the Nyquist frequency to avoid making the task dominated by numerical aliasing. However, for high orders, e.g. $5$, the resulting source may contain unresolved modes unless de-aliasing or oversampling is used.

The product $u_1 \cdots u_p$ generates new frequencies through nonlinear mode coupling, while the inverse Laplacian smooths the resulting source according to the physically meaningful multiplier $1/|k|^2$.

The annular input distribution therefore creates a controlled but nontrivial spectral-learning problem: the model must represent a high-order nonlinear interaction in physical space and the corresponding nonlocal smoothing in Fourier space.

\begin{figure}[!h]
    \centering\includegraphics[width=1\columnwidth]{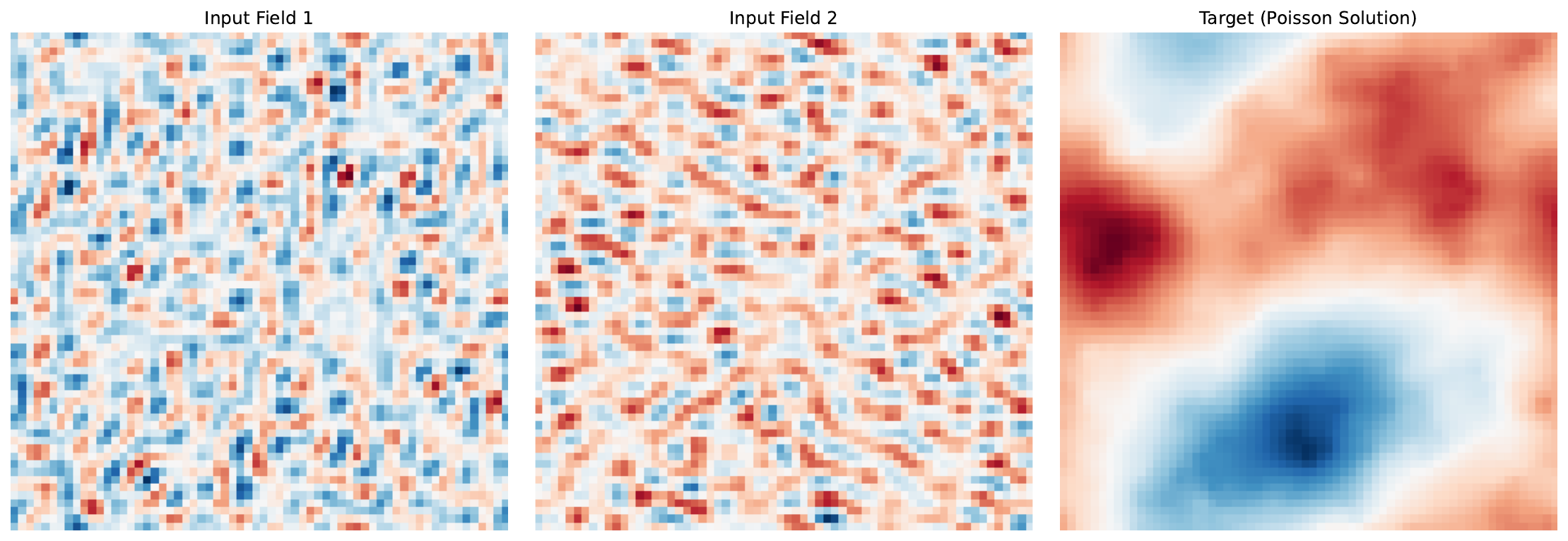}
    \caption{Sample from the Polynomial-Source Poisson dataset with $p=2$. The $2$ images on the left are the input fields, whose pointwise product defines the source term of the Poisson equation; the image on the right is the target Poisson solution.
    } 
    \vspace{-1em}
    \label{fig:backbones}
\end{figure}

\begin{figure}[!h]
    \centering\includegraphics[width=1\columnwidth]{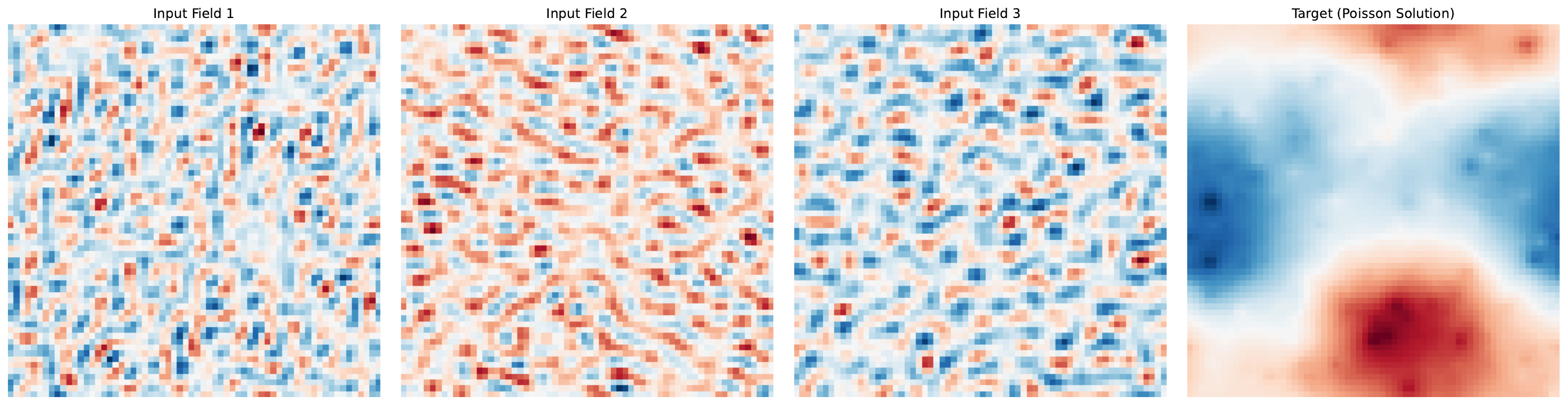}
    \caption{Sample from the Polynomial-Source Poisson dataset with $p=3$. The $3$ images on the left are the input fields, whose pointwise product defines the source term of the Poisson equation; the image on the right is the target Poisson solution.
    } 
    \vspace{-1em}
    \label{fig:backbones}
\end{figure}

\begin{figure}[!h]
    \centering\includegraphics[width=1.0\columnwidth]{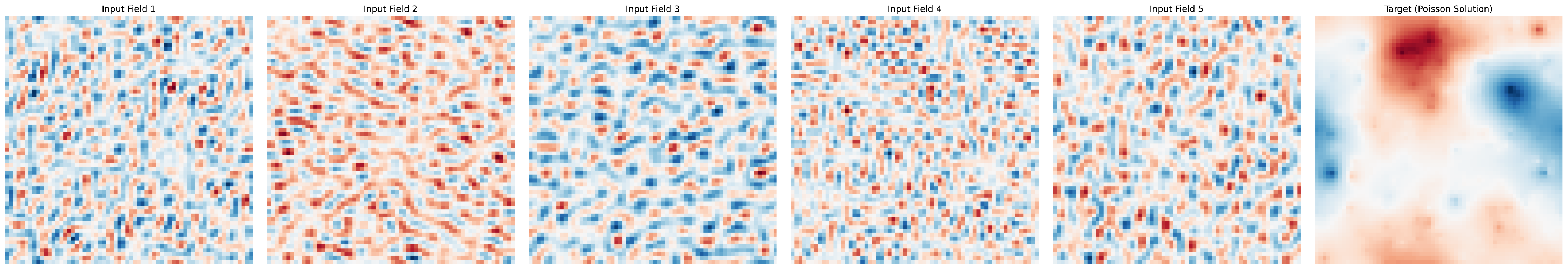}
    \caption{Sample from the Polynomial-Source Poisson dataset with $p=5$. The $5$ images on the left are the input fields, whose pointwise product defines the source term of the Poisson equation; the image on the right is the target Poisson solution.
    } 
    \vspace{-1em}
    \label{fig:backbones}
\end{figure}
\newpage
\section{Detailed discussion on the extension to manifolds: example of SWE on the sphere.}\label{app: extension to manifolds}
The classical Fourier transform is defined for functions defined on the torus $\mathbb{T}_d$. When a function is instead defined on a manifold $\mathcal{M} \subset \mathbb{R}^D$, one can still apply the classical Fourier transform by first extending the function to the ambient euclidean space $\mathbb{R}^D$. While this procedure makes the transform computable, the resulting representation ignores 
the geometry of the domain $\mathcal{M}$ of the function and therefore provides a sub-optimal representation. 
To overcome this limitation, the notion of a Fourier basis has been generalized to arbitrary compact Riemannian manifolds $\mathcal{M}$ through the spectral decomposition of the Laplace–Beltrami operator. Concretely, one considers the eigenvalue problem of Eq. \ref{eq:eigen}, where $\Delta_g f = \operatorname{div}_g(\nabla_g f)$ denotes the Laplacian, defined as the divergence of the Riemannian gradient.
\begin{equation} \label{eq:eigen}
-\Delta_{g}\phi_j = \lambda_j \phi_j \quad \text{on } \mathcal{M}
\end{equation}
The eigenfunctions ${\phi_j}$ serve as generalized Fourier modes, while the corresponding eigenvalues ${\lambda_j}$ play the role of frequencies.
\newline
For most manifolds, the eigenfunctions of the Laplace–Beltrami operator do not admit a closed-form expression and must be precomputed numerically \citep{NORM}. An important exception is the sphere, where the generalized Fourier modes correspond to the well-known \textit{spherical harmonics}. This extension of the Fourier transform naturally induces a corresponding notion of convolution, defined as a linear diagonal operator in the generalized Fourier domain. In the same spirit, Higher-Order Spectral Convolutions also extend to arbitrary geometries, and the theoretical framework developed in the classical Fourier setting remains directly applicable. 
We illustrate this by experimenting with the rotated, hyperviscous, forced Shallow Water Equation (SWE) on the sphere, with results reported in Table~\ref{tab:results_spherical}, and experimental settings described in Section \ref{app:xp_details}.

\begin{table}[!h] 
    \centering
    \small
    \vspace{-1em}
    \caption{Test performance on rotated, hyperviscous, forced Shallow Water Equation (SWE). Models are trained with MSE. We report variants of MSE for time intervals $(0, 10), \ (11, 25), \ (26, 50)$ and full rollout. Best per metric in \textbf{bold}.}
    \label{tab:results_spherical_app}
    \begin{tabular}{lcc}
        \toprule
        \textbf{Metric} & SFNO & HO-SFNO (ours) \\
        \midrule
        
        MSE & $8.23$ & $\mathbf{5.56}$ \\
        NRMSE & $1.7 \times 10^{-2}$ & $\mathbf{1.3 \times 10^{-2}}$ \\
        Rollout ($0:10$) & $9.9 \times 10^{-2}$ & $\mathbf{8.0 \times 10^{-2}}$ \\
        Rollout ($11:25$) & $3.0 \times 10^{-1}$ & $\mathbf{2.6 \times 10^{-1}}$ \\
        Rollout ($26:50$) & $7.2 \times 10^{-1}$ & $\mathbf{6.2 \times 10^{-1}}$ \\
        Rollout & $7.7 \times 10^{-1}$ & $\mathbf{7.0 \times 10^{-1}}$ \\
        
        \bottomrule
    \end{tabular}
\end{table}

\section{Ablation study on resolution-equivariance}\label{app:resolution_equivariance}
Resolution equivariance is a desirable architectural property in neural operator learning: since the underlying PDE operator is defined in the continuum, the learned model should transfer across different discretizations without retraining.
\newline
To assess whether HO-FNO preserves the resolution equivariance property of FNO, we train an order-$2$ HO-FNO on the Darcy flow dataset at resolution $200\times 200$, and evaluate the resulting model at test resolutions ranging from $50\times 50$ to $400\times 400$. Figure~\ref{fig:resolution_equivariance} shows the same trend for FNO and HO-FNO: models trained at $200\times 200$ exhibit no noticeable performance degradation when evaluated between $100\times 100$ and $400\times 400$, while performance decreases at $50\times 50$, where the discretization becomes too coarse and fine-scale information is lost.
\textit{This indicates that introducing higher-order mixing preserves the ability of FNO to generalize across discretizations.}

\begin{figure}[!h]
\centering
\includegraphics[width=0.6\columnwidth]{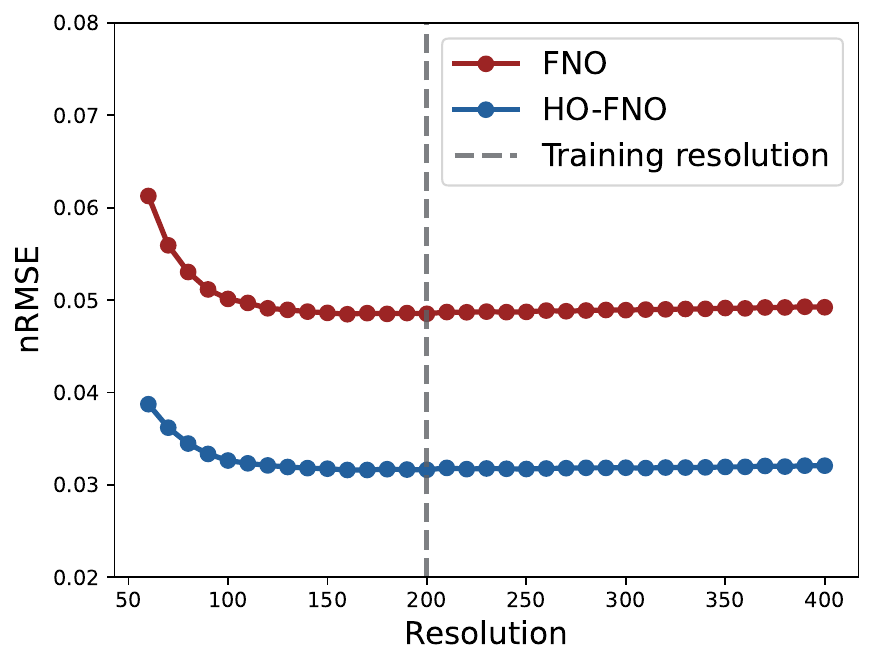}
\caption{Resolution equivariance of FNO and HO-FNO on the Darcy flow dataset.}
\label{fig:resolution_equivariance}
\vspace{-1em}
\end{figure}

\newpage
\section{Ablation study on the structure of the matrices added to induce higher order interactions}\label{app:structure_matrices_ablation}
We study the structure of the matrices $A_s$ used to induce polynomial interaction in our proposed convolution. Notably, this is the only addition to the mixer of the Fourier Neural Operator and it's therefore interesting to investigate different parametrizations.

\paragraph{Results} The dense parameterization of $A_s$ achieves the best performance, suggesting that more expressive matrices are beneficial for modeling higher-order interactions. However, all structured variants remain competitive, including diagonal, low-rank, and shared parameterizations. This indicates that the improvement does not only come from increasing the number of parameters, but from the higher-order mixing mechanism itself. In particular, even constrained matrix structures achieve losses comparable to strong frequency-based and Transformer baselines, supporting the role of higher-order spectral interactions as the main source of improvement.

\begin{table}[!h]
\centering
\caption{Structure of the matrices $A_s$ on the Darcy flow equation for HO-FNO of order 2 structure.}\label{fig:structure_matrices}
\begin{tabular}{lc}
\hline
Structure & loss  \\
\hline
Dense & $5.2 \times 10^{-3}$\\
identity  & $7.4 \times 10^{-3}$ \\
diagonal  &  $7.0 \times 10^{-3}$\\
low-rank (128) & $6.6 \times 10^{-3}$ \\
low-rank (64)  & $7.5 \times 10^{-3}$  \\
diagonal + low-rank (64)  & $7.2 \times 10^{-3}$  \\
diagonal + low-rank (128)  & $6.4 \times 10^{-3}$  \\
shared  & $6.7 \times 10^{-3}$  \\
\hline
\end{tabular}
\end{table}

\newpage
\section{Extended results on the Poisson equation with polynomial source}\label{app:extended_results_poisson}
We report in this section comprehensive results relative to the Poisson equation benchmark. Tables~\ref{tab:poisson-p1-mean-std-transposed}, ~\ref{tab:polynomial-source-poisson-p2},  ~\ref{tab:polynomial-source-poisson-p3} and  ~\ref{tab:polynomial-source-poisson-p5} reports MSE (i.e. $L^2$) and relative MSE (i.e. relative $L^2$) test performance for FNO and HO-FNO with order $m=1, \ 2, \ 3, \ 5$. For an easier quantitative assessment we also include line plots of MSE in Figure~\ref{fig: app Polynomial-Source Poisson_layers} 
in function of the number of layers. 

\paragraph{Baseline details} FNO employs here the architecture we refer to as the ``modern backbone'', that is in fact a pre-norm residual backbone typical of most modern transformers, that is getting standard today but was not at the time the original FNO was proposed. We underline that this architectural modification was done to ensure comparability on the results and isolate the contribution of the proposed higher-order spectral convolution from the backbone design.

\paragraph{Results} The HO-FNO variants consistently improve on the FNO baselines across all the variants of the dataset. In the Poisson datasets of order $3$ and $5$, a single layer of HO-FNO performs better than FNO models up to $16$ layers with a gap of, respectively, $1$ and $2$ orders of magnitude. On the Poisson dataset of order $p=2$, a single layer of HO-FNO performs better than FNO models up to $8$ layer while the $16$ layers' FNO performs slightly better than the $1$ layer HO-FNO. However, HO-FNO needs just $8$ layers to perform better than the two times bigger FNO. And also in this dataset, a variant of FNO (the one of order $5$), outperforms the $16$ layer FNO with a gap of an order of magnitude test MSE.

Interestingly, the experiments show that increasing the order of HO-FNO typically provide better performance, at exception of HO-FNO of order $5$ on the Poisson dataset of order $2$ when it has $4$ or more layers. This is due to excessive capacity compared to the problem complexity. We suggest to match the order of HO-FNO with the order of non-linearity of the physical system you aim to simulate for better performance, as suggested by the experiments in this section. 

In particular we notice that in the most challenging benchmark, when the forcing term has polynomial order $5$, only HO-FNO of order $5$ is able to provide good predictions. On this dataset, $p=5$, HO-FNO of order $2$ and $3$ improves slightly on FNO and have relative $L^2$ close to $1$ that is the same performance of the zero predictor (the model that predicts zero everywhere) meaning that all the models drastically failed except the HO-FNO variant of order $5$ that has a relative loss less than $10^{-1}$ for any number of layers in $\{1,2,4,8,16\}$.

\begin{table*}[!h]
\centering
\small
\caption{Polynomial-source Poisson $p=1$ results averaged over three runs.}
\label{tab:poisson-p1-mean-std-transposed}
\resizebox{\textwidth}{!}{%
\begin{tabular}{llcccc}
\toprule
Layers & Metric
& FNO (order 1)
& HO-FNO (order 2)
& HO-FNO (order 3)
& HO-FNO (order 5) \\
\midrule
\multirow{2}{*}{1 layer}
& Test MSE
& \cellcolor{blue!10}$0.00457 \pm 0.00059$
& $0.00664 \pm 0.00063$
& \cellcolor{orange!20}$\mathbf{0.00273 \pm 0.00022}$
& \cellcolor{violet!10}$0.00580 \pm 0.00058$ \\
& rel. L2
& \cellcolor{blue!10}$0.06932 \pm 0.00424$
& $0.08426 \pm 0.00400$
& \cellcolor{orange!20}$\mathbf{0.05359 \pm 0.00219}$
& \cellcolor{violet!10}$0.07833 \pm 0.00348$ \\
\midrule
\multirow{2}{*}{2 layers}
& Test MSE
& $0.00555 \pm 0.00022$
& \cellcolor{orange!20}$\mathbf{0.00175 \pm 0.00009}$
& \cellcolor{violet!10}$0.00217 \pm 0.00029$
& \cellcolor{blue!10}$0.00206 \pm 0.00006$ \\
& rel. L2
& $0.07406 \pm 0.00154$
& \cellcolor{orange!20}$\mathbf{0.04174 \pm 0.00101}$
& \cellcolor{violet!10}$0.04515 \pm 0.00288$
& \cellcolor{blue!10}$0.04471 \pm 0.00071$ \\
\midrule
\multirow{2}{*}{4 layers}
& Test MSE
& $0.00375 \pm 0.00048$
& \cellcolor{orange!20}$\mathbf{0.00229 \pm 0.00040}$
& \cellcolor{violet!10}$0.00294 \pm 0.00043$
& \cellcolor{blue!10}$0.00287 \pm 0.00040$ \\
& rel. L2
& $0.05931 \pm 0.00374$
& \cellcolor{orange!20}$\mathbf{0.04649 \pm 0.00434}$
& \cellcolor{violet!10}$0.05213 \pm 0.00413$
& \cellcolor{blue!10}$0.05119 \pm 0.00171$ \\
\midrule
\multirow{2}{*}{8 layers}
& Test MSE
& $0.00291 \pm 0.00025$
& \cellcolor{orange!20}$\mathbf{0.00168 \pm 0.00002}$
& \cellcolor{violet!10}$0.00282 \pm 0.00060$
& \cellcolor{blue!10}$0.00210 \pm 0.00021$ \\
& rel. L2
& $0.05231 \pm 0.00280$
& \cellcolor{orange!20}$\mathbf{0.04027 \pm 0.00085}$
& \cellcolor{violet!10}$0.04896 \pm 0.00541$
& \cellcolor{blue!10}$0.04323 \pm 0.00198$ \\
\midrule
\multirow{2}{*}{16 layers}
& Test MSE
& $0.00362 \pm 0.00029$
& \cellcolor{orange!20}$\mathbf{0.00269 \pm 0.00040}$
& \cellcolor{violet!10}$0.00347 \pm 0.00023$
& \cellcolor{blue!10}$0.00300 \pm 0.00042$ \\
& rel. L2
& $0.05887 \pm 0.00263$
& \cellcolor{orange!20}$\mathbf{0.05039 \pm 0.00413}$
& \cellcolor{violet!10}$0.05764 \pm 0.00157$
& \cellcolor{blue!10}$0.05358 \pm 0.00349$ \\
\bottomrule
\end{tabular}%
}
\end{table*}

\begin{table*}[!h]
\centering
\small
\caption{Polynomial-source Poisson $p=2$ results.}
\label{tab:polynomial-source-poisson-p2}
\resizebox{\textwidth}{!}{%
\begin{tabular}{llcccc}
\toprule
Layers & Metric & FNO (order 1) & HO-FNO (order 2) & HO-FNO (order 3) & HO-FNO (order 5) \\
\midrule
\multirow{2}{*}{1 layer}
& Test MSE & $0.09514 \pm 0.00333$ & \cellcolor{violet!10}$0.00073 \pm 0.00020$ & \cellcolor{blue!10}$0.00036 \pm 0.00005$ & \cellcolor{orange!20}$\mathbf{0.00024 \pm 0.00008}$ \\
& rel. L2  & $0.32724 \pm 0.00524$ & \cellcolor{violet!10}$0.02740 \pm 0.00375$ & \cellcolor{blue!10}$0.02025 \pm 0.00132$ & \cellcolor{orange!20}$\mathbf{0.01607 \pm 0.00277}$ \\
\midrule
\multirow{2}{*}{2 layers}
& Test MSE & $0.00207 \pm 0.00022$ & \cellcolor{blue!10}$0.00040 \pm 0.00011$ & \cellcolor{violet!10}$0.00058 \pm 0.00012$ & \cellcolor{orange!20}$\mathbf{0.00025 \pm 0.00008}$ \\
& rel. L2  & $0.04708 \pm 0.00264$ & \cellcolor{blue!10}$0.02063 \pm 0.00272$ & \cellcolor{violet!10}$0.02521 \pm 0.00295$ & \cellcolor{orange!20}$\mathbf{0.01597 \pm 0.00268}$ \\
\midrule
\multirow{2}{*}{4 layers}
& Test MSE & $0.00193 \pm 0.00049$ & \cellcolor{orange!20}$\mathbf{0.00033 \pm 0.00009}$ & \cellcolor{blue!10}$0.00049 \pm 0.00015$ & \cellcolor{violet!10}$0.00090 \pm 0.00073$ \\
& rel. L2  & $0.04455 \pm 0.00570$ & \cellcolor{orange!20}$\mathbf{0.01798 \pm 0.00255}$ & \cellcolor{blue!10}$0.02244 \pm 0.00398$ & \cellcolor{violet!10}$0.02538 \pm 0.00807$ \\
\midrule
\multirow{2}{*}{8 layers}
& Test MSE & $0.00243 \pm 0.00075$ & \cellcolor{orange!20}$\mathbf{0.00021 \pm 0.00002}$ & \cellcolor{blue!10}$0.00025 \pm 0.00004$ & \cellcolor{violet!10}$0.00121 \pm 0.00010$ \\
& rel. L2  & $0.04509 \pm 0.00722$ & \cellcolor{orange!20}$\mathbf{0.01465 \pm 0.00054}$ & \cellcolor{blue!10}$0.01613 \pm 0.00098$ & \cellcolor{violet!10}$0.03660 \pm 0.00173$ \\
\midrule
\multirow{2}{*}{16 layers}
& Test MSE & $0.00205 \pm 0.00022$ & \cellcolor{orange!20}$\mathbf{0.00029 \pm 0.00006}$ & \cellcolor{blue!10}$0.00110 \pm 0.00078$ & \cellcolor{violet!10}$0.00139 \pm 0.00068$ \\
& rel. L2  & $0.04690 \pm 0.00467$ & \cellcolor{orange!20}$\mathbf{0.01754 \pm 0.00167}$ & \cellcolor{blue!10}$0.02880 \pm 0.01212$ & \cellcolor{violet!10}$0.03254 \pm 0.00843$ \\
\bottomrule
\end{tabular}%
}
\end{table*}

\begin{table*}[!h]
\centering
\small
\caption{Polynomial-source Poisson $p=3$ results.}
\label{tab:polynomial-source-poisson-p3}
\resizebox{\textwidth}{!}{%
\begin{tabular}{llcccc}
\toprule
Layers & Metric & FNO (order 1) & HO-FNO (order 2) & HO-FNO (order 3) & HO-FNO (order 5) \\
\midrule
\multirow{2}{*}{1 layer}
& Test MSE & $0.72888 \pm 0.00601$ & \cellcolor{violet!10}$0.53714 \pm 0.03983$ & \cellcolor{blue!10}$0.00234 \pm 0.00023$ & \cellcolor{orange!20}$\mathbf{0.00059 \pm 0.00007}$ \\
& rel. L2  & $0.87913 \pm 0.00811$ & \cellcolor{violet!10}$0.75098 \pm 0.02441$ & \cellcolor{blue!10}$0.05028 \pm 0.00241$ & \cellcolor{orange!20}$\mathbf{0.02555 \pm 0.00156}$ \\
\midrule
\multirow{2}{*}{2 layers}
& Test MSE & $0.01802 \pm 0.00549$ & \cellcolor{violet!10}$0.00394 \pm 0.00089$ & \cellcolor{blue!10}$0.00174 \pm 0.00030$ & \cellcolor{orange!20}$\mathbf{0.00101 \pm 0.00021}$ \\
& rel. L2  & $0.12847 \pm 0.02011$ & \cellcolor{violet!10}$0.06208 \pm 0.00713$ & \cellcolor{blue!10}$0.04149 \pm 0.00334$ & \cellcolor{orange!20}$\mathbf{0.03258 \pm 0.00330}$ \\
\midrule
\multirow{2}{*}{4 layers}
& Test MSE & $0.01105 \pm 0.00257$ & \cellcolor{violet!10}$0.00333 \pm 0.00129$ & \cellcolor{orange!20}$\mathbf{0.00134 \pm 0.00027}$ & \cellcolor{blue!10}$0.00159 \pm 0.00008$ \\
& rel. L2  & $0.10088 \pm 0.01145$ & \cellcolor{violet!10}$0.05410 \pm 0.01010$ & \cellcolor{orange!20}$\mathbf{0.03567 \pm 0.00326}$ & \cellcolor{blue!10}$0.04121 \pm 0.00136$ \\
\midrule
\multirow{2}{*}{8 layers}
& Test MSE & $0.00599 \pm 0.00100$ & \cellcolor{violet!10}$0.00345 \pm 0.00062$ & \cellcolor{orange!20}$\mathbf{0.00157 \pm 0.00030}$ & \cellcolor{blue!10}$0.00164 \pm 0.00086$ \\
& rel. L2  & $0.07466 \pm 0.00641$ & \cellcolor{violet!10}$0.05836 \pm 0.00532$ & \cellcolor{blue!10}$0.03963 \pm 0.00344$ & \cellcolor{orange!20}$\mathbf{0.03724 \pm 0.00987}$ \\
\midrule
\multirow{2}{*}{16 layers}
& Test MSE & $0.00530 \pm 0.00046$ & \cellcolor{violet!10}$0.00354 \pm 0.00116$ & \cellcolor{orange!20}$\mathbf{0.00105 \pm 0.00038}$ & \cellcolor{blue!10}$0.00317 \pm 0.00162$ \\
& rel. L2  & $0.07259 \pm 0.00398$ & \cellcolor{violet!10}$0.05637 \pm 0.01122$ & \cellcolor{orange!20}$\mathbf{0.03035 \pm 0.00629}$ & \cellcolor{blue!10}$0.05141 \pm 0.01687$ \\
\bottomrule
\end{tabular}%
}
\end{table*}

\begin{table*}[!h]
\centering
\small
\caption{Polynomial-source Poisson $p=5$ results.}
\label{tab:polynomial-source-poisson-p5}
\resizebox{\textwidth}{!}{%
\begin{tabular}{llcccc}
\toprule
Layers & Metric & FNO (order 1) & HO-FNO (order 2) & HO-FNO (order 3) & HO-FNO (order 5) \\
\midrule
\multirow{2}{*}{1 layer}
& Test MSE & \cellcolor{violet!10}$1.11682 \pm 0.01314$ & $1.13766 \pm 0.03109$ & \cellcolor{blue!10}$1.02385 \pm 0.02889$ & \cellcolor{orange!20}$\mathbf{0.00851 \pm 0.00076}$ \\
& rel. L2  & \cellcolor{violet!10}$1.06408 \pm 0.00726$ & $1.07505 \pm 0.01781$ & \cellcolor{blue!10}$1.02586 \pm 0.01529$ & \cellcolor{orange!20}$\mathbf{0.09699 \pm 0.00414}$ \\
\midrule
\multirow{2}{*}{2 layers}
& Test MSE & \cellcolor{blue!10}$1.02980 \pm 0.20812$ & \cellcolor{violet!10}$1.21244 \pm 0.15537$ & $1.39929 \pm 0.02853$ & \cellcolor{orange!20}$\mathbf{0.00430 \pm 0.00176}$ \\
& rel. L2  & \cellcolor{blue!10}$1.01548 \pm 0.11067$ & \cellcolor{violet!10}$1.11958 \pm 0.07867$ & $1.20560 \pm 0.01451$ & \cellcolor{orange!20}$\mathbf{0.06429 \pm 0.01470}$ \\
\midrule
\multirow{2}{*}{4 layers}
& Test MSE & \cellcolor{blue!10}$0.98404 \pm 0.08967$ & \cellcolor{violet!10}$1.06894 \pm 0.14487$ & $1.19126 \pm 0.01585$ & \cellcolor{orange!20}$\mathbf{0.01249 \pm 0.00840}$ \\
& rel. L2  & \cellcolor{blue!10}$0.99307 \pm 0.04720$ & \cellcolor{violet!10}$1.03589 \pm 0.08020$ & $1.10348 \pm 0.00659$ & \cellcolor{orange!20}$\mathbf{0.09656 \pm 0.03562}$ \\
\midrule
\multirow{2}{*}{8 layers}
& Test MSE & \cellcolor{violet!10}$1.10587 \pm 0.03485$ & $1.15183 \pm 0.01744$ & \cellcolor{blue!10}$1.10466 \pm 0.00928$ & \cellcolor{orange!20}$\mathbf{0.02420 \pm 0.01467}$ \\
& rel. L2  & \cellcolor{violet!10}$1.06051 \pm 0.01992$ & $1.08270 \pm 0.00789$ & \cellcolor{blue!10}$1.05755 \pm 0.00603$ & \cellcolor{orange!20}$\mathbf{0.13644 \pm 0.04524}$ \\
\midrule
\multirow{2}{*}{16 layers}
& Test MSE & $1.20773 \pm 0.00633$ & \cellcolor{violet!10}$1.12194 \pm 0.03506$ & \cellcolor{blue!10}$1.00528 \pm 0.08521$ & \cellcolor{orange!20}$\mathbf{0.00918 \pm 0.00314}$ \\
& rel. L2  & $1.11114 \pm 0.00539$ & \cellcolor{violet!10}$1.05818 \pm 0.02435$ & \cellcolor{blue!10}$0.99823 \pm 0.04897$ & \cellcolor{orange!20}$\mathbf{0.09390 \pm 0.01953}$ \\
\bottomrule
\end{tabular}%
}
\end{table*}

\definecolor{FNOColor}{RGB}{220,75,60}       
\definecolor{HO2Color}{RGB}{46,160,67}       
\definecolor{HO3Color}{RGB}{30,136,229}      
\definecolor{HO5Color}{RGB}{142,68,173}      
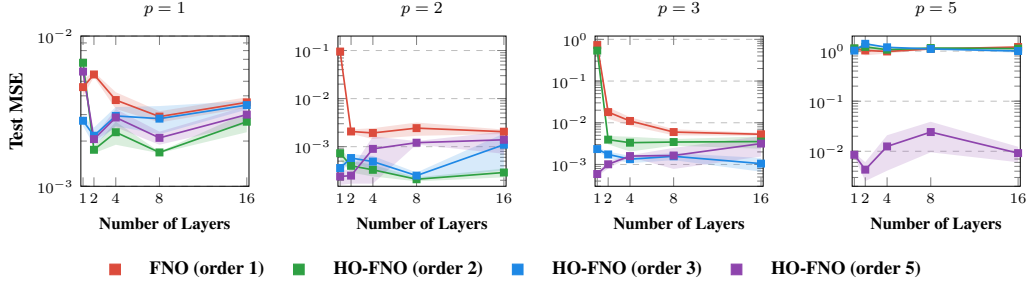
\begin{figure*}[!h]
\centering
\resizebox{\textwidth}{!}{%
\begin{tabular}{cccc}
\begin{tikzpicture}
\begin{axis}[
    width=4.4cm, height=4.1cm,
    xlabel={\textbf{Number of Layers}},
    ylabel={\small \textbf{Test MSE}},
    title={$p=1$},
    xmin=0.8, xmax=16.2,
    ymin=1e-3, ymax=1e-2,
    ymode=log,
    xtick={1,2,4,8,16},
    ymajorgrids=true,
    grid style=dashed,
    title style={font=\footnotesize},
    label style={font=\footnotesize},
    tick label style={font=\scriptsize},
    unbounded coords=jump
]

\addplot[name path=FNOupperp1, draw=none] coordinates {
    (1,0.00516) (2,0.00577) (4,0.00423) (8,0.00316) (16,0.00391)
};
\addplot[name path=FNOlowerp1, draw=none] coordinates {
    (1,0.00398) (2,0.00533) (4,0.00327) (8,0.00266) (16,0.00333)
};
\addplot[
    fill=FNOColor,
    fill opacity=0.18,
    draw=none,
    forget plot
] fill between[of=FNOupperp1 and FNOlowerp1];
\addplot[
    color=FNOColor,
    mark=square*,
    mark size=1.5pt,
    line width=0.8pt
] coordinates {
    (1,0.00457) (2,0.00555) (4,0.00375) (8,0.00291) (16,0.00362)
};

\addplot[name path=HO2upperp1, draw=none] coordinates {
    (1,0.00727) (2,0.00184) (4,0.00269) (8,0.00170) (16,0.00309)
};
\addplot[name path=HO2lowerp1, draw=none] coordinates {
    (1,0.00601) (2,0.00166) (4,0.00189) (8,0.00166) (16,0.00229)
};
\addplot[
    fill=HO2Color,
    fill opacity=0.18,
    draw=none,
    forget plot
] fill between[of=HO2upperp1 and HO2lowerp1];
\addplot[
    color=HO2Color,
    mark=square*,
    mark size=1.5pt,
    line width=0.8pt
] coordinates {
    (1,0.00664) (2,0.00175) (4,0.00229) (8,0.00168) (16,0.00269)
};
\addplot[name path=HO3upperp1, draw=none] coordinates {
    (1,0.00295) (2,0.00246) (4,0.00337) (8,0.00342) (16,0.00370)
};
\addplot[name path=HO3lowerp1, draw=none] coordinates {
    (1,0.00251) (2,0.00188) (4,0.00251) (8,0.00222) (16,0.00324)
};
\addplot[
    fill=HO3Color,
    fill opacity=0.18,
    draw=none,
    forget plot
] fill between[of=HO3upperp1 and HO3lowerp1];
\addplot[
    color=HO3Color,
    mark=square*,
    mark size=1.5pt,
    line width=0.8pt
] coordinates {
    (1,0.00273) (2,0.00217) (4,0.00294) (8,0.00282) (16,0.00347)
};

\addplot[name path=HO5upperp1, draw=none] coordinates {
    (1,0.00638) (2,0.00212) (4,0.00327) (8,0.00231) (16,0.00342)
};
\addplot[name path=HO5lowerp1, draw=none] coordinates {
    (1,0.00522) (2,0.00200) (4,0.00247) (8,0.00189) (16,0.00258)
};
\addplot[
    fill=HO5Color,
    fill opacity=0.18,
    draw=none,
    forget plot
] fill between[of=HO5upperp1 and HO5lowerp1];
\addplot[
    color=HO5Color,
    mark=square*,
    mark size=1.5pt,
    line width=0.8pt
] coordinates {
    (1,0.00580) (2,0.00206) (4,0.00287) (8,0.00210) (16,0.00300)
};

\end{axis}
\end{tikzpicture}%
&
\begin{tikzpicture}
\begin{axis}[
    width=4.4cm, height=4.1cm,
    xlabel={\textbf{Number of Layers}},
    title={$p=2$},
    xmin=0.8, xmax=16.2,
    ymin=1.5e-4, ymax=2e-1,
    ymode=log,
    xtick={1,2,4,8,16},
    ymajorgrids=true,
    grid style=dashed,
    title style={font=\footnotesize},
    label style={font=\footnotesize},
    tick label style={font=\scriptsize},
    unbounded coords=jump
]

\addplot[name path=FNOupperp2, draw=none] coordinates {
    (1,0.09847) (2,0.00229) (4,0.00242) (8,0.00318) (16,0.00227)
};
\addplot[name path=FNOlowerp2, draw=none] coordinates {
    (1,0.09181) (2,0.00185) (4,0.00144) (8,0.00168) (16,0.00183)
};
\addplot[
    fill=FNOColor,
    fill opacity=0.18,
    draw=none,
    forget plot
] fill between[of=FNOupperp2 and FNOlowerp2];
\addplot[
    color=FNOColor,
    mark=square*,
    mark size=1.5pt,
    line width=0.8pt
] coordinates {
    (1,0.09514) (2,0.00207) (4,0.00193) (8,0.00243) (16,0.00205)
};

\addplot[name path=HO2upperp2, draw=none] coordinates {
    (1,0.00093) (2,0.00051) (4,0.00042) (8,0.00023) (16,0.00035)
};
\addplot[name path=HO2lowerp2, draw=none] coordinates {
    (1,0.00053) (2,0.00029) (4,0.00024) (8,0.00019) (16,0.00023)
};
\addplot[
    fill=HO2Color,
    fill opacity=0.18,
    draw=none,
    forget plot
] fill between[of=HO2upperp2 and HO2lowerp2];
\addplot[
    color=HO2Color,
    mark=square*,
    mark size=1.5pt,
    line width=0.8pt
] coordinates {
    (1,0.00073) (2,0.00040) (4,0.00033) (8,0.00021) (16,0.00029)
};

\addplot[name path=HO3upperp2, draw=none] coordinates {
    (1,0.00041) (2,0.00070) (4,0.00064) (8,0.00029) (16,0.00188)
};
\addplot[name path=HO3lowerp2, draw=none] coordinates {
    (1,0.00031) (2,0.00046) (4,0.00034) (8,0.00021) (16,0.00032)
};
\addplot[
    fill=HO3Color,
    fill opacity=0.18,
    draw=none,
    forget plot
] fill between[of=HO3upperp2 and HO3lowerp2];
\addplot[
    color=HO3Color,
    mark=square*,
    mark size=1.5pt,
    line width=0.8pt
] coordinates {
    (1,0.00036) (2,0.00058) (4,0.00049) (8,0.00025) (16,0.00110)
};

\addplot[name path=HO5upperp2, draw=none] coordinates {
    (1,0.00032) (2,0.00033) (4,0.00163) (8,0.00131) (16,0.00207)
};
\addplot[name path=HO5lowerp2, draw=none] coordinates {
    (1,0.00016) (2,0.00017) (4,0.00017) (8,0.00111) (16,0.00071)
};
\addplot[
    fill=HO5Color,
    fill opacity=0.18,
    draw=none,
    forget plot
] fill between[of=HO5upperp2 and HO5lowerp2];
\addplot[
    color=HO5Color,
    mark=square*,
    mark size=1.5pt,
    line width=0.8pt
] coordinates {
    (1,0.00024) (2,0.00025) (4,0.00090) (8,0.00121) (16,0.00139)
};

\end{axis}
\end{tikzpicture}%
&
\begin{tikzpicture}
\begin{axis}[
    width=4.4cm, height=4.1cm,
    xlabel={\textbf{Number of Layers}},
    title={$p=3$},
    xmin=0.8, xmax=16.2,
    ymin=3e-4, ymax=1.2,
    ymode=log,
    xtick={1,2,4,8,16},
    ymajorgrids=true,
    grid style=dashed,
    title style={font=\footnotesize},
    label style={font=\footnotesize},
    tick label style={font=\scriptsize},
    unbounded coords=jump
]

\addplot[name path=FNOupperp3, draw=none] coordinates {
    (1,0.73489) (2,0.02351) (4,0.01362) (8,0.00699) (16,0.00576)
};
\addplot[name path=FNOlowerp3, draw=none] coordinates {
    (1,0.72287) (2,0.01253) (4,0.00848) (8,0.00499) (16,0.00484)
};
\addplot[
    fill=FNOColor,
    fill opacity=0.18,
    draw=none,
    forget plot
] fill between[of=FNOupperp3 and FNOlowerp3];
\addplot[
    color=FNOColor,
    mark=square*,
    mark size=1.5pt,
    line width=0.8pt
] coordinates {
    (1,0.72888) (2,0.01802) (4,0.01105) (8,0.00599) (16,0.00530)
};

\addplot[name path=HO2upperp3, draw=none] coordinates {
    (1,0.57697) (2,0.00483) (4,0.00462) (8,0.00407) (16,0.00470)
};
\addplot[name path=HO2lowerp3, draw=none] coordinates {
    (1,0.49731) (2,0.00305) (4,0.00204) (8,0.00283) (16,0.00238)
};
\addplot[
    fill=HO2Color,
    fill opacity=0.18,
    draw=none,
    forget plot
] fill between[of=HO2upperp3 and HO2lowerp3];
\addplot[
    color=HO2Color,
    mark=square*,
    mark size=1.5pt,
    line width=0.8pt
] coordinates {
    (1,0.53714) (2,0.00394) (4,0.00333) (8,0.00345) (16,0.00354)
};

\addplot[name path=HO3upperp3, draw=none] coordinates {
    (1,0.00257) (2,0.00204) (4,0.00161) (8,0.00187) (16,0.00143)
};
\addplot[name path=HO3lowerp3, draw=none] coordinates {
    (1,0.00211) (2,0.00144) (4,0.00107) (8,0.00127) (16,0.00067)
};
\addplot[
    fill=HO3Color,
    fill opacity=0.18,
    draw=none,
    forget plot
] fill between[of=HO3upperp3 and HO3lowerp3];
\addplot[
    color=HO3Color,
    mark=square*,
    mark size=1.5pt,
    line width=0.8pt
] coordinates {
    (1,0.00234) (2,0.00174) (4,0.00134) (8,0.00157) (16,0.00105)
};

\addplot[name path=HO5upperp3, draw=none] coordinates {
    (1,0.00066) (2,0.00122) (4,0.00167) (8,0.00250) (16,0.00479)
};
\addplot[name path=HO5lowerp3, draw=none] coordinates {
    (1,0.00052) (2,0.00080) (4,0.00151) (8,0.00078) (16,0.00155)
};
\addplot[
    fill=HO5Color,
    fill opacity=0.18,
    draw=none,
    forget plot
] fill between[of=HO5upperp3 and HO5lowerp3];
\addplot[
    color=HO5Color,
    mark=square*,
    mark size=1.5pt,
    line width=0.8pt
] coordinates {
    (1,0.00059) (2,0.00101) (4,0.00159) (8,0.00164) (16,0.00317)
};

\end{axis}
\end{tikzpicture}%
&
\begin{tikzpicture}
\begin{axis}[
    width=4.4cm, height=4.1cm,
    xlabel={\textbf{Number of Layers}},
    title={$p=5$},
    xmin=0.8, xmax=16.2,
    ymin=2e-3, ymax=2,
    ymode=log,
    xtick={1,2,4,8,16},
    ymajorgrids=true,
    grid style=dashed,
    title style={font=\footnotesize},
    label style={font=\footnotesize},
    tick label style={font=\scriptsize},
    unbounded coords=jump
]

\addplot[name path=FNOupperp5, draw=none] coordinates {
    (1,1.12996) (2,1.23792) (4,1.07371) (8,1.14072) (16,1.21406)
};
\addplot[name path=FNOlowerp5, draw=none] coordinates {
    (1,1.10368) (2,0.82168) (4,0.89437) (8,1.07102) (16,1.20140)
};
\addplot[
    fill=FNOColor,
    fill opacity=0.18,
    draw=none,
    forget plot
] fill between[of=FNOupperp5 and FNOlowerp5];
\addplot[
    color=FNOColor,
    mark=square*,
    mark size=1.5pt,
    line width=0.8pt
] coordinates {
    (1,1.11682) (2,1.02980) (4,0.98404) (8,1.10587) (16,1.20773)
};

\addplot[name path=HO2upperp5, draw=none] coordinates {
    (1,1.16875) (2,1.36781) (4,1.21381) (8,1.16927) (16,1.15700)
};
\addplot[name path=HO2lowerp5, draw=none] coordinates {
    (1,1.10657) (2,1.05707) (4,0.92407) (8,1.13439) (16,1.08688)
};
\addplot[
    fill=HO2Color,
    fill opacity=0.18,
    draw=none,
    forget plot
] fill between[of=HO2upperp5 and HO2lowerp5];
\addplot[
    color=HO2Color,
    mark=square*,
    mark size=1.5pt,
    line width=0.8pt
] coordinates {
    (1,1.13766) (2,1.21244) (4,1.06894) (8,1.15183) (16,1.12194)
};

\addplot[name path=HO3upperp5, draw=none] coordinates {
    (1,1.05274) (2,1.42782) (4,1.20711) (8,1.11394) (16,1.09049)
};
\addplot[name path=HO3lowerp5, draw=none] coordinates {
    (1,0.99496) (2,1.37076) (4,1.17541) (8,1.09538) (16,0.92007)
};
\addplot[
    fill=HO3Color,
    fill opacity=0.18,
    draw=none,
    forget plot
] fill between[of=HO3upperp5 and HO3lowerp5];
\addplot[
    color=HO3Color,
    mark=square*,
    mark size=1.5pt,
    line width=0.8pt
] coordinates {
    (1,1.02385) (2,1.39929) (4,1.19126) (8,1.10466) (16,1.00528)
};

\addplot[name path=HO5upperp5, draw=none] coordinates {
    (1,0.00927) (2,0.00606) (4,0.02089) (8,0.03887) (16,0.01232)
};
\addplot[name path=HO5lowerp5, draw=none] coordinates {
    (1,0.00775) (2,0.00254) (4,0.00409) (8,0.00953) (16,0.00604)
};
\addplot[
    fill=HO5Color,
    fill opacity=0.18,
    draw=none,
    forget plot
] fill between[of=HO5upperp5 and HO5lowerp5];
\addplot[
    color=HO5Color,
    mark=square*,
    mark size=1.5pt,
    line width=0.8pt
] coordinates {
    (1,0.00851) (2,0.00430) (4,0.01249) (8,0.02420) (16,0.00918)
};
\end{axis}
\end{tikzpicture}%
\end{tabular}%
}

\begin{tikzpicture}
\begin{axis}[
    hide axis,
    legend columns=-1,
    axis lines=none,
    ticks=none,
    legend style={
        draw=none,
        column sep=2ex,
        font=\scriptsize
    },
    xmin=0, xmax=1, ymin=0, ymax=1
]
\addlegendimage{only marks, mark=square*, color=FNOColor}
\addlegendentry{\textbf{FNO (order 1)}}

\addlegendimage{only marks, mark=square*, color=HO2Color}
\addlegendentry{\textbf{HO-FNO (order 2)}}

\addlegendimage{only marks, mark=square*, color=HO3Color}
\addlegendentry{\textbf{HO-FNO (order 3)}}

\addlegendimage{only marks, mark=square*, color=HO5Color}
\addlegendentry{\textbf{HO-FNO (order 5)}}
\end{axis}
\end{tikzpicture}%

\caption{\small Test MSE as a function of the number of layers on the Polynomial-Source Poisson datasets for $p=1,2,3,$ and $5$. Solid lines denote the mean over runs, and shaded bands indicate one standard deviation. Lower values indicate better performance.}
\label{fig: app Polynomial-Source Poisson_layers}
\end{figure*}

\newpage
\section{Extended Efficiency Analysis}
\subsection{Exact parameter count}\label{sec: parameter count}
For completeness, we report the number of learnable parameters introduced by each higher order spectral layer. Consider a layer with $C$ input and output channels, and let $M$ denote the number of retained Fourier modes. A standard FNO spectral convolution applies, for each retained mode, a complex-valued channel-mixing matrix of size $C \times C$. Therefore, up to the constant factor associated with the real and imaginary parts, the number of spectral parameters scales as
\begin{align*}
    MC^{2}.
\end{align*}
Our higher-order spectral convolution of order $m$ augments this linear spectral convolution with $m$ additional channel-mixing matrices, each of size $C \times C$. These matrices are shared across spatial locations and Fourier modes, and therefore contribute only
\begin{align*}
    mC^{2}
\end{align*}
additional parameters. The resulting number of parameters per spectral layer is thus
\begin{align*}
    MC^{2} + mC^{2}
    =
    (M+m)C^{2}.
\end{align*}
Hence, the parameter count grows only linearly with the interaction order $m$. Since the order $m$ is small compared with the number of retained modes $M$, the extra parameters introduced by higher-order mixing are negligible relative to the parameters already present in a standard FNO spectral convolution.

The depthwise variant used in the main experiments further reduces the number of parameters by removing channel mixing inside the Fourier-domain convolution. In this case, each retained Fourier mode has one complex-valued coefficient per channel, so the spectral convolution contains
\begin{align*}
    MC
\end{align*}
parameters instead of $MC^{2}$. The higher-order component is unchanged, since it still uses $m$ shared $C \times C$ matrices. Therefore, the total number of parameters per depthwise higher-order spectral layer becomes
\begin{align*}
    MC + mC^{2}.
\end{align*}
This shows that, for both the original spectral convolution and its more modern depthwise variant, the higher-order extension introduces only a negligible number of additional learnable parameters.

\subsection{Wall-Clock Times and Memory Usage}\label{app:wall_clock_times_and_memory_usage}
We report in this section Wall-Clock time recorded for one-epoch inference and training and peak memory usage on $3$ datasets from the standard benchmarks of Table~\ref{tab:main_result}. The experiments were conducted on a single NVIDIA GPU A100 with 80GB of memory. All models were tested with the hyperparameters reported in Section~\ref{app:implementation_details}. We report quantitative results in Table~\ref{tab:efficiency} and hystograms of raw values and normalized values with respect to FNO's performance in, respectively, Figure~\ref{fig:efficiency_plot} and Figure~\ref{fig:efficiency_plot_normalized}.

\paragraph{Baselines}
For the efficiency analysis, we consider one representative model from each main architectural category. As a frequency-based baseline, we use FNO, since it is the model on which our proposed architecture builds. As a state-space model, we use LaMO, which is also the strongest competitor in the experiments reported in Table~\ref{tab:main_result}. Finally, as a Transformer-based baseline, we use Transolver, since it is among the most effective Transformer architectures for physics problems and is also designed with computational efficiency in mind.

\paragraph{Results}The experiments show that HO-FNO largely preserves the computational profile of FNO. Across the three benchmarks, the training overhead of HO-FNO with respect to FNO remains limited: $+1.8\%$ on Navier--Stokes, $+6.2\%$ on Airfoil, and $+10.4\%$ on Pipe. A similar trend holds at inference time, where HO-FNO is almost identical to FNO on Navier--Stokes $(+0.4\%)$ and remains close on Airfoil and Pipe, with overheads of $+4.8\%$ and $+15.9\%$, respectively. This indicates that the higher-order spectral layer adds only a modest computational cost compared with the corresponding FNO backbone.

Compared with the non-Fourier baselines, HO-FNO is substantially more efficient, especially on Airfoil and Pipe. On Pipe, LaMO and Transolver require more than three times the training and inference time of FNO, with increases of $+215.9\%$ in training time and $+258.0\%$ in inference time, whereas HO-FNO remains within $+10.4\%$ and $+15.9\%$. On Airfoil, both LaMO and Transolver more than double the inference cost relative to FNO, while HO-FNO adds only $+4.8\%$. The memory overhead of HO-FNO is also moderate: it is almost negligible on Navier--Stokes $(+0.2\%)$ and remains comparable to the other efficient baselines on Airfoil and Pipe. Overall, these results show that HO-FNO retains the main efficiency advantages of FNO, namely fast training, fast inference, and moderate memory consumption.
\begin{figure*}[!h]
    \centering
    \setlength{\tabcolsep}{5pt}
    \definecolor{emerald}{RGB}{80,200,120}
    \definecolor{cobalt}{RGB}{0,71,171}

    \begin{tabular}{ccc}
    \begin{tikzpicture}
    \begin{axis}[
        width=4.7cm, height=4cm,
        ybar,
        bar width=5pt,
        ylabel=\tiny \textbf{Training time (s/epoch)},
        symbolic x coords={NS, Airfoil, Pipe},
        xtick=data,
        x tick label style={
            font=\tiny,
            rotate=0,
            anchor=center
        },
        ymin=0, 
        enlarge x limits=0.25,
        enlarge y limits=0.05,
        yticklabel style={
            /pgf/number format/.cd,
            fixed,
            precision=0
        }
    ]
    \addplot[fill=orange] coordinates {
        (NS,117.10)
        (Airfoil,20.33)
        (Pipe,64.25)
    };

    \addplot[fill=emerald] coordinates {
        (NS, 123.63)
        (Airfoil, 12.01)
        (Pipe, 20.34)
    };

    \addplot[fill=cobalt] coordinates {
        (NS, 125.81)
        (Airfoil, 12.76)
        (Pipe, 22.45)
    };

    \addplot[fill=violet] coordinates {
        (NS,130.11)
        (Airfoil,22.45)
        (Pipe,64.25)
    };
    \end{axis}
    \end{tikzpicture}
    &
    \begin{tikzpicture}
    \begin{axis}[
        width=4.7cm, height=4cm,
        ybar,
        bar width=5pt,
        ylabel=\tiny \textbf{Inference time (s/epoch)},
        symbolic x coords={NS, Airfoil, Pipe},
        xtick=data,
        x tick label style={
            font=\tiny,
            rotate=0,
            anchor=center
        },
        ymin=0,
        enlarge x limits=0.25,
        enlarge y limits=0.05,
        yticklabel style={
            /pgf/number format/.cd,
            fixed,
            precision=0
        }
    ]
    \addplot[fill=orange] coordinates {
        (NS,10.17)
        (Airfoil,1.77)
        (Pipe,5.62)
    };

    \addplot[fill=emerald] coordinates {
        (NS, 8.28)
        (Airfoil, 0.83)
        (Pipe, 1.57)
    };

    \addplot[fill=cobalt] coordinates {
        (NS, 8.31)
        (Airfoil, 0.87)
        (Pipe, 1.82)
    };

    \addplot[fill=violet] coordinates {
        (NS,12.72)
        (Airfoil,1.77)
        (Pipe,5.62)
    };
    \end{axis}
    \end{tikzpicture}
    &
    \begin{tikzpicture}
    \begin{axis}[
        width=4.7cm, height=4cm,
        ybar,
        bar width=5pt,
        ylabel=\tiny \textbf{Memory (GB)},
        symbolic x coords={NS, Airfoil, Pipe},
        xtick=data,
        x tick label style={
            font=\tiny,
            rotate=0,
            anchor=center
        },
        ymin=0, 
        enlarge x limits=0.25,
        enlarge y limits=0.05,
        yticklabel style={
            /pgf/number format/.cd,
            fixed,
            precision=0
        }
    ]
    \addplot[fill=orange] coordinates {
        (NS,28.45)
        (Airfoil,16.44)
        (Pipe,10.73)
    };

    \addplot[fill=emerald] coordinates {
        (NS, 45.85)
        (Airfoil, 13.85)
        (Pipe, 8.23)
    };

    \addplot[fill=cobalt] coordinates {
        (NS, 45.94)
        (Airfoil, 16.36)
        (Pipe, 9.77)
    };

    \addplot[fill=violet] coordinates {
        (NS,72.65)
        (Airfoil, 16.06)
        (Pipe,10.07)
    };
    \end{axis}
    \end{tikzpicture}
    \end{tabular}
    \vspace{0.2cm}

    \begin{tikzpicture}
    \begin{axis}[
        hide axis,
        legend columns=-1,
        axis lines=none,
        ticks=none,
        legend style={
            draw=none,
            column sep=2ex,
            font=\scriptsize
        },
        xmin=0, xmax=1, ymin=0, ymax=1
    ]
    \addlegendimage{only marks, mark=*, color=orange}
    \addlegendentry{\textbf{LaMO}}

    \addlegendimage{only marks, mark=*, color=emerald}
    \addlegendentry{\textbf{FNO}}

    \addlegendimage{only marks, mark=*, color=cobalt}
    \addlegendentry{\textbf{HO-FNO}}

    \addlegendimage{only marks, mark=*, color=violet}
    \addlegendentry{\textbf{Transolver}}
    \end{axis}
    \end{tikzpicture}

    \caption{Efficiency comparison on NS, Airfoil, and Pipe on a single Nvidia A100 GPU: (\textbf{Left}) training time, (\textbf{Middle}) inference time, and (\textbf{Right}) memory consumption per epoch. FNO is has the modern backbone}
    \label{fig:efficiency_plot}
\end{figure*}
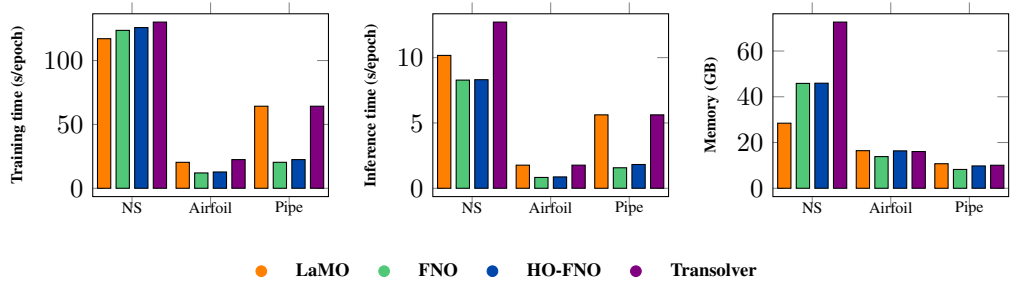

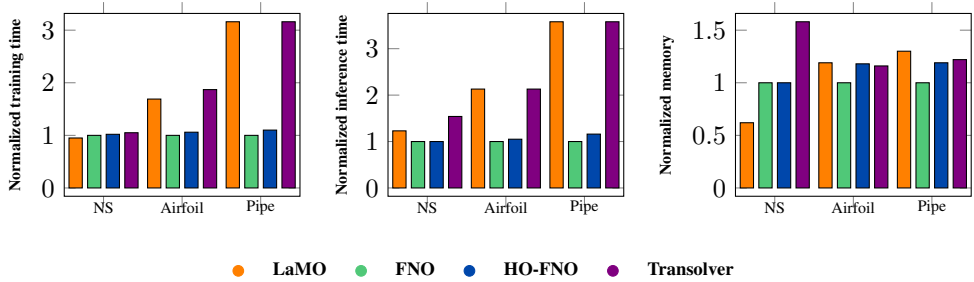
\begin{figure*}[!h]
    \centering
    \setlength{\tabcolsep}{5pt}
    \definecolor{emerald}{RGB}{80,200,120}
    \definecolor{cobalt}{RGB}{0,71,171}

    \begin{tabular}{ccc}
    \begin{tikzpicture}
    \begin{axis}[
        width=4.7cm, height=4cm,
        ybar,
        bar width=5pt,
        ylabel=\tiny \textbf{Normalized training time},
        symbolic x coords={NS, Airfoil, Pipe},
        xtick=data,
        x tick label style={
            font=\tiny,
            rotate=0,
            anchor=center
        },
        ymin=0,
        enlarge x limits=0.25,
        enlarge y limits=0.05,
        yticklabel style={
            /pgf/number format/.cd,
            fixed,
            precision=1
        }
    ]
    \addplot[fill=orange] coordinates {
        (NS,0.95)
        (Airfoil,1.69)
        (Pipe,3.16)
    };

    \addplot[fill=emerald] coordinates {
        (NS,1.00)
        (Airfoil,1.00)
        (Pipe,1.00)
    };

    \addplot[fill=cobalt] coordinates {
        (NS,1.02)
        (Airfoil,1.06)
        (Pipe,1.10)
    };

    \addplot[fill=violet] coordinates {
        (NS,1.05)
        (Airfoil,1.87)
        (Pipe,3.16)
    };
    \end{axis}
    \end{tikzpicture}
    &
    \begin{tikzpicture}
    \begin{axis}[
        width=4.7cm, height=4cm,
        ybar,
        bar width=5pt,
        ylabel=\tiny \textbf{Normalized inference time},
        symbolic x coords={NS, Airfoil, Pipe},
        xtick=data,
        x tick label style={
            font=\tiny,
            rotate=0,
            anchor=center
        },
        ymin=0,
        enlarge x limits=0.25,
        enlarge y limits=0.05,
        yticklabel style={
            /pgf/number format/.cd,
            fixed,
            precision=1
        }
    ]
    \addplot[fill=orange] coordinates {
        (NS,1.23)
        (Airfoil,2.13)
        (Pipe,3.58)
    };

    \addplot[fill=emerald] coordinates {
        (NS,1.00)
        (Airfoil,1.00)
        (Pipe,1.00)
    };

    \addplot[fill=cobalt] coordinates {
        (NS,1.00)
        (Airfoil,1.05)
        (Pipe,1.16)
    };

    \addplot[fill=violet] coordinates {
        (NS,1.54)
        (Airfoil,2.13)
        (Pipe,3.58)
    };
    \end{axis}
    \end{tikzpicture}
    &
    \begin{tikzpicture}
    \begin{axis}[
        width=4.7cm, height=4cm,
        ybar,
        bar width=5pt,
        ylabel=\tiny \textbf{Normalized memory},
        symbolic x coords={NS, Airfoil, Pipe},
        xtick=data,
        x tick label style={
            font=\tiny,
            rotate=0,
            anchor=center
        },
        ymin=0,
        enlarge x limits=0.25,
        enlarge y limits=0.05,
        yticklabel style={
            /pgf/number format/.cd,
            fixed,
            precision=1
        }
    ]
    \addplot[fill=orange] coordinates {
        (NS,0.62)
        (Airfoil,1.19)
        (Pipe,1.30)
    };

    \addplot[fill=emerald] coordinates {
        (NS,1.00)
        (Airfoil,1.00)
        (Pipe,1.00)
    };

    \addplot[fill=cobalt] coordinates {
        (NS,1.00)
        (Airfoil,1.18)
        (Pipe,1.19)
    };

    \addplot[fill=violet] coordinates {
        (NS,1.58)
        (Airfoil,1.16)
        (Pipe,1.22)
    };
    \end{axis}
    \end{tikzpicture}
    \end{tabular}

    \vspace{0.2cm}

    \begin{tikzpicture}
    \begin{axis}[
        hide axis,
        legend columns=-1,
        axis lines=none,
        ticks=none,
        legend style={
            draw=none,
            column sep=2ex,
            font=\scriptsize
        },
        xmin=0, xmax=1, ymin=0, ymax=1
    ]
    \addlegendimage{only marks, mark=*, color=orange}
    \addlegendentry{\textbf{LaMO}}

    \addlegendimage{only marks, mark=*, color=emerald}
    \addlegendentry{\textbf{FNO}}

    \addlegendimage{only marks, mark=*, color=cobalt}
    \addlegendentry{\textbf{HO-FNO}}

    \addlegendimage{only marks, mark=*, color=violet}
    \addlegendentry{\textbf{Transolver}}
    \end{axis}
    \end{tikzpicture}

    \caption{Efficiency comparison on NS, Airfoil, and Pipe after normalization with respect to FNO for each dataset and metric on a single Nvidia A100 GPU. A value of $1$ corresponds to the computational cost of FNO; values below $1$ indicate improved efficiency relative to FNO, whereas values above $1$ indicate higher computational cost.}
    \label{fig:efficiency_plot_normalized}
\end{figure*}

\begin{table*}[!h]
\centering
\small
\setlength{\tabcolsep}{6pt}
\begin{tabular}{c|c|ccc|ccc|ccc}
\toprule
& & \multicolumn{3}{c|}{\textbf{Training time (s)}} 
  & \multicolumn{3}{c|}{\textbf{Inference time (s)}} 
  & \multicolumn{3}{c}{\textbf{Memory (GB)}} \\
\textbf{Dataset} & \textbf{Model} 
& Value & $\Delta$ (\%) 
&  & Value & $\Delta$ (\%) 
&  & Value & $\Delta$ (\%) \\
\midrule

\multirow{4}{*}{NS}
& LaMO       & 117.10 & -5.3  && 10.17 & +22.8 && 28.45 & -38.0 \\
& FNO        & 123.63 & 0.0   && 8.28  & 0.0   && 45.85 & 0.0 \\
& HO-FNO     & 125.81 & +1.8  && 8.31  & +0.4  && 45.94 & +0.2 \\
& Transolver & 130.11 & +5.2  && 12.72 & +53.6 && 72.65 & +58.5 \\
\midrule

\multirow{4}{*}{Airfoil}
& LaMO       & 20.33 & +69.2  && 1.77 & +113.3 && 16.44 & +18.7 \\
& FNO        & 12.01 & 0.0    && 0.83 & 0.0    && 13.85 & 0.0 \\
& HO-FNO     & 12.76 & +6.2   && 0.87 & +4.8   && 16.36 & +18.1 \\
& Transolver & 22.45 & +86.9  && 1.77 & +113.3 && 16.06 & +16.0 \\
\midrule

\multirow{4}{*}{Pipe}
& LaMO       & 64.25 & +215.9 && 5.62 & +258.0 && 10.73 & +30.4 \\
& FNO        & 20.34 & 0.0    && 1.57 & 0.0    && 8.23  & 0.0 \\
& HO-FNO     & 22.45 & +10.4  && 1.82 & +15.9  && 9.77  & +18.7 \\
& Transolver & 64.25 & +215.9 && 5.62 & +258.0 && 10.07 & +22.4 \\
\bottomrule
\end{tabular}
\caption{Comparison of computational efficiency across datasets. For each metric, we report the absolute value and the relative variation (in \%) with respect to FNO. The relative variation, $\Delta$, is computed relative to FNO as
$\Delta(\%) = 100 \times (\mathrm{Value}_{\mathrm{model}} - \mathrm{Value}_{\mathrm{FNO}})/\mathrm{Value}_{\mathrm{FNO}}$. Negative values indicate an improvement over FNO (lower is better), while positive values indicate a degradation.
}
\label{tab:efficiency}
\end{table*}

\newpage
\section{Spectral Analysis}
\label{app:spectral_analysis}

In this section, we complement the aggregate test errors with a frequency-domain analysis of the predictions. For each benchmark, we decompose the error into spectral bands and report the relative error on the full spectrum, as well as on low-, medium-, and high-frequency modes. This analysis is useful because two models with similar global relative error may behave differently across scales: low frequencies typically capture the dominant large-scale structure of the solution, whereas medium and high frequencies encode sharper spatial variations and fine-scale details.

Table~\ref{tab:spectral_anaylisis_table} reports the comparison between FNO and HO-FNO on Navier--Stokes, Airfoil, and Pipe. Across the three datasets, HO-FNO improves over FNO on the full-spectrum error, with relative gains of $54.10\%$ on Navier--Stokes, $26.1\%$ on Airfoil, and $22.7\%$ on Pipe. The improvement is also visible across most frequency regimes, indicating that the higher-order spectral convolution does not only improve the dominant low-frequency component, but also affects the representation of medium- and high-frequency structures.

On Airfoil and Pipe, the gains are particularly consistent across the spectrum. HO-FNO reduces the full error and improves the low-, medium-, and high-frequency errors, suggesting that the proposed higher-order interactions help recover both the global flow structure and finer spatial variations around the geometry. On Navier--Stokes, HO-FNO provides a large improvement on the full error and a strong reduction in the low-frequency regime. The medium- and high-frequency errors are comparable to FNO, which suggests that the main gain on this benchmark comes from a more accurate reconstruction of the energetically dominant large-scale modes.

Figures~\ref{fig:spectrum_airfoil}, \ref{fig:spectrum_pip}, and \ref{fig:spectrum_ns} provide a more detailed visualization of this behavior. Each row compares the spectral profiles of FNO and HO-FNO on one benchmark. The plots show that HO-FNO generally follows the target spectrum more closely, especially in regimes where FNO either underestimates or overestimates the energy contained in specific frequency bands. This supports the interpretation that higher-order spectral mixing improves the model's ability to represent nonlinear interactions between modes, rather than merely improving the average prediction error.

Overall, the spectral analysis confirms that the gains of HO-FNO are not restricted to the physical-space relative error. The proposed higher-order spectral convolution leads to better frequency-wise behavior across different PDE benchmarks, with particularly clear improvements on Airfoil and Pipe and substantial low-frequency gains on Navier--Stokes.

\begin{table}[!h]
\centering
\resizebox{\textwidth}{!}{
\begin{tabular}{cc|ccc}
\hline
Dataset & Regime & FNO & HO-FNO & Improvement over FNO \\
\hline

\multirow{4}{*}{Navier--Stokes}
& Full   & 0.0976   & 0.0449   & 54.10\% \\
& Low    & 0.145982 & 0.055541 & 61.95\% \\
& Medium & 0.282080 & 0.206717 & 26.72\% \\
& High   & 0.458579 & 0.372876 & 18.69\% \\
\hline

\multirow{4}{*}{Airfoil}
& Full   & 0.0065   & 0.0051   & 26.1\% \\
& Low    & 0.025199 & 0.017370 & 31.07\% \\
& Medium & 0.070703 & 0.045756 & 35.28\% \\
& High   & 0.146905 & 0.121412 & 17.35\% \\
\hline

\multirow{4}{*}{Pipe}
& Full   & 0.0069   & 0.0054   & 32.7\% \\
& Low    & 0.191112 & 0.141260 & 26.1\% \\
& Medium & 1.040052 & 1.031456 & 0.83\% \\
& High   & 1.253831 & 1.182214 & 5.71\% \\

\hline
\end{tabular}}
\caption{Comparison of FNO and HO-FNO across datasets and regimes (Full, Low, Medium, High).}
\label{tab:spectral_anaylisis_table}
\end{table}

\begin{figure}[!h] \centering\includegraphics[width=\columnwidth]{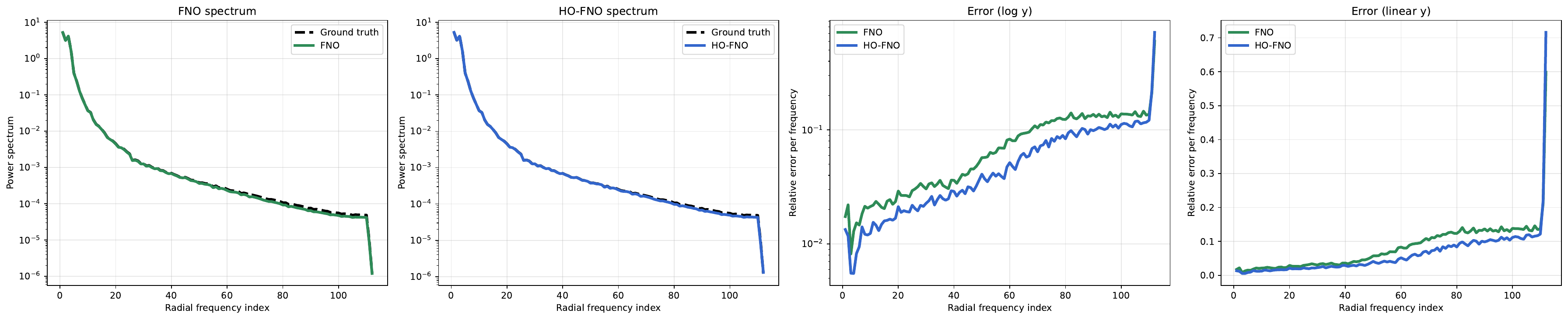}
    \caption{Spectrum of the predictions of FNO and HO-FNO (order 2) on the Airfoil benchmark.
    } 
    \vspace{-1em}
    \label{fig:spectrum_airfoil}
\end{figure}

\begin{figure}[!h]
    \centering\includegraphics[width=\columnwidth]{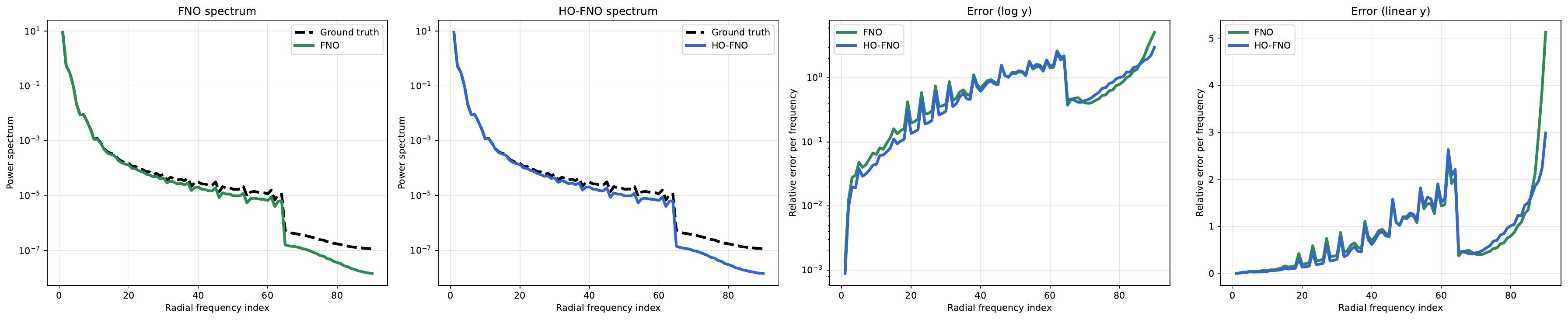}
    \caption{Spectrum of the predictions of FNO and HO-FNO (order 2) on the Pipe benchmark.
    } 
    \vspace{-1em}
    \label{fig:spectrum_pip}
\end{figure}

\begin{figure}[!h]
    \centering\includegraphics[width=\columnwidth]{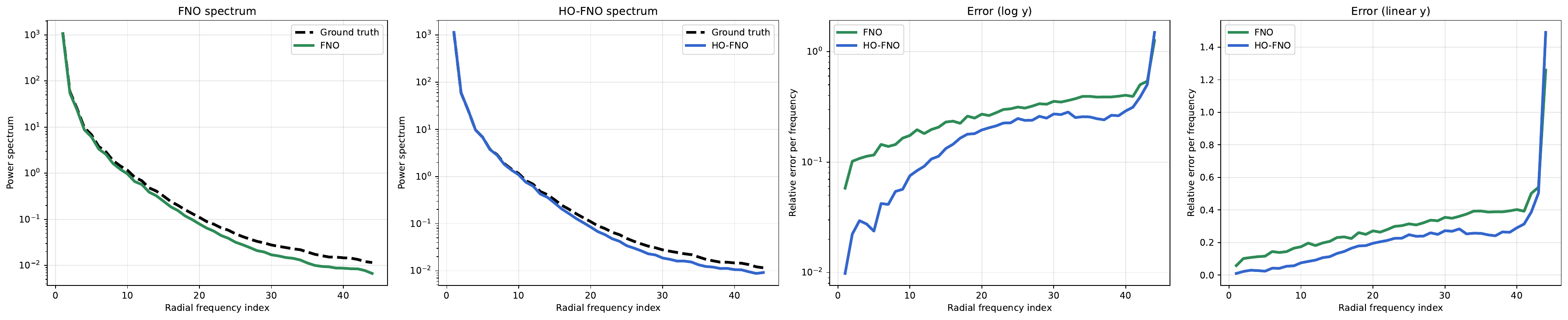}
    \caption{Spectrum of the predictions of FNO and HO-FNO (order 2) on the Navier--Stokes benchmark.
    } 
    \vspace{-1em}
    \label{fig:spectrum_ns}
\end{figure}

\newpage
\section{Implementation Details}
\label{app:implementation_details}
This section presents a comprehensive overview of the experimental setup, covering benchmark datasets, evaluation metrics,
and implementation details to ensure a rigorous and reproducible analysis.

\label{app:xp_details}
This section presents a comprehensive overview of the experimental setup, covering benchmark datasets, evaluation metrics,
and implementation details to ensure a rigorous and reproducible analysis.

\begin{table}[!h]
    \caption{Training and model configurations of HO-FNO relative to the experiments in Table~\ref{tab:main_result}. Training configurations are inspired by previous works in order to compare with the baselines at similar parameter count and capacity. For the Darcy dataset, we adopt an additional spatial gradient regularization term $l_{\mathrm{gdl}}$ following ONO and LaMO.}
    \label{tab:training_detail}
    \centering
    \resizebox{0.9\textwidth}{!}{%
    \begin{threeparttable}
        \begin{sc}
            \renewcommand{\multirowsetup}{\centering}
            \renewcommand{\arraystretch}{1.4}
            \setlength{\tabcolsep}{6pt}
            \begin{tabular}{llcccccc}
                \toprule
                \multicolumn{2}{c}{\multirow{2}{*}{Configuration}} 
                & \multicolumn{6}{c}{Benchmarks} \\
                \cmidrule(lr){3-8}
                & & Darcy & Navier--Stokes & Elasticity & Plasticity & Airfoil & Pipe \\
                \midrule
                \multirow{7}{*}{\rotatebox[origin=c]{90}{Training}}
                    & Loss Function 
                    & $l_2 + 0.1l_{\mathrm{gdl}}$ 
                    & \multicolumn{5}{c}{Relative $l_2$} \\
                    \cmidrule(lr){2-8}

                    & Epochs 
                    & \multicolumn{6}{c}{500} \\
                    \cmidrule(lr){2-8}

                    & Initial LR 
                    & $5 \cdot 10^{-4}$ 
                    & \multicolumn{5}{c}{$1 \cdot 10^{-3}$} \\
                    \cmidrule(lr){2-8}

                    & Final LR 
                    & \multicolumn{6}{c}{$5 \cdot 10^{-5}$} \\
                    \cmidrule(lr){2-8}

                    & Optimizer 
                    & \multicolumn{6}{c}{AdamW} \\
                    \cmidrule(lr){2-8}

                    & Batch Size 
                    & $8$ & $6$ & \multicolumn{4}{c}{$8$} \\
                    \cmidrule(lr){2-8}

                    & Scheduler 
                    & \multicolumn{6}{c}{CosineAnnealingLR} \\
                \midrule

                \multirow{4}{*}{\rotatebox[origin=c]{90}{Architecture}}
                    & Layers 
                    & \multicolumn{6}{c}{8} \\
                    \cmidrule(lr){2-8}

                    & Embedding Dim 
                    & $128$ & $256$ & \multicolumn{4}{c}{$128$} \\
                    \cmidrule(lr){2-8}

                    & Retained modes on the $x$-axis 
                    & $32$ & $16$ & $10$ & $24$ & $50$ & $50$ \\
                    \cmidrule(lr){2-8}

                    & Retained modes on the $y$-axis 
                    & $16$ & $16$ & $10$ & $12$ & $24$ & $24$ \\
                \bottomrule
            \end{tabular}
        \end{sc}
    \end{threeparttable}%
    }
\end{table}

\subsection{Training Details}\label{app:training_details}
In Table~\ref{tab:main_result}, the datasets are split into training and testing sets by configurations adopted from prior works~\citep{GEO-FNO} and described in Table~\ref{tab:benchmark_geofno}

\begin{table*}[!h]
\centering
\caption{Benchmark overview for Table~\ref{tab:main_result}, following the settings of~\citep{GEO-FNO}, with task description, dataset size, and input--output tensor formats.}
\label{tab:benchmark_geofno}
\resizebox{\textwidth}{!}{
\begin{tabular}{lccccccc}
\toprule
\textbf{Overview} 
& \textbf{Elasticity} 
& \textbf{Plasticity} 
& \textbf{Navier--Stokes} 
& \textbf{Airfoil} 
& \textbf{Pipe} 
& \textbf{Darcy}
& \textbf{Polynomial-Source Poisson} \\
\midrule

Task 
& Estimate Stress 
& Model Deformation 
& Model Viscous Flow 
& Estimate Velocity 
& Estimate Velocity 
& Estimate Pressure
& Solve Nonlinear Poisson \\

Input 
& Material Structure 
& External Force 
& Past Velocity 
& Structure 
& Structure 
& Porous Medium
& $p$ Input Fields \\

Output 
& Inner Stress 
& Mesh Displacement 
& Future Velocity 
& Mach Number 
& Fluid Velocity 
& Fluid Pressure
& Poisson Solution \\

Train Set Size 
& 1000 
& 900 
& 1000 
& 1000 
& 1000 
& 1000
& 1000 \\

Test Set Size 
& 200 
& 80 
& 200 
& 100 
& 200 
& 200
& 200 \\

Input Tensor 
& $(/,972,2)$ 
& $(/,101 \times 31,2)$ 
& $(10,64 \times 64,1)$ 
& $(/,221 \times 51,2)$ 
& $(/,129 \times 129,2)$ 
& $(/,85 \times 85,1)$
& $(/,64 \times 64,p)$ \\

Output Tensor 
& $(/,972,1)$ 
& $(20,101 \times 31,4)$ 
& $(10,64 \times 64,1)$ 
& $(/,221 \times 51,1)$ 
& $(/,129 \times 129,1)$ 
& $(/,85 \times 85,1)$
& $(/,64 \times 64,1)$ \\

\bottomrule
\end{tabular}
}
\end{table*}

Each dataset corresponds to a specific task and follows these established settings. For the results reported in Table~\ref{tab:main_result}, we train for $500$ epochs with AdamW~\citep{AdamW} and a cosine-annealing learning-rate schedule~\citep{CosineAnneal}, using the mean relative $\ell_2$ error (nRMSE) both as training objective and evaluation metric. The training configurations include the use of a relative \(l_2\) loss term across all datasets, with an additional spatial gradient regularization term \(l_{\mathrm{gdl}}\) incorporated for the Darcy dataset as \( l_2 + 0.1l_{\mathrm{gdl}} \) following ONO \cite{ONO}. Further specifics on the dataset splits, tasks, and training configurations are available in the respective tables for additional clarity for each benchmark dataset. A comprehensive summary of the HO-FNO hyperparameters is provided in Table~\ref{tab:training_detail}. The table shows that all baselines and benchmarks were trained using consistent configurations, ensuring that HO-FNO maintains a comparable number of parameters to LaMO and fewer than the transformer-based baselines.

In Table~\ref{tab:results_spherical} and Table~\ref{tab:results1M}, we train our model and the baselines for $500$ epochs with AdamW~\citep{AdamW} and a cosine-annealing learning-rate schedule~\citep{CosineAnneal}, using the  single-step $\ell_2$ error (RMSE)as training objective and single-step $\ell_2$ error (MSE) and relative $\ell_2$ error (nRMSE) and rollout relative $\ell_2$ error as evaluation metric. 

For the experiments reported in Table~\ref{tab:results1M} and Table~\ref{tab:results_spherical}, we adopt a single-step training loss, as multi-step supervision did not lead to consistent performance gains while substantially increasing computational and memory costs. Moreover, for the Navier–Stokes datasets, unlike the experiments in Table~\ref{tab:main_result} where we follow the setting of~\citep{FNO} and perform autoregressive prediction using the past $10$ timesteps, in Table~\ref{tab:results1M} we instead predict future states based solely on the current timestep. This choice simplifies training and reduces temporal dependencies during inference. Although leveraging past information may seem advantageous, prior work suggests an opposite trend: while historical context can improve short-term accuracy, it often compromises the stability and reliability of long-term rollouts due to error accumulation and distribution shift~\cite{pde_refiner, moe}. In our experiments, this simpler formulation resulted in more stable optimization and more robust long-horizon predictions. However, we evaluate predictions starting from timestep $10$ to match the test set of~\cite{FNO} and and ensure comparability with prior work, even though all baselines reported in this paper are retrained under the same protocol. 

For the experiments on the Poisson equation with polynomial forcing discussed in Section~\ref{subsec: poly_Poisson}, all FNO and HO-FNO variants are trained for $100$ epochs using the AdamW optimizer, with weight decay $10^{-5}$ and learning rate $2 \times 10^{-3}$. All models retain $16$ Fourier modes along each spatial axis, i.e., $\texttt{modes1}=\texttt{modes2}=16$, and use an MLP expansion factor of $2$ in each block.
\subsection{Backbone Details}\label{app:backbone_details}
We discuss in this section the differences in the backbone of the original FNO and the more modern backbone adopted in our experiments. Figure~\ref{fig:backbones} offers a visualization of the $2$ versions.

\textbf{The original FNO block~\citeyearpar{FNO}} has a simple structure: it combines a spectral convolution with a learnable residual connection, parameterized by a pointwise linear layer. 

Motivated by the success of modern architectures, recent works~\citep{multi-grid_fno} have adopted FNO backbones closer in structure to Transformers and modern CNNs. We follow this design philosophy and use a backbone similar to the one introduced for FNO variants in~\citep{multi-grid_fno}, which is also currently adopted in the Neural Operator Library~\citep{neural_operator_library}.

More specifically, we use a standard \textbf{pre-norm residual block}, as commonly employed in modern Transformer architectures. RMSNorm is applied before both the mixer and the feedforward network (FFN). The mixer corresponds either to a standard spectral convolution or to one of our higher-order spectral convolution variants. The FFN is implemented as a two-layer MLP with an intermediate expansion factor.

Using a more engineered backbone is important for a fair comparison with recent Transformer- and state-space-based neural operators, which already benefit from these architectural improvements.
\begin{figure}[!h]
    \centering
    \includegraphics[width=0.3\linewidth]{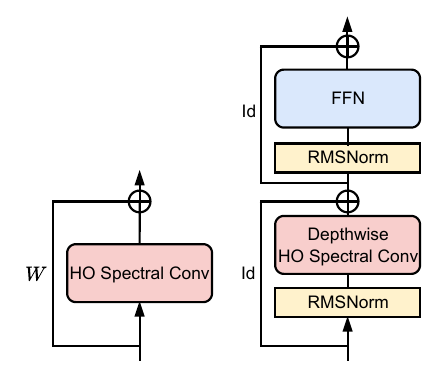}

    \caption{Comparison of the network blocks used in our experiments. We compare the standard FNO block~\cite{FNO} with our improved backbone based on a pre-norm residual architecture. The feed-forward network (FFN) consists of a two-layer MLP with an expansion factor of $4$ and GELU nonlinearity. Each sub-block applies RMS normalization before the Mixer and the FFN, followed by a residual connection.}
    \label{fig:backbones}
    \vspace{-1em}
\end{figure}

\subsection{Baselines}\label{app:baselines}

We compare HO-FNO against classical convolutional architectures, transformer-based neural operators, state-space models, and frequency-domain neural operators. We briefly describe below the baselines used in the main comparison.

\paragraph{UNet \citeyearpar{UNET}.}
UNet is an encoder--decoder convolutional architecture with skip connections between feature maps at matching resolutions. The downsampling path extracts increasingly coarse features, while the upsampling path reconstructs high-resolution outputs using both coarse and fine-scale information.

\paragraph{ResNet \citeyearpar{RESNET}.}
ResNet introduces residual blocks of the form $x \mapsto x + F(x)$, which ease the optimization of deep convolutional networks. In our setting, it serves as a purely local convolutional baseline without explicit spectral or operator-learning structure.

\paragraph{Swin Transformer \citeyearpar{SWIN}.}
Swin Transformer replaces global self-attention with attention inside shifted local windows. The window-shifting mechanism allows information exchange across neighboring windows while keeping attention cost manageable on dense grids.

\paragraph{DeepONet \citeyearpar{DEEPONET}.}
DeepONet represents operators through a branch--trunk factorization: the branch network encodes sensor values of the input function, while the trunk network encodes the query coordinates. The prediction is obtained by combining these two representations, typically through an inner product.

\paragraph{Galerkin Transformer \citeyearpar{Galerkin_Trans}.}
Galerkin Transformer modifies standard attention by using a linearized Galerkin-type attention mechanism. Instead of computing full quadratic self-attention, it projects features into a kernelized representation, improving scalability for operator learning.

\paragraph{HT-Net \citeyearpar{Ht-net}.}
HT-Net combines transformer blocks with hierarchical token representations for PDE modeling. It introduces coarse-to-fine feature processing so that attention can capture long-range interactions without operating only at the finest resolution.

\paragraph{OFormer \citeyearpar{OFormer}.}
OFormer introduces an encoder--decoder transformer for operator learning, where the encoder processes input function values and the decoder evaluates the solution at target coordinates. It is designed to handle mappings between functions sampled on possibly different discretizations.

\paragraph{GNOT \citeyearpar{GNOT}.}
GNOT extends transformer-based neural operators to multiple input functions and geometries by using heterogeneous attention over discretized function values and query points. It explicitly incorporates geometric information into the attention-based operator representation.

\paragraph{FactFormer \citeyearpar{FactFormer}.}
FactFormer introduces factorized attention for neural operators by decomposing multidimensional attention into lower-dimensional attention operations. This reduces the cost of global mixing while preserving interactions across spatial dimensions.

\paragraph{ONO \citeyearpar{ONO}.}
ONO introduces an orthogonal neural operator architecture, using orthogonal basis representations and attention-style mixing to improve stability and operator approximation. It is a strong recent baseline among transformer-inspired neural operators.

\paragraph{Transolver \citeyearpar{transolver}.}
Transolver introduces physics-attention, where spatial tokens are softly assigned to a smaller set of learned physical states. Attention is then computed in this latent physical-state space, reducing cost while capturing global interactions.

\paragraph{Erwin \citeyearpar{erwin}.}
Erwin introduces ball attention, restricting attention to local neighborhoods defined by spatial proximity. This preserves geometric locality and improves scalability compared with full attention over all mesh or grid points.

\paragraph{Transolver++ \citeyearpar{transolver++}.}
Transolver++ improves Transolver through a refined physics-attention design and stronger architectural components. It keeps the idea of latent physical states while enhancing the efficiency and expressivity of the attention blocks.

\paragraph{MSPT \citeyearpar{mspt}.}
MSPT introduces a multi-scale point transformer architecture for PDE modeling. It processes point or mesh features across several spatial resolutions and uses attention-based interactions to exchange information between local and coarse representations.

\paragraph{LaMO \citeyearpar{LaMO}.}
LaMO replaces transformer attention with a Mamba-style state-space mixing mechanism adapted to neural operators. It models long-range dependencies through selective state-space updates, providing a strong alternative to attention-based operator models.

\paragraph{WMT \citeyearpar{WMO}.}
WMT uses wavelet-domain representations instead of purely Fourier-domain representations. The wavelet transform provides both spatial localization and frequency information, making it suitable for multiscale PDE features.

\paragraph{U-FNO \citeyearpar{U-FNO}.}
U-FNO augments Fourier layers with a UNet-like path. The Fourier blocks capture global spectral interactions, while the UNet component improves local and multiscale feature extraction.

\paragraph{FNO \citeyearpar{FNO, GEO-FNO}.}
FNO applies global convolution in the Fourier domain by truncating to a fixed number of low-frequency modes and learning complex-valued weights on these modes. It is the main first-order spectral baseline and the closest reference model to HO-FNO.

\paragraph{U-NO \citeyearpar{UNO}.}
U-NO combines neural operator layers with a U-shaped encoder--decoder structure. It applies operator blocks at multiple resolutions, allowing the model to mix global operator learning with multiscale spatial reconstruction.

\paragraph{F-FNO \citeyearpar{F-FNO}.}
F-FNO factorizes the Fourier convolution across spatial dimensions. This reduces the parameter and computational cost of spectral mixing while retaining the main Fourier-domain inductive bias of FNO.

\paragraph{LSM \citeyearpar{LSM}.}
LSM introduces a latent spectral representation for operator learning. Instead of performing all computations directly on the full grid, it maps the input to a latent space where spectral mixing is applied more efficiently.

\paragraph{FNO with modern backbone.}
We include an internal baseline that uses the same modern backbone as HO-FNO but keeps the spectral convolution first-order. This separates the gain due to backbone design from the gain due to higher-order spectral mixing.

\subsection{Metrics description} \label{app:metrics}

We evaluate the predictive performance of our models using the following metrics:

\paragraph{Mean Squared Error (MSE).}  
Given ground truth $y \in \mathbb{R}^d$ and prediction $\hat{y} \in \mathbb{R}^d$, the MSE, sometimes called $L_2$-norm, is defined as
\begin{align}
    \text{MSE}(y,\hat{y}) = \frac{1}{d} \sum_{i=1}^d \lVert y_i - \hat{y}_i \rVert^2.
\end{align}
This metric measures the average squared deviation between predictions and targets. It is numerically stable and therefore commonly used as a training loss, as we do in our experiments. At test time, MSE is also informative since it provides a physically meaningful error measure in the original space.  

However, MSE scales quadratically with multiplicative factors applied to $y$ and $\hat{y}$, and it is affected by the discretization of the domain. As a result, it is not directly comparable across different datasets or resolutions. For this reason, it is often preferred to also report the Normalized Mean Squared Error (NRMSE) at evaluation time.

\paragraph{Normalized Mean Squared Error (nRMSE).}  
The nRMSE, often called relative $L_2$-norm, is the MSE normalized by the norm of the target:
\begin{equation}\label{eq:nRMSE_app}
    \text{NRMSE}(y,\hat{y}) = \frac{1}{d} \sum_{i=1}^d \frac{\lVert y_i - \hat{y}_i\rVert^2}{\lVert y_i \rVert^2}.
\end{equation}
Unlike MSE, which reports squared units, RMSE is expressed in the same units as the target variable. 

This makes the error magnitude directly comparable to the physical scale of the data, providing a more intuitive sense of accuracy therefore providing a fair comparisons across datasets and resolutions.

\paragraph{Rollout Error.}  \label{app:rollout}
Since we deal with time-dependent systems, we evaluate multi-step predictions by iteratively feeding model outputs back as inputs. The rollout error is computed as the average of a choosen loss , $\mathcal{L}$, across all timesteps:
\begin{align}
    \text{Rollout}(y_{1:T}, \hat{y}_{1:T}) = \frac{1}{T} \sum_{t=1}^T \mathcal{L}(y_t, \hat{y}_t),
\end{align}
where $T$ is the total number of time steps of the dataset. This metric captures error accumulation over long-term forecasts. 

Even though rollout stability is beyond the scope of this work, it remains informative to assess how new models perform in this setting, which more closely reflects real-world applications than the teacher-forcing setup. For this reason, we report rollout metrics in our experiments.

\newpage
\section{Visualizations of the predictions}
We report in this section the visualization of the predictions for a qualitative visual comparison between FNO and HO-FNO. Both models employ the modern backbone.
\begin{figure}[!h]
    \centering\includegraphics[width=0.7\columnwidth]{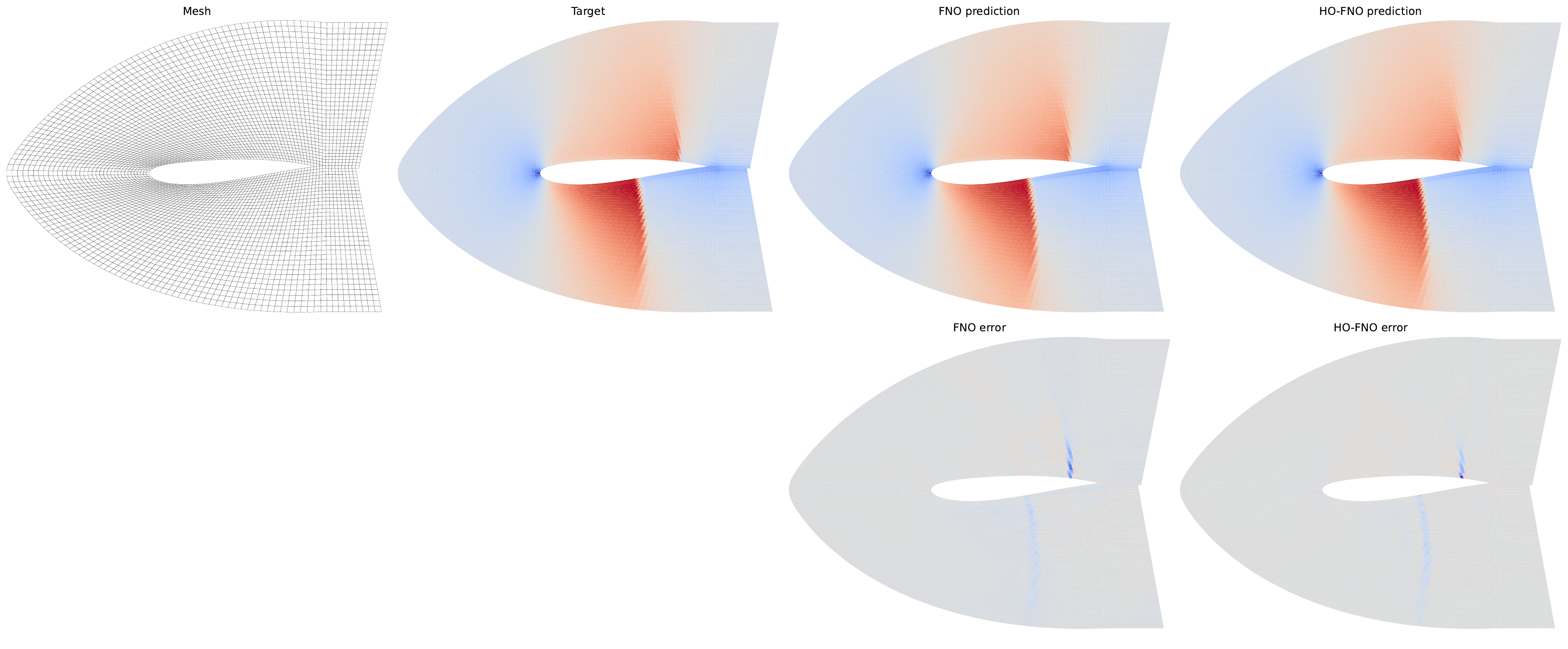}
    \caption{Qualitative visualization of predictions of FNO and HO-FNO (order 2) on the Airfoil dataset.
    } 
    \vspace{-1em}
    \label{fig:backbones}
\end{figure}

\begin{figure}[!h]
    \centering\includegraphics[width=0.7\columnwidth]{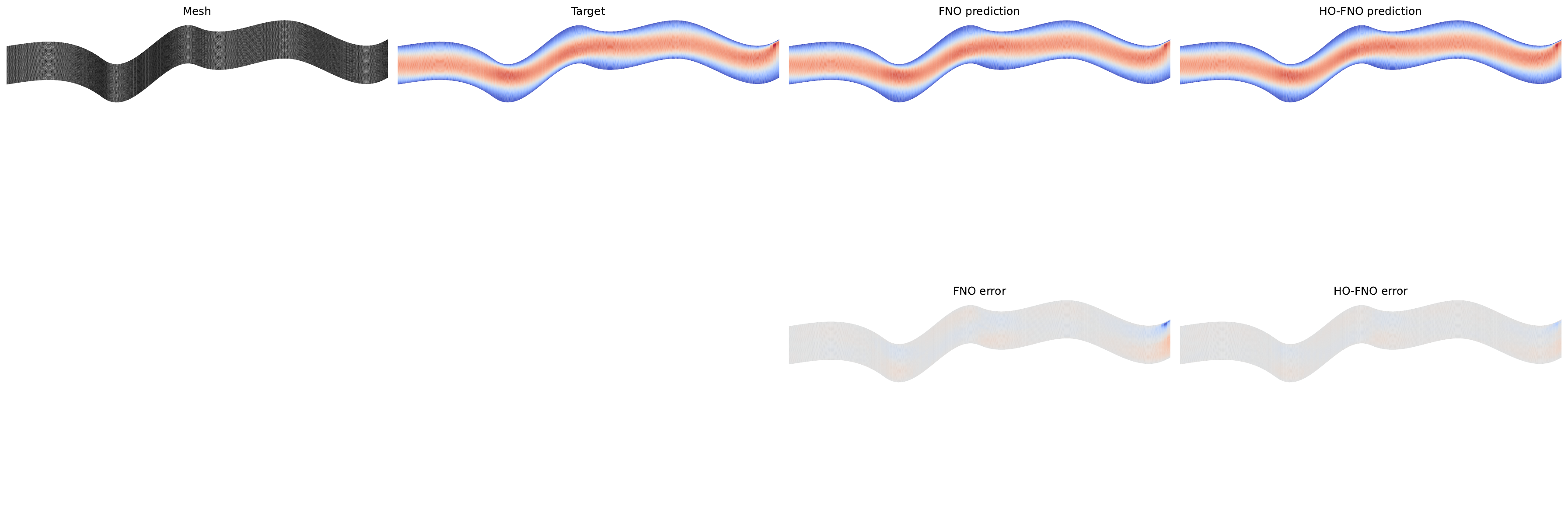}
    \caption{Qualitative visualization of predictions of FNO and HO-FNO (order 2) on the Pipe dataset.
    } 
    \vspace{-1em}
    \label{fig:backbones}
\end{figure}

\begin{figure}[!h]
    \centering\includegraphics[width=0.7\columnwidth]{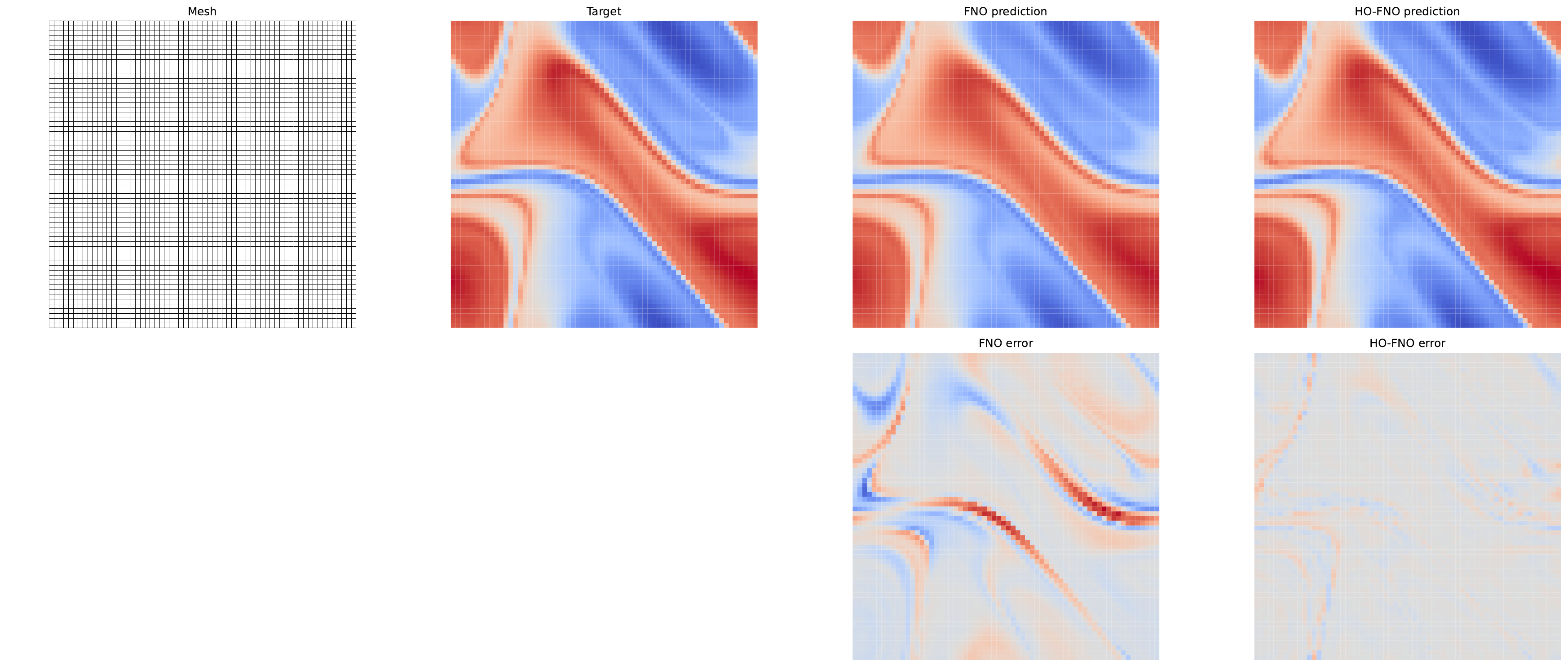}
    \caption{Qualitative visualization of predictions of FNO and HO-FNO (order 2) on the Navier--Stokes dataset.
    } 
    \vspace{-1em}
    \label{fig:backbones}
\end{figure}

\clearpage
\newpage

\end{document}